\documentclass[a4paper,twocolumn,11pt,accepted=2019-07-22]{quantumarticle}
\pdfoutput=1
\usepackage[utf8]{inputenc}
\usepackage[english]{babel}
\usepackage[T1]{fontenc}
\usepackage[backend=bibtex,style=numeric,sorting=none]{biblatex}
\bibliography{HEmaintextV2_plus_SupMatRefs}
\usepackage{amsmath,amsfonts,amsthm,amssymb}
\usepackage{setspace}
\usepackage{appendix,verbatim,bbm,bm}
\usepackage[normalem]{ulem}
\usepackage{amsbsy,dsfont,graphicx,float,epsfig,epstopdf,dsfont,xcolor,soul}
\usepackage{hyperref}
\usepackage{url}
\usepackage[figure,table]{hypcap}
\usepackage{enumerate}
\usepackage{lineno,titlesec}

\newtheorem{theorem}{Theorem}
\newtheorem{lemma}{Lemma}
\newtheorem{definition}{Definition}
\newcommand{\be}{\begin{equation}}
\newcommand{\ee}{\end{equation}}
\newcommand{\ba}{\begin{align}}
\newcommand{\ea}{\end{align}}
\newcommand{\tr}{\textup{tr}}
\newcommand{\la}{\langle}
\newcommand{\ra}{\rangle}
\newcommand{\proj}[1]{\ensuremath{| #1\rangle\!\langle #1 |} }

\newcommand{\suppl}{Supplementary Material}
\newcommand{\doc}{\text{manuscript}} 

\newcommand{\nn}{{\mathbbm{N}}}
\newcommand{\rr}{{\mathbbm{R}}}

	\newcommand{\id}{{\mathbbm{1}}}

\newcommand{\me}{\mathrm{e}}
\newcommand{\mi}{\mathrm{i}}

\newcommand{\bo}{{\Theta}}
\newcommand{\soo}{o}
\newcommand{\helmholtz}{non-equilibrium}

\hypersetup{
	bookmarksnumbered,
	pdfstartview={FitH},
	citecolor={darkgreen},
	linkcolor={darkred},
	urlcolor={darkblue},
	pdfpagemode={UseOutlines}}
\definecolor{darkgreen}{RGB}{50,190,50}
\definecolor{darkblue}{RGB}{0,0,190}
\definecolor{darkred}{RGB}{238,0,0}

\renewcommand{\d}{\dagger}

\newcommand{\ket}[1]{\ensuremath{\left\lvert#1\right\rangle}}

\newcommand{\ketbra}[2]{\ensuremath{\left| #1\right\rangle\!\!\left\langle #2\right|} }

\renewcommand{\Re}{\mathrm{Re}}
\renewcommand{\Im}{\mathrm{Im}}

\newcommand{\hot}{\textup{Hot}}
\newcommand{\cold}{\textup{Cold}}
\newcommand{\batt}{\textup{W}}
\newcommand{\battery}{\textup{W}}
\newcommand{\CW}{\textup{ColdW}}
\newcommand{\CMW}{\textup{ColdMW}}
\newcommand{\total}{\textup{ColdHotMW}}
\newcommand{\mach}{\textup{M}}

\newcommand{\sttotal}{\textup{ColdHotMW}}
\newcommand{\stcold}{\textup{Cold}}
\newcommand{\sthot}{\textup{Hot}}
\newcommand{\rhoin}{\rho_{\rm in}}
\newcommand{\rhoout}{\rho_{\rm out}}

\newcommand{\onecold}{\textup{c}}

\newcommand{\kappabar}{\bar{\kappa}}

\renewcommand{\emph}{\textit}

\newcommand\mpwS[1]{{\let\helpcmd\sout\parhelp#1\par\relax\relax} }
\long\def\parhelp#1\par#2\relax{%
	\helpcmd{#1}\ifx\relax#2\else\par\parhelp#2\relax\fi%
}

\newcommand{\myacknowledgments}{\begin{center}{\bf Acknowledgments}\end{center}\par}
\newcounter{lastnote}

\usepackage{appendix}
\usepackage{enumitem} 

\newtheorem{remark}{Remark}
\newcommand{\htwo}{\textup{h}_2}
\usepackage{soul}

\begin{document} 
	\title{The maximum efficiency of nano heat engines depends on more than temperature}

\author{Mischa P. Woods}
\affiliation{Institute for Theoretical Physics, ETH Zurich, Switzerland}
\affiliation{QuTech, Delft University of Technology, Lorentzweg 1, 2611 CJ Delft, Netherlands}
\author{Nelly Huei Ying Ng}
\affiliation{QuTech, Delft University of Technology, Lorentzweg 1, 2611 CJ Delft, Netherlands}
\affiliation{Dahlem Center for Complex Quantum Systems, Freie Universit{\"a}t Berlin, 14195 Berlin, Germany}
\author{Stephanie Wehner}
\affiliation{QuTech, Delft University of Technology, Lorentzweg 1, 2611 CJ Delft, Netherlands}

\begin{abstract}
		Sadi Carnot's theorem regarding the maximum efficiency of heat engines is considered to be 
		of fundamental importance in thermodynamics. This theorem famously states that the maximum efficiency depends only on the temperature of the heat baths used by the engine, but not on the specific structure of baths.
		Here, we show that when the heat baths are finite in size, and when the engine operates in the quantum nanoregime, a revision to this statement is required. We show that one may still achieve the Carnot efficiency, when certain conditions on the bath structure are satisfied; however if that is not the case, then the maximum achievable efficiency can reduce to a value which is strictly less than Carnot. We derive the maximum efficiency for the case when one of the baths is composed of qubits. Furthermore, we show that the maximum efficiency is determined by either the standard second law of thermodynamics, analogously to the macroscopic case, or by the non increase of the max relative entropy, which is a quantity previously associated with the single shot regime in many quantum protocols. This relative entropic quantity emerges as a consequence of additional constraints, called generalized free energies, that govern thermodynamical transitions in the nanoregime. Our findings imply that in order to maximize efficiency, further considerations in choosing bath Hamiltonians should be made, when explicitly constructing quantum heat engines in the future.
		This understanding of thermodynamics has implications for nanoscale engineering aiming to construct small thermal machines.
\end{abstract}
\maketitle 

\section{Introduction}
	Nicolas L{\'e}onard Sadi Carnot is often described as the ``father of thermodynamics''. 
	In his only publication in 1824 \cite{carnot}, Carnot gave the first successful theory in analysing
	the maximum efficiency of heat engines. 
	It was later used by Rudolf Clausius and Lord Kelvin to formalize the second law of 
	thermodynamics and define the concept of entropy~\cite{clausius_entropy,Kelvin2ndlaw}.
	In particular, Carnot studied heat engines where \textit{working fluids} undergo heating and cooling between two heat sources at different temperatures. In 1824, he concluded that the maximum efficiency attainable did not depend upon the 
	exact nature of the working fluids~\cite{carnot}:
		\begin{quote}
			\textit{The motive power of heat is independent of the agents employed to realize it; 
			its quantity is fixed solely by the temperatures of the bodies between which is effected,
			finally, the transfer of caloric.}
		\end{quote}
	For his ``motive power of heat'', we would today say ``the efficiency of a reversible heat engine'', 
	and ``transfer of caloric'' we would replace with ``reversible transfer of heat''. 
	Carnot knew intuitively that his engine would have maximum efficiency, 
	but was unable to state what that efficiency should be. 
He also defined a hypothetical heat engine (now known as the \emph{Carnot engine}) 
	which would achieve the maximum efficiency. 
	Later, this efficiency --- now known as the Carnot efficiency --- was shown to be
	\be 
		\eta_C=1-\frac{\beta_\hot}{\beta_\cold},
	\ee 
	where $\beta_\cold$, $\beta_\hot$ are the inverse temperatures of the cold and hot heat baths\footnote{Throughout this manuscript we set the Boltzmann's and Planck's constants, $k_{\rm B}$ and $\hbar$ to unity.}; as they are now more commonly referred to.

	Unlike the large scale heat engines that inspired thermodynamics, we are now able to build
	nanoscale quantum machines consisting of a mere handful of particles, and this has
	prompted many efforts to understand quantum thermodynamics~\cite{Ng2018,brandao2013resource,HO-limitations,2ndlaw,aaberg2013truly,dahlsten2011inadequacy,uzdin2016quantum,definingworkEisert,gemmer2015single,LJR2015description,scully2003extracting,QCTTST2015,QTgeneralprocesses,epsdetwork,tajima2014refined,gemmer:book,nanoscaleHE,universalityQCE2015,extworkcorr, alhambra2016fluctuatingstates, alhambra2016fluctuatingwork}. 
	In particular, given such nanoscale devices, one of the main issues addressed in quantum thermodynamics is the single-shot analysis of thermodynamical state transitions or work extraction. This approach addresses the scenario of moving away from the thermodynamic limit of infinitely many identical particles: one is increasingly interested in how a single copy of some small system (such as one to two atoms), in one instance of its evolution, may exhibit thermodynamical behaviour (due to interactions with its environment). These approaches are complementary to previous approaches of studying the ensemble or time averaged behaviour of the system. Many of the results in single-shot analysis have shown that the workings of thermodynamics become more intricate in such regimes
		~\cite{HO-limitations,2ndlaw,aaberg2013truly,dahlsten2011inadequacy}. While earlier efforts in quantum thermodynamics apply methods in statistical physics to average over time~\cite{spohn1978irreversible}, or particular models of open systems dynamics \cite{Mitchison2018,Batalho2018,SBLP2011smallest,BLPS2012virtual,MarMit,ent_enchance_cooling,Alicki1979,GK2015,Kosloff13,KL14}, single-shot thermodynamics adopts tools from quantum information theory to contribute to answering similar physical problems in a different light. 
		
		Following this approach of single-shot quantum thermodynamics, we show in this \doc~that unlike at the macroscopic scale --- where Carnot's fundamental results undoubtedly hold --- 
	there are new fundamental limitations to the maximal efficiency at the nanoscale. 
	Most significantly, we show that this maximum efficiency\,\footnote{We emphasize that by ``maximum efficiency'', it is understood that, for fixed hot and cold bath temperatures, we are maximising the efficiency over all possible heat engine cycles, in other words any machine that may interact with the different baths and undergo a cyclic process.} \emph{depends} on the heat baths. 
	In other words, we find that the Carnot efficiency can be achieved, but only when certain conditions on the bath  Hamiltonian are satisfied. 
	Otherwise, a reduced efficiency is obtained, highlighting the significant difference in the performance of heat engines in the single-shot regime. 
	
	This \doc~is organized as follows: in Section \ref{sec:setup main}, we first introduce the setup, and clearly detail all assumptions made about the heat engine model of our study. Next in Section \ref{sec:work nano regime main} we introduce different notions of work in the nano regime, which will be important for understanding our results. Next, in Section \ref{sec:results main} we detail our findings. We start by showing that although a positive amount of ``perfect work'' cannot be extracted, some amount of ``near perfect'' work is possible, and we derive the efficiency that can be achieved while extracting this type of work. Finally we conclude with a summary of our results and open questions in Section \ref{sec:conclusion}.

\section{Setup}\label{sec:setup main}
\subsection{The heat engine model}

A heat engine (see Fig.~\ref{fig:heatEngine}) is a procedure for extracting work from a temperature difference between two systems. It comprises of four basic elements: the two thermal baths at distinct inverse temperatures $\beta_\hot$ and $\beta_\cold$, a machine, and a system to which work is extracted, often referred to as a battery.
The machine interacts with these baths in such a way that utilizes the temperature difference between the two baths to perform work extraction. 
The battery is a particularly useful way of quantifying extracted work in this model as it allows for the transfer/storage of energy into the battery ancillary system, while the machine returns to its original state.
Different battery models such as the work qubit or qubits~\cite{HO-limitations,2ndlaw, uzdin2016quantum, bylicka2016thermodynamic},  
 the weight \cite{skrzypczyk2014work,BLPS2012virtual} and the purity battery \cite{del2011thermodynamic,qlandauer} have been recently studied and used to quantify work. Although the concept resembles the notion of a work reservoir such as in \cite{deffner2013information}, in these recent works, the problem of extracting work is cast in a strongly operational perspective: one is not only interested in increasing the average energy of the battery, but is also interested in the final state such a battery takes, so that it may be used in the future to enable other processes \cite{campaioli2017enhancing}.
 
 \begin{figure}[t]
 	\begin{center}		\includegraphics[scale=0.33]{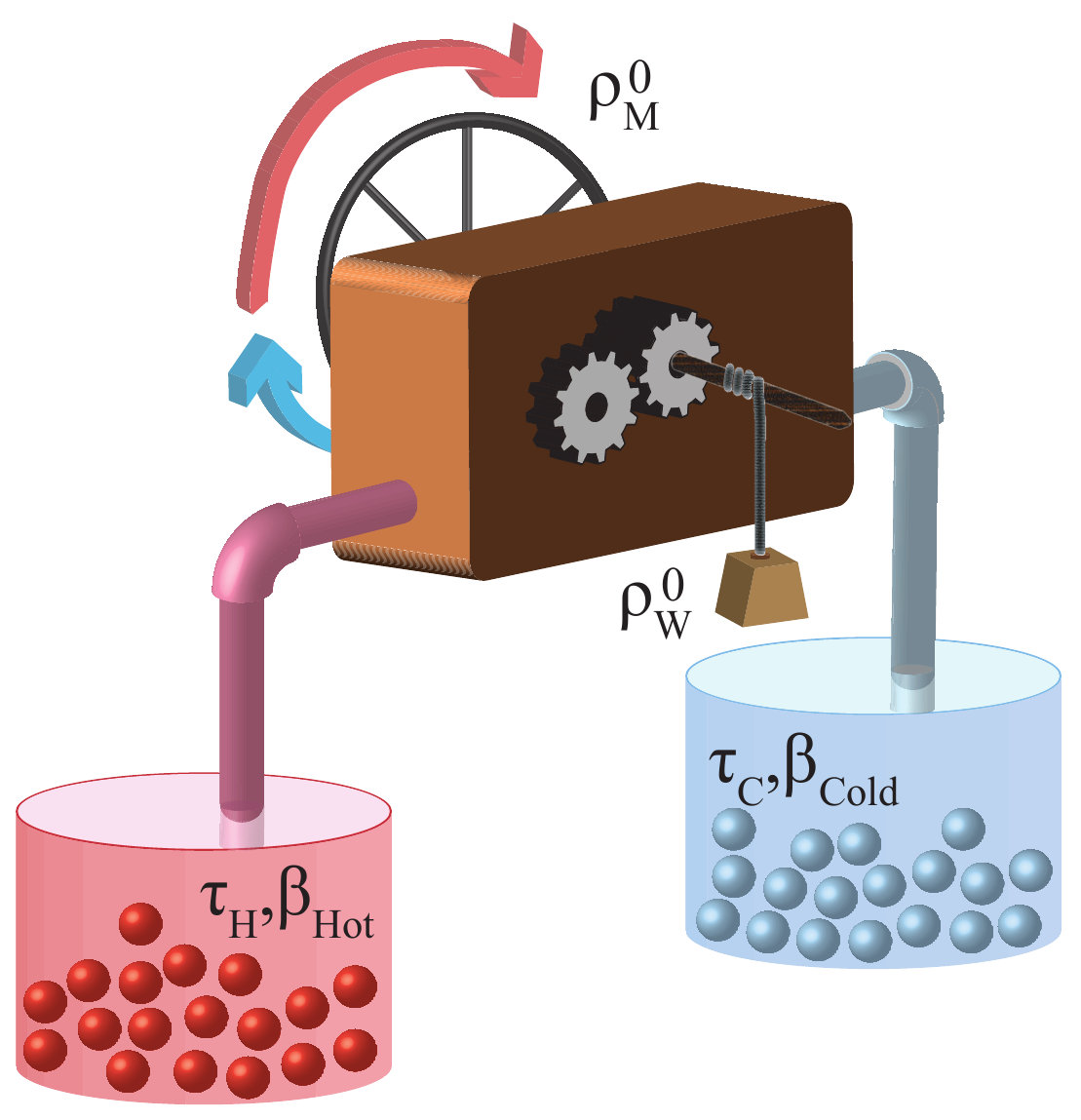}
 		\caption{Illustration of a heat engine that contains four main components: two baths at different temperatures, a machine and a system we call a battery.
 		}\label{fig:heatEngine}
 	\end{center}
 \end{figure}
In this section, we describe a general heat engine setup, where all involved systems and changes in energy are accounted for explicitly. Let us begin with the total Hamiltonian 
\begin{equation}\label{eq:totalH_main}
\hat H_\textup{tot} = \hat H_\cold +\hat H_\hot + \hat H_\mach +\hat H_\batt,
\end{equation}
where the indices Hot, Cold, M, W represent a hot thermal bath (Hot), a cold thermal bath (Cold), a machine (M), and a battery (W) respectively. 
We adopt a resource theoretic approach, which allows all energy-preserving unitaries $U(t)$ on the global system, i.e. all unitaries which obey 
\be\label{eq:global u main text}
[U(t),\hat H_\textup{tot}]=0. 
\ee 

Note that this global, time-independent Hamiltonian, is non-interacting; which is an indispensable feature of the resource theory framework~\cite{HO-limitations, brandao2013resource}. This approach allows us to mark a clear separation of the subsystems, and thus justifies the use of Gibbs states as thermal reservoirs. This is because for any given Hamiltonian, the Gibbs states uniquely satisfy complete passitivity \cite{2ndlaw,pusz1978passive}, which forbids the extraction of average energy via unitary operations on the state. Nevertheless, interactions of arbitrary strength between all the different systems are still allowed, under the only condition of Eq. ~\eqref{eq:global u main text}. 
Indeed, $U(t)$ can be of the form
\be 
U(t)=\me^{\mi t (\hat H_\textup{tot}+\hat I_{\cold\hot\mach\batt})},
\ee
with the norm of $\hat I_{\cold\hot\mach\batt}$ being arbitrarily large, so long as $ [\hat I_{\cold\hot\mach\batt},\hat H_\textup{tot}]=0 $ in order to preserve energy. In order to implement the unitary over the heat engine, one may use an auxiliary system. This is an aspect of the model which is common to all resource theory approaches to quantum thermodynamics. The optimal way to achieve this remains an open question. See \cite{woods2016autonomous} for partial results and Section V in \cite{Ng2018} for a more in-depth discussion.

The initial state of our heat engine will be of the form
\begin{align}\label{eq:rho0_main}
\rho^{0}_\total=\tau_\cold^{0}\otimes \tau_\hot^{0}\otimes \rho_\mach^{0}\otimes \rho_\battery^{0}.
\end{align}
The state $\tau_\hot^{0}$ ($\tau_\cold^{0}$) is the initial thermal state at inverse temperature $\beta_\hot$ ($\beta_\cold$), corresponding to the hot (cold) bath Hamiltonians $\hat H_\hot, \hat H_\cold$, with $\beta_\cold> \beta_\hot$. More generally, given any Hamiltonian $\hat H$ and inverse temperature $\beta$, the thermal state is defined as $\tau=\frac{1}{\tr (e^{-\beta\hat H})}	e^{-\beta\hat H}$. 
The initial machine $(\rho_\mach^0, \hat H_\mach)$ can be chosen arbitrarily, as long as its final state is preserved, and therefore the machine acts like a catalyst. Lastly, the initial battery state in our setup $\rho_\battery^0$ is any energy eigenstate of the battery --- see Section \ref{sec:work nano regime main} for a further description of the battery model used in this manuscript.

Heat engines in practice tend to have their hot and cold thermal baths of different {relative} sizes. Power stations near the ocean are good examples of this. 
Here we allow our baths to play a similar role. The hot bath may be arbitrarily large, and acts like a reservoir, while the cold bath is of some fixed finite size. A reversal of the hot and cold bath sizes would be possible within our framework but unnecessary to reach our conclusions.

One cycle of the heat engine process produces a final \emph{reduced} state~\footnote{For any bipartite state $\rho_{\rm AB}$, we use the notation of reduced states $\rho_{\rm A}:=\tr_\textup{B}(\rho_{\rm AB})$, $\rho_{\rm B}:=\tr_\textup{A}(\rho_{\rm BA})$.} 
\begin{align}\label{eq:rho1 corrs}
	\rho_\CMW^1 = \tr_\hot \left[U(t)\rho^{0}_\total U(t)^\dagger\right],  
	\end{align}
where the machine should be preserved, i.e. $\rho_\mach^{1}=\rho_\mach^{0}$, and $\rho_\cold^1,$ $\rho_\battery^{1}$ are the final local states of the cold bath and battery. Note that we allow for arbitrary correlations to exist in the final state $\rho_\CMW^1$, as long as the reduced state of the machine is preserved. Our central quantity of interest is the maximum efficiency of work extraction; and whether it can be as large as the Carnot efficiency. One may argue that allowing for arbitrary final correlations (quantum or classical) is not physically well motivated; since the correlations might degrade the functionality of the machine and potentially reduce the efficiency of the heat engine in subsequent cycles. Fortunately, this debate can be avoided, since we also show that if the heat engine setup is such that Carnot efficiency is achievable, then the final state of the heat engine must be of product form:
\begin{align}\label{eq:rho1_main}
\rho_\CMW^1  =\rho_\cold^{1}\otimes  \rho_\mach^{1}\otimes\rho^1_\battery.\qquad\qquad
\end{align}
The fact that correlations in the final state never allow, in all cases, for a larger amount of work nor an increased efficiency is proven in Section \ref{Correlated final state of battery} 
 of the~\suppl. However, one can intuitively see why it is so, from the fact that the \helmholtz~
 free energy~\cite{wilming,Esposito_2010,parrondo2015thermodynamics} is super-additive. More precisely, for any system of state $\rho$ in thermal contact with a bath at inverse temperature $\beta$, consider the \helmholtz~
 free energy\,\footnote{For systems in thermodynamics equilibrium, this free energy is also known as the Helmholtz free energy. For simplicity, we will also refer to this quantity as \emph{standard} free energy in the text.}
	\begin{equation}\label{eq:HFE}
	F(\rho):= \tr(\hat H\rho) - \beta^{-1} S(\rho),
	\end{equation} 
	where $S(\rho)=-\tr(\rho\ln\rho),$ is the von Neumann entropy. 
	For our setup, super-additivity is then the statement that for any bipartite state $\rho_{AB}$, $ F(\rho_{AB}) \geq F(\rho_A)+F(\rho_B) $, where equality is achieved if and only if $ A $ and $ B $ are uncorrelated. This implies that creating a final state with correlations is at least as hard as creating one without correlations, and thus allowing the final components of the heat engine to become correlated, cannot help one to achieve the Carnot efficiency; which may be achieved only in the limit when $F(\cdot)$ is invariant under one cycle of the heat engine (i.e. the limit in which the heat engine is macroscopically reversible)\footnote{As a side remark for readers familiar with \cite{Muller2017}, the reason why correlating catalysts do not boost efficiency in our set-up, is because the dimension of the catalyst diverges in the limit approaching the Carnot efficiency. Furthermore, the techniques developed in this manuscript might potentially pave way to solving an open problem in \cite{Muller2017}, as pointed out by the authors.}. 
As such, since a heat engine process needs to satisfy the non-increase of the free energy, it turns out that without loss of generality we can assume $ \rho_\CMW^1 $ to be of the form given in Eq.~\eqref{eq:rho1_main} when investigating the achievability of Carnot efficiency in our setup. We therefore base the rest of our analysis in the main text on Eq.~\eqref{eq:rho1_main}. 

Since $(\tau_\hot^0, \hat H_\hot)$ and $(\rho_\mach^0, \hat H_\mach)$ can be arbitrarily chosen and since Eq. \eqref{eq:rho1_main} is assumed, 
	the setup now corresponds to the set of \emph{catalytic thermal operations} \cite{2ndlaw} one can perform on the joint state $\CW$. This implies that the cold bath is used as a resource state.
By catalytic thermal operations that act on the cold bath, using the hot bath as a thermal reservoir, and the machine as a catalyst, one can possibly extract work and store it in the battery.
In the next section, we see how work is defined and categorized according to initial and final states of the battery $ \rho_\batt^0 $ and $ \rho_\batt^1 $.
As for now, to summarize, the following assumptions are made in our heat engine setup:
\begin{enumerate}
	\item The initial global state is a product state between all the systems, as shown in Eq.~\eqref{eq:rho0_main}. 
	\item
	Cyclicity of the machine, i.e. system $ M $ undergoes a cyclic process:  $\rho_\mach^0=\rho_\mach^{1}.$ 
	\item
	The heat engine as a whole is isolated from and does not interact with the world, i.e. $[U(t),\hat H_\textup{tot}]=0$. 
	 This assumption ensures that all possible resources in a work extraction process has been accounted for.
	\item The Hilbert space associated with $\rho_\sttotal^0$ is finite dimensional but can be arbitrarily large. 
\end{enumerate}

\subsection{Work in the nanoregime}\label{sec:work nano regime main}

	The definition of work when dealing with nanoscopic quantum systems has seen much attention lately
	\cite{HO-limitations, 2ndlaw, aaberg2013truly, dahlsten2011inadequacy, definingworkEisert, gemmer2015single}. 
	Performing work is always understood as changing the energy of a system, which in this \doc~is called \emph{battery}.
	In the macroregime, one often pictures raising a weight on a string. 
	In the nanoregime, this corresponds to changing the energy of a quantum system by pumping it to an excited state (see Fig.~\ref{fig:battery}).
	In particular, a minimalistic battery model can be demonstrated as a two-level system~\cite{2ndlaw}. 
	Performing work corresponds to bringing the state from its ground state to the excited state, 
	where the energy gap is fine-tuned to the amount of work $W_\textup{ext}$ to be done. 
	
	While an arbitrary energy spacing is difficult to realize in a two-level system, 
	it can be done by picking two levels with the desired spacing from a \emph{quasi-continuum} battery: 
	this battery comprises of a large but finite number of discrete levels which form a quasi-continuum. 
	Such a battery closely resembles the classical notion of a ``weight attached to a string'' as considered in~\cite{skrzypczyk2014work}. 
	The battery can be charged by bringing it from a particular state (e.g. the ground state) to any of the higher energy levels. 
	In this paper, we adopt the use of such a quasi-continuum battery model.
	This battery $ W $ has a Hamiltonian (written in its diagonal form)
	\begin{align}
	\hat H_\battery:=\sum_{i=1}^{n_\battery} E^\battery_i|E_i\ra \la E_i|_\battery,
	\end{align}
	where $ \lbrace E_i^\batt \rbrace_{i=1}^{n_\batt} $ is a set which can be arbitrarily large, but of fixed cardinality; while its elements $E_i^\battery\in\rr$ may or may not be uniformly bounded.
	
	\begin{figure}[h!]
		\begin{center}
			\includegraphics[scale=0.4]{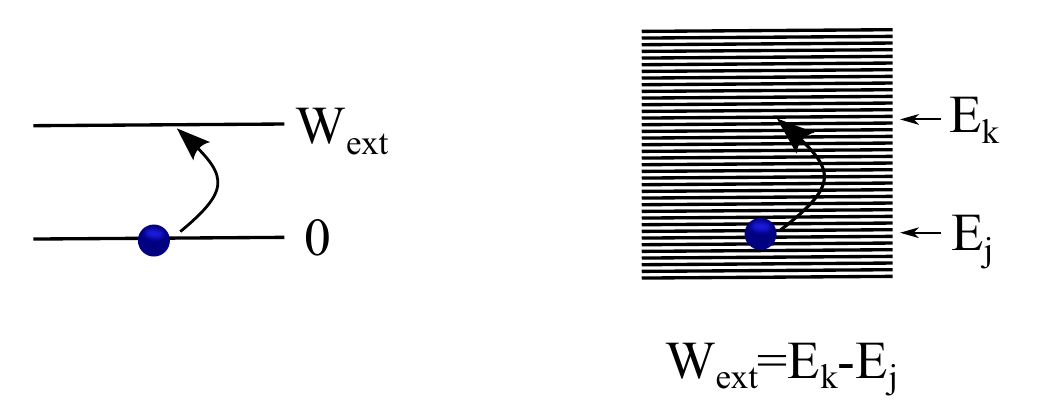}
			\caption{
				\textbf{The battery} models a work-storage component of a heat engine. In our setup, an adaptive quasi-continuum of energy levels is used.
			}\label{fig:battery}
		\end{center}
	\end{figure}

	One aspect of extracting work $W$ is to bring the battery's initial state $\rho_\textup{W}^0$ to some final state $\rho_\textup{W}^1$ such that 
	$W = \tr(\rho_\textup{W}^1 \hat{H}_\textup{W}) - \tr(\rho_\textup{W}^0 \hat{H}_\textup{W}) > 0$. 
	However, a change in energy alone, does not yet correspond to performing work. 
	It is implicit in our macroscopic understanding of work that the energy transfer takes place in an ordered form. 
	When lifting a weight, we \emph{know} its final position and can exploit this precise knowledge to transfer all the work onto a third system without --- in principle --- losing any energy in the process. 
	
	In the quantum regime, such knowledge corresponds to $\rho_\textup{W}^1$ being a pure state. 
	When $\rho_\textup{W}^1$ is diagonal in the energy eigenbasis of $\hat{H}_\textup{W}$, then $\rho_\textup{W}^1$ is an energy eigenstate.
	We can thus understand work as an energy transfer about which we have perfect information, while heat, in contrast, is an energy transfer about which we hold essentially no information, other than average energy increase. 
	Clearly, there is also an intermediary regime in which we transfer energy, while having some, but imperfect information.
	
	To illustrate this, consider the quasi-continuum battery described above, and starting out from an arbitrary initial energy eigenstate $ \rho_\battery^{0}=\proj{E_j}_\battery
	$. Ideally, we want to extract work and store it in the battery, by inducing a transition to another energy eigenstate $\proj{E_k}_\batt$, where $E^\batt_k>E^\batt_j$. 
	Let $ \varepsilon $ denote the failure probability of our doing so, and
	\be\label{W general_main}
	W_\textup{ext} = E^\batt_k-E^\batt_j>0
	\ee
	the work extracted. In the case where $W_{\rm ext}$ is a value such that the transition $\rho_\CW^0\rightarrow\rho_\CW^1$ is possible via catalytic thermal operations, it corresponds to extracting work.
	 We first define the case in which we always succeed:
	\begin{definition}\label{def: main text: perferct work}
		The work extracted $W_\textup{ext}$, is called \textup{perfect work} when $ \rho_\batt^1 =  \proj{E_k}_\batt$, in other words, $ \varepsilon = 0 $.
	\end{definition}
In general however, there is a non-zero failure probability of work extraction. This causes us to lose some information about the battery state:
	 rather than the final state of the battery being $\proj{E_k}_\batt$, it could be \textit{any} state which is a distance $ \varepsilon $ away, namely
	\be \label{eq:trace dist ep bat main text}
	d(\rho_\batt^1,\proj{E_k}_\batt)=\varepsilon,
	\ee   
	where we use $ d(\rho,\sigma) $ to denote the trace distance of two states $ \rho $ and $ \sigma $ (see Eq. \eqref{eq:trace dist def} for a definition). 
	The usage of trace distance here is motivated by its strong operational meaning: the smaller the trace distance between two states, the harder it is to distinguish them using any quantum measurements. We allow for any value of
	\be \label{eq:uniformity condition on varepsilon}
	\varepsilon\in [0,l], \textup{ for any fixed $l<1$},
	\ee	
 and the smaller $\varepsilon$ is, the closer we are to the situation of perfect work\,\footnote{In Eq. \eqref{eq:uniformity condition on varepsilon}, the constraint of a fixed $l <1$, is simply to rule out the physically irrelevant limit $\varepsilon\rightarrow 1$,  where the final state has no overlap at all with $\proj{E_\batt^k}$. We could choose $l=0.999$ for example.}. Note that any energy incoherent state $\rho_W^1$ such that $d(\rho_W^1,\proj{E_k}) = \varepsilon$ can be written in the form
 	\begin{equation}\label{eq:rhow1revised}
 	\rho_W^1 = (1-\varepsilon)\proj{E_k} + \varepsilon \rho_{\rm junk},
 	\end{equation} 
 	where $ \rho_{\rm junk} $ is any density matrix of the battery that does not have support on $ \proj{E_k} $. 
	The parameter $W_{\rm ext}>0$ is defined as an energy difference in Eq. \eqref{W general_main}.
	Our goal is to maximize the achievable value of $W_\textup{ext}$, while allowing for all final battery states of the form of Eq.~\eqref{eq:rhow1revised}, such that there exists a catalytic thermal operation enabling the transition $\rho_\CW^0\rightarrow\rho_\CW^1$.

	In this \doc, we propose a characterization of the quality of extracted energy by the entropy difference 
	\be\label{eq:Delta S def main text}
	 \Delta S := S(\rho_\batt^1)-S(\rho_\batt^0),
	\ee 
	where $S(\rho)=-\tr(\rho\ln\rho),$ is the von Neumann entropy. 
	For perfect work, $\Delta S = 0$. Another type of example one sometimes comes across, is when $ \rho_\batt^1 $ becomes a thermal state \cite{BLPS2012virtual, ent_enchance_cooling}. In such cases, although the average energy of the battery may still increase, the corresponding entropy increase is maximal~\cite{jaynes1957information}.  
	Moreover, when $\Delta S > 0$, what is relevant is not its value as an absolute, but relative to the energy $W_\textup{ext}$ that is extracted. 
	We are thus interested in the quantity $\Delta S/W_\textup{ext}$,
	in particular the limit $\Delta S/W_\textup{ext} \rightarrow 0$ \footnote{Note that $\Delta S$ and $W_\textup{ext}$ have different units. If one prefers to work with a unitless measure, one can instead work with $c\cdot \Delta S/W_\textup{ext}$ for any constant $c$ with units of inverse temperature w.l.o.g., since the limit  $c\cdot \Delta S/W_\textup{ext}\rightarrow 0$ holds iff $\Delta S/W_\textup{ext}\rightarrow 0$ holds.}, which motivates the following definition:
	
		\begin{definition}(Near perfect work)\label{def:near perfect work main text}
		We say that a partially ordered set of heat engine protocols characterized by the extracted amount of work and and corresponding failure parameters $ \mathcal{S}^\textup{H.E.}=\lbrace (W_{\rm ext},\varepsilon)\rbrace $, leads to \emph{near perfect work extraction} if
		\begin{itemize}
			\item[1)] For all values of $\varepsilon$, $0< \varepsilon\leq l,~$ for some fixed $l<1$ and
			\item[2)] 
			For all $1>p>0$, there exists a non-trivial subset of protocols $ \mathcal{S}^\textup{H.E.}_p \subset \mathcal{S}^\textup{H.E.} $ such that when $(W_{\rm ext},\varepsilon)\in\mathcal{S}_p$\,, then \be \frac{\Delta S}{W_\textup{ext}}<p.\ee
\end{itemize}
	\end{definition} 

We shall see in Section \ref{sec:results main} that the maximum efficiency of a heat engine cycle can only be achieved in the limit where $ W_\textup{ext}\rightarrow 0 $. 
In such cases, condition 2) in Def.~\ref{def:near perfect work main text} also corresponds to the limit $\Delta S/W_\textup{ext} \rightarrow 0$. In Section \ref{section:The setting} 
 of the~\suppl, we show that the restriction of near perfect work can be re-cast in terms of an equivalent condition involving the probability of failure $\varepsilon$, namely: the two conditions of Definition \ref{def:near perfect work main text} above  
are satisfied iff 
\be\label{eq:alternative equiv to n.p.w. def main text}
\displaystyle\lim_{\varepsilon\rightarrow 0^+} \frac{\Delta S}{W_\textup{ext}} =0.
\ee

Perfect and near perfect work are the two types of energy investigated in this paper, due to their strong operational significance in capturing both the essence of energy increase and knowledge about the final battery state. Other types of energy increase, which we refer to as \textit{imperfect work}, are studied in a separate paper, due to the large qualitative differences in the extracted energy~\cite{paper2}. For example, it is worth noting that for both perfect and near perfect work, one may recover Carnot's results about the efficiency of heat engines by invoking the non-increase of free energy 
(see Section \ref{sec:Comparison with standard entropy results}); this is not the case for imperfect work, where one can surpass the Carnot efficiency, since $ \Delta S $ is non-negligible compared to $ W_\textup{ext} $; thus heat contributions are not separated from the extracted energy. From another perspective, this means that imperfect work is a more debatable way of quantifying energy as work, nevertheless, this is addressed in \cite{paper2} due to the extensive literature on quantum thermodynamics that uses mean energy increase as a quantifier of work. 

\subsection{Definition of efficiency and maximum efficiency}\label{Definition of efficiency and maximum efficiency}

The efficiency of a heat engine is defined as 
\be\label{eq:efficiency main text}
\eta:=\frac{W_{\rm ext}}{\Delta H},
\ee
where $\Delta H$ is the amount of heat drawn from the hot bath, namely $\Delta H= \tr(\hat H_\hot \rho_\sthot^0)-\tr(\hat H_\hot \rho_\sthot^1)$. Since the machine's final and initial states are the same after one cycle, and the initial state of the cold bath is fixed; due to total mean energy conservation, $\Delta H$ can be expressed solely as a function of $\rho_\cold^1$ and $W_\textup{ext}$. From Eq. \eqref{eq:efficiency main text}, we can define the maximum achievable efficiency in the nanoregime $\eta^\textup{nano}$ as a function of the final state of the cold bath $\rho_\cold^1$. More precisely,
\begin{align}
&\eta^\textup{nano}
 (\rho_\cold^1)\label{eq:max nano as function fo cold bath_main}
:=\sup_{W_\textup{ext}>0} \eta(\rho_\cold^1) \quad \textup{ subject to }&\\
&F_\alpha(\rho^0_\battery\otimes\tau_\cold^0
)\geq F_\alpha(\rho^1_\battery\otimes\rho_\stcold^1
)\quad \forall \alpha\geq 0,&\label{eq:second laws line}
\end{align}
where $F_\alpha$ are the \textit{generalized free energies}
(see Eq.~\eqref{eq:generalfreeenergyMT} in Section \ref{sub:proof_overview} for definition). Note that by fulfilling Eq.~\eqref{eq:second laws line}, we are already maximizing over all possible heat engine cycles for a given value of $ W_\textup{ext} $. Eq.~\eqref{eq:max nano as function fo cold bath_main} further maximizes efficiency over achievable values of $ W_{\rm ext} $. Recall that $ \beta_\cold,\beta_\hot $ are fixed, and $ \rho_{\rm W}^0 $ is an energy eigenstate. In the resource theoretic framework for thermodynamics, the generalized free energies provide conditions for the possibility of catalytic thermal operations between the initial and final states to take place. In other words, Eq.~\eqref{eq:second laws line} is an application of the so-called \textit{generalized second laws} of quantum thermodynamics, to our heat engine setup; as shown in detail in Section \ref{sub:proof_overview}. They can be seen as a generalization of the second law of thermodynamics (see Section \ref{sec:Comparison with standard entropy results}), in single-shot quantum thermodynamics. 
From Eq.~\eqref{eq:generalfreeenergyMT}, we see that the constraint of Eq. \eqref{eq:second laws line} is equivalent to $D_\alpha(\rho^0_\battery\otimes\tau_\cold^0\|\tau_\CW^h
)\geq D_\alpha(\rho^1_\battery\otimes\rho_\stcold^1\|\tau_\CW^h
)$ for all $\alpha\geq 0$, where $\tau_\CW^h$ is the Gibbs state of the cold bath and battery at inverse temperature $\beta_h$, and $D_\alpha$ are called \textit{$ \alpha $-R\'enyi divergences}. 
The family of $ \alpha $-R\'enyi divergences (which the max relative entropy discussed later is a member of) have been a powerful tool in single-shot quantum information theory, when it comes to tasks such as randomness extraction \cite{renner2004smooth,konig2009operational,tomamichel2011leftover}, source coding \cite{renner2004smooth,wang2009simple}, or hypothesis testing \cite{csiszar1995generalized,shayevitz2011renyi} for finite number of trials.

 With Eqns.~\eqref{eq:max nano as function fo cold bath_main} and \eqref{eq:second laws line} at hand, one can now define the maximum efficiency across all final states in the cold bath Hilbert space  $\mathcal{S}(\mathcal{H}_\cold)$, when demanding near perfect work. Specifically,
\be \label{eq:max eff general main text}
\eta_\textup{max}=\sup_{\rho_\cold^1\in\mathcal{S}(\mathcal{H}_\cold)\\ }\eta^\textup{nano}(\rho_\cold^1),
\ee
where the supremum is also over all partially ordered heat engine protocols corresponding to near perfect work.
In our analysis, a particular notion of efficiency emerges as the quantity of interest, which we refer to as the \textit{quasi-static} limit.  This corresponds to the maximum efficiency when the final state of the cold bath is thermal and its temperature only increases by an infinitesimal amount, namely
\be \label{eq:max eff general stat main text}
\eta_\textup{max}^\textup{stat}=\lim_{g\rightarrow 0^+}\eta^\textup{nano}(\tau(g))
\ee
where $\tau(g)$ is the Gibbs state on $\mathcal{S}(\mathcal{H}_\cold)$ at inverse temperature $\beta_f=\beta_\cold-g$. The reason why this limit emerges as a relevant scenario is as follows: first of all, we show that $ \eta^{\rm nano} (\rho_\cold^1) $ depends on two quantities, namely the average energy change in the cold bath (denoted as $ \Delta C $ throughout this \doc), and the extractable work $ W_{\rm ext} $. Given any fixed amount of $ \Delta C >0$, maximizing the efficiency over $\rho_\cold^1 $ corresponds to further maximizing $ W_{\rm ext} $ according to the $ F_1 $ constraint in Eq.~\eqref{eq:second laws line}. Moreover, $ F_1 $ is precisely the non-equilibrium free energy in Eq.~\eqref{eq:HFE}. We show that this maximum $ W_{\rm ext} $ occurs precisely when $ \rho_\cold^1 = \tau(g) $ for the particular value of $ g $ that corresponds to the fixed $ \Delta C $. Furthermore, we also prove that the efficiency is monotonically decreasing with this parameter $ g $, which means that the maximum efficiency occurs at the quasi-static limit described in Eq.~\eqref{eq:max eff general main text}. Indeed, if one evaluates the efficiency taking the quasi-static limit when assuming that only the condition on $ F_1 $ needs to be satisfied, one obtains the Carnot efficiency $ \eta_C $ in the limit $ g\rightarrow 0^+ $.
Since this is a \textit{necessary} condition for the possibility of a heat engine process, we will thus frequently work in the quasi-static limit in the rest of this \doc.

\section{Main results}\label{sec:results main}


\subsection{No perfect work}\label{No perfect work}

		Before establishing our main result, we first show that in the nanoscopic regime, 
		no heat engine can output perfect work (Def. \ref{def: main text: perferct work}). 
		That is, the efficiency of any such heat engine,
		\be 
		\sup_{ \rho_\cold^1\in\mathcal{S}(\mathcal{H}_\cold)}	\eta^\textup{nano}(\rho_\cold^1)= 0.
		\ee 
		In other words, there exists no global energy preserving unitary $U(t)$ obeying Eq. \eqref{eq:global u main text} for which $W_\textup{ext} > 0$ can be achieved. 
		The proof of this statement can be found in Section \ref{subsec:impossible} 
		 of the \suppl. In fact, it is interesting to note that the impossibility of drawing perfect work is a direct consequence of needing to satisfy one instance of the generalized second laws, in particular $ F_\alpha(\rho) $ in Eq. \eqref{eq:second laws line} when $ \alpha=0 $.
		
		While this might appear puzzling at first glance, it has a very nice analog in information theory; namely zero-error data compression. The scenario is as follows: suppose one desires to send a message across a particular channel. Depending on the redundancy of your data, you might not need to send the full file over: you can send a compressed version of the data, that only sends a fraction of symbols. But if your data is distributed over symbols with respect to a probability distribution of full rank, then theoretically you cannot perform compression with precisely zero error --- compression may be possible, however, if very small errors are allowed.
		
		

\subsection{Obtainable efficiency}

		Clearly, however, it is unreasonable to say that no heat engine could work at all in the quantum nanoregime, prompting the question how this might be possible.
		We show that for \emph{any} $\varepsilon > 0$, there exists a heat engine such that $W_\textup{ext} > 0$ can be achieved.
		Therefore, a heat engine is possible if we ask only for near perfect work. 
		Interestingly, even in the macroscopic regime, we can envision a heat engine that only extracts work with probability $1-\varepsilon$, but over many cycles of the engine we do not notice this feature when looking at the average work gained in each run.

		To study the efficiency in the nanoscale regime we make crucial use of the second laws of quantum thermodynamics~\cite{2ndlaw}.
		It is apparent from these laws that we might only discover further limitations to the efficiency than we see at the macroscopic scale. 
		Indeed they do arise, as we find that the efficiency no longer depends on just the temperatures of the heat baths. 
		Instead, the explicit structure of the cold bath Hamiltonian $\hat{H}_{\cold}$ becomes important (a similar argument can be made for the hot bath) --- even when choosing the optimal machine.

		We conduct the full analysis of the efficiency according to the second laws of thermodynamics considering a cold bath comprised of $n$ non-interacting two-level systems (qubits) each with its own energy gap  $\bar E_k$,
		\begin{equation}
			\hat H_\cold=\sum_{k=1}^n \id^{\otimes (k-1)}\otimes \bar E_k|\bar E_k\ra\la \bar E_k|\otimes\id^{\otimes (n-k)},\label{eq:omega for eqch qubit}
		\end{equation} 
		where $n$ can be arbitrarily large, but finite. 
		Let us denote the spectral gap of the cold bath --- the energy gap between its ground state and first excited state --- by $E_\textup{min}$. 
		We can then define the quantity
		\be\label{eq:eq23inmaintext}
			\Omega=\frac{E_\textup{min}(\beta_\cold-\beta_\hot)}{1+e^{-\beta_\cold E_\textup{min}}},
		\ee 
		
		Whenever $\Omega \leq 1$, consider all qubits on sites $k$ for which
		\be\label{eq:bar E and others condition}
		\frac{\bar E_k(\beta_\cold-\beta_\hot)}{1+e^{-\beta_\cold \bar E_k}}\leq 1,
		\ee
		holds. Since $\Omega \leq 1$ in this case, there will be a non-trivial subset $ \mathcal{C} $ of the cold bath qubits (at least one qubit) where Eq. \eqref{eq:bar E and others condition} holds. We show that the quasi-static efficiency $\eta_\textup{max}^\textup{stat}$ (for which the cold bath is taken over $ \mathcal{C} $) is indeed the familiar Carnot efficiency, which can be expressed as
		\be \label{eq:carnotOmegaleq1}
			\eta_\textup{max}=\eta_\textup{max}^\textup{stat}=\left(1+\frac{\beta_\hot}{\beta_\cold-\beta_\hot}\right)^{-1}.
		\ee
	Note that this is true for any $ n $ number of qubits, in particular also when $ n=1 $, which remarkably tells us that even when the cold bath consists of only a single qubit, Carnot efficiency can still be achieved when $ \Omega\leq 1 $ is satisfied. Intuitively, although the constrained optimization of Eq. \eqref{eq:max nano as function fo cold bath_main} looks complicated in general, nevertheless only the constraint imposed by the $\alpha=1$ second law in Eq. \eqref{eq:second laws line} is the most stringent one, and the optimal transition is when $F(\rho_\CW^0) = F(\rho_\CW^1)$. The other constraints in Eq. \eqref{eq:second laws line} are trivially satisfied with an inequality. Therefore, the second laws give effectively the same constraint as the usual second law.

		However, when $\Omega>1$, we find a new nanoscale limitation. In this situation, the efficiency for near perfect work is only
		\be \label{eq:subCarnotOmega>1}
			\eta_\textup{max}^\textup{stat}=\left(1+\frac{\beta_\hot}{\beta_\cold-\beta_\hot} \Omega \right)^{-1}
		\ee
for a quasi-static heat engine. Furthermore, note that for the case of a single qubit, all energy-incoherent states are thermal states with a particular temperature, and therefore the quasi-static limit is the only possible parametrization for the limit $ \rho_\cold^1\rightarrow\tau_\cold^0 $. This means that for a single-qubit cold bath, if $ \Omega >1 $ holds, then Eq.~\eqref{eq:subCarnotOmega>1} gives the maximum achievable efficiency, which will be strictly less than $ \eta_C $.

		Eq. \eqref{eq:subCarnotOmega>1} marks a limitation at the nano/quantum scale. This limitation occurs because in the constrained optimization of Eq. \eqref{eq:second laws line}, unlike when $\Omega\leq 1$, the $\alpha=\infty$ second law poses the strongest constraint (even stronger than the constraint of the \helmholtz~free energy), and therefore becomes solely relevant in dictating the state transition. In particular, the $\alpha=1$ second law can only be satisfied with a strict inequality when $\Omega>1$.
		The $ \alpha $-R\'enyi divergence $D_\infty$, corresponding to the $\alpha=\infty$ second law in Eq. \eqref{eq:second laws line}, is a well-known quantity in single-shot information theory called the \textit{max relative entropy}, denoted by $D_\textup{max}$~\cite{datta2009min,brandao2011one}. 
			 As such, in this regime, we find that instead of the \helmholtz~free energy, it is the max relative entropy that determines the efficiency of the heat engine cycle.
			The emergence of the max relative entropy as the deciding factor in obtainable efficiency is a signature of single-shot effects coming into play. 
			In the quantum thermodynamics literature, it has also been shown that the max relative entropy dictates the minimum amount of input work required to create a state via catalytic thermal operations \cite{HO-limitations}. 
		
The restriction of near perfect work per cycle can now be further justified by examining how well the heat engine performs when the machine runs over many cycles: we find that if $\Omega\leq 1$, the heat engine can be run quasi-statically with an efficiency arbitrarily close to the Carnot efficiency while extracting \textit{any} finite amount of work with an arbitrarily small entropy increase in the battery. This follows from repeatedly applying our single-shot results in the $\Omega\leq 1$ regime, as shown in Section \ref{Running the heat engine for many cycles quasi-statically} 
 of the~\suppl.	
		\begin{figure}
		\begin{center}
			\includegraphics[scale=0.33]{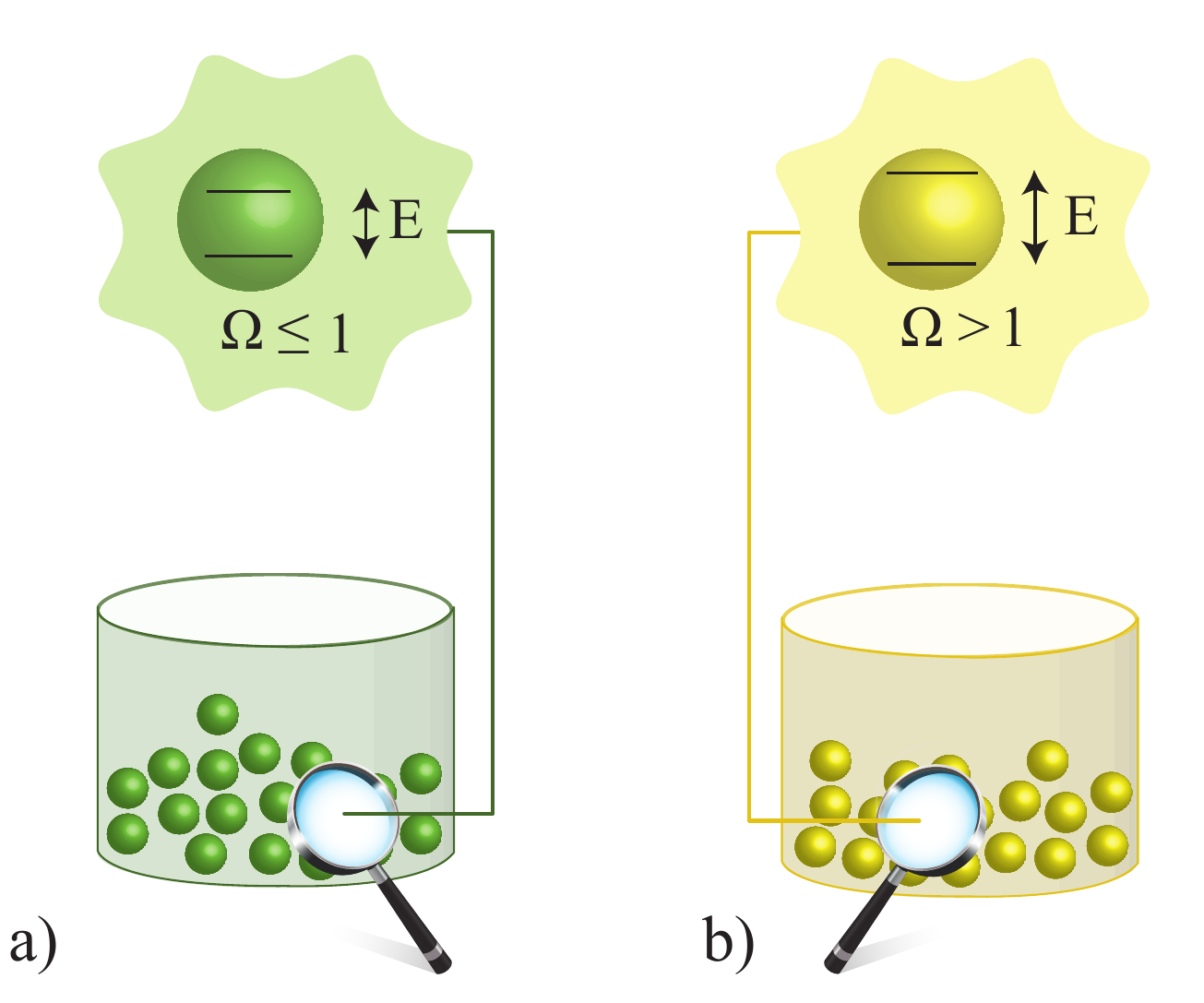}
			\caption{
				{\bf For fixed inverse temperatures $\beta_\cold$, $\beta_\hot$, the efficiency of a nanoscale heat engine depends on the structure of the cold bath.} 
				At the nano/quantum scale, Carnot's statement about the universality of heat engines does not hold. 
				We find that the maximum efficiency of a heat engine, does not \textit{only} depend on the inverse temperatures $\beta_\cold, \beta_\hot$ of the heat baths.
				In (a) the energy gaps are small enough to allow the heat engine to achieve Carnot efficiency, i.e., $\Omega\leq 1$. 
				In (b) the efficiency of the heat engine is reduced below the Carnot efficiency because the energy gap of the qubits are above the critical value $\Omega>1$.}
\label{fig:dependent agents}
		\end{center}
		\end{figure}

\subsection{Comparison to standard free energy results}\label{sec:Comparison with standard entropy results}

		For any system in thermal contact with a bath at inverse temperature $\beta$, consider the \helmholtz~free energy defined in Eq. \eqref{eq:HFE}. In the macroregime, the usual second law states that the \helmholtz~free energy never increases,
		\begin{equation}\label{eq:standardSecond}
			F(\rho_0) \geq F(\rho_1),
		\end{equation}
when the system goes from a state $\rho_0$ to a state $\rho_1$. Note that this quantity, $ F(\rho)$ 
	 is defined for arbitrary non-equilibrium states \cite{Esposito2011,parrondo2015thermodynamics}, where $ \beta $ is simply the inverse temperature of its surrounding bath. In the setting where one averages over infinitely many cycles, this quantity has been shown to correspond to the amount of work extractable from a generic, non-equilibrium state \cite{parrondo2015thermodynamics}. It also implies that $ F(\rho) $ dictates the possibility of asymptotically transforming $ n $ copies of $ \rho_0 $ into $ \rho_1 $, in the limit where $ n\rightarrow\infty $. 
	
In the single-shot quantum regime, however, Eq.~\eqref{eq:standardSecond} is but one of many conditions necessary for a state transformation. The generalized second laws \cite{2ndlaw} are a core result of quantum thermodynamics, that were derived based only on the foundations of quantum theory. Using these second laws, not only that many physical assumptions in classical thermodynamics can be avoided, but also one sees a more refined structure of the second law expressed in the context of nanoscale quantum systems. The limitations we observe on the efficiency are a consequence of having such generalized second laws.

From another perspective, the fact that more laws appear in this regime can intuitively be understood as being analogous to the fact that when performing a probabilistic experiment only a handful of times, not just the average, but other moments of a distribution become relevant. 
Indeed, all second laws converge to the standard second law in the limit of infinitely many particles~\cite{2ndlaw}, illustrating why we are traditionally accustomed to only this second law. 
The standard second law also emerges in some regimes of inexact catalysis~\cite{2ndlaw}, however, this corresponds to a degradation of the machine in each cycle, which would need one to use work to repair. Another example of a regime where the second law emerges, is when the catalysts are allowed to become correlated among themselves or with the rest of the machine \cite{lostaglio2014extracting,Muller2017}; however, in the limiting case in which only the second law is relevant, the dimension of the catalysts diverge. Furthermore, correlations between catalysts and the other states may require some additional work to destroy, hence making it less clear how to faithfully characterise the net extracted work. 

		It is illustrative to analyze our problem when we apply just the standard second law in Eq.~\eqref{eq:standardSecond} to derive bounds on efficiency, which is indeed a matter of textbook thermodynamics~\cite{reif1965fundamentals}. 
		However, we here apply the law precisely to the heat engine model as given in Section \ref{sec:setup main}, in which all energy flows are accounted for and (near) perfect work is performed. 
		One might wonder whether the limitations we observe are just due to either a limited model, or our demand for near perfect instead of average work, and might thus also arise even when invoking only the standard second law. 
		That is, are these newfound limitations really a consequence of the need to obey a wider family of second laws, or would the standard free energy predict the same things when energy is quantized, and quantum correlations are possible?

		We show independently of whether we consider perfect or near perfect work, 
		that according to the standard free energy in Eq.~\eqref{eq:standardSecond}, the maximum achievable efficiency is given by 
the Carnot efficiency. 
		Furthermore, we recover Carnot's statement that the Carnot efficiency can be achieved for \textit{any} cold bath (i.e. for a cold bath with any finite dimensional pure point spectrum). 
		We also see that for this case, the Carnot efficiency can always be achieved for quasi-static heat engines. 
		These results can be proven without invoking any additional assumptions than those laid out in Section \ref{sec:setup main}. In contrast, usual proofs of the second law require assumptions such as reversibility or that the system is in thermodynamic equilibrium at all times. Therefore, with our setup we recover exactly what Carnot predicted, namely that the maximum efficiency of heat engines only depends on the temperatures of the hot and cold bath. This rules out that our inability to achieve what Carnot predicted according to the macroscopic laws of thermodynamics is not merely the consequence of an overly stringent heat engine model, or definition of work.
	
		Finally, it is important to note that there have been several recent works \cite{uzdin2016quantum,SBLP2011smallest, skrzypczyk2014work, BLPS2012virtual, ent_enchance_cooling, tajima2014refined,Zeno_heatEngine,PhysRevX_Raam} on analyzing the efficiencies of small quantum heat engines, and had achieved Carnot efficiency. In \cite{skrzypczyk2014work}, a protocol was even constructed to achieve the Carnot efficiency for any system Hamiltonian and any arbitrary quantum state --- albeit considering operations which only preserve total energy on average. However, more common to all these approaches is that they consider an average notion of work, without directly accounting for a contribution from disordered energy (heat). Instead, one aims keeps the entropy of the battery low~\cite{skrzypczyk2014work}, or bound the higher moments of the energy distribution~\cite{BLPS2012virtual}. However, these only limit contributions from heat, but do not fully prevent them. Our notion of (near) perfect work now makes this aspect of macroscopic work explicit in the nanoregime, which has not been studied in the previous work. It is important to note that our work does not contradict previous results such as that of \cite{skrzypczyk2014work}. For example, the analysis of \cite{skrzypczyk2014work} shows that in each step of their protocol to achieve Carnot efficiency, the amount of energy change scales the same as the amount of entropy change, which does not correspond to perfect or near perfect work. 
		Needless to say, imperfect work with some contribution of heat can also be useful in certain scenarios. Yet, it does not quite constitute work if we cannot explicitly single out a contribution from heat. One could construct a machine which extracts some amount of energy, with some non-negligible amount of information. It is proven in this case that Carnot efficiency can even be exceeded~\cite{paper2}. This should not come as a surprise, because we are no longer asking for work --- energy transfer about which we have (near) perfect information.

		\section{Proof Overview}\label{sub:proof_overview}

To quantify the amount of extractable work, we apply the generalized second laws derived in \cite{2ndlaw}. The initial cold bath $\rho_\cold^0$ is thermal, and therefore diagonal in the energy eigenbasis, while the initial battery state $\rho_\batt^0$ is also a pure energy eigenstate (see Fig.~\ref{fig:battery}). Since the unitary $U(t)$ is energy conserving, it will never increase coherences between global energy eigenstates \cite{2ndlaw}. We can therefore conclude that $\rho_\CW^1$ is also diagonal in the energy eigenbasis. We can thus invoke the \emph{necessary and sufficient conditions} for a transformation to be possible~\cite{2ndlaw}. 
		Specifically, $\rho_\cold^0\otimes\rho_\batt^0\rightarrow\rho_\CW^1$ iff $\forall \alpha\geq 0$, 
		\begin{equation}\label{eq:secondlaws}
		F_\alpha(\rho_\cold^0\otimes\rho^0_\batt,\tau_\CW^h)\geq F_\alpha(\rho_\CW^1,\tau_\CW^h),
		\end{equation}
		where $\tau_\CW^h$ is the thermal state of the joint system (cold bath and battery) at temperature $T_\hot$. 
		The generalized free energy $F_\alpha$ is defined as
		\begin{equation}\label{eq:generalfreeenergyMT}
		F_\alpha(\rho,\tau):=\frac{1}{\beta_\hot} \left[D_\alpha(\rho\|\tau)-\ln Z_\hot\right],
		\end{equation}
		where $D_\alpha(\rho\|\tau)$ are known as $\alpha$-R{\'e}nyi divergences. For states $\rho,\tau$ which are diagonal in the same eigenbasis, the R{\'e}nyi divergences can be simplified to
		\begin{equation}\label{eq:renyi alfa divs Main text}
		D_\alpha(\rho\|\tau)=\frac{1}{\alpha-1}\ln \sum_i p_i^\alpha q_i^{1-\alpha},
		\end{equation}
		where $p_i,$ $q_i$ are the eigenvalues of $\rho$ and $\tau$ respectively. The case $\alpha= 1$ is defined by continuity in $\alpha$. Taking the limit $\alpha\rightarrow 1$ for Eq.~\eqref{eq:generalfreeenergyMT}, one recovers the \helmholtz~free energy, $F(\rho)=\langle\hat{H}\rangle_\rho -\beta_\hot^{-1} S(\rho)$. 
		Using the second laws~\cite{2ndlaw} is a powerful tool, since when searching for the optimum efficiency, we do not have to optimize explicitly
		over the possible machines $(\rho_\mach,\hat{H}_{\mach})$, the form of the hot bath $\hat{H}_\hot$, or the energy conserving unitary $U(t)$. 
		Whenever Eq.~\eqref{eq:secondlaws} is satisfied, then we are guaranteed a suitable choice exists and hence we can focus solely on the possible final states $\rho_\CW^1$. 
		
		Since we know that $\rho_\CW^1$ is diagonal in the energy eigenbasis, the correlations between cold bath and battery can only be classical (w.r.t. energy eigenbasis). However, even such correlations cannot improve the efficiency: we show in the~\suppl~that we may take the output state
		to have the form $\rho_\CW^1=\rho_\cold^1\otimes\rho_\batt^1$ in order to achieve the maximum efficiency.
		According to Fig.~\ref{fig:battery}, consider $\rho_\textup{W}^{0} = \proj{E_\batt^j}$ and the final state
		\begin{equation}\label{eq:rhobatt1}
		\rho_\textup{W}^{1} = (1-\varepsilon) \proj{E_\batt^k} + \varepsilon \proj{E_\batt^j}.
		\end{equation} 
		Although this is a particularly simple case of $ \rho_\batt^1 $, we can show that it is actually sufficient for our analysis, i.e. allowing a more general battery state does not change the maximum efficiency. We do this by analyzing the generalized free energy $ F_\alpha $ in the limit of $ \alpha\rightarrow\infty $, and show that any other final battery state achieves at most the same maximum efficiency given by Eq.~\eqref{eq:rhobatt1} (see \suppl~Section \ref{A more general final battery state}). 
		From the second laws Eq.~\eqref{eq:secondlaws}, we may derive the maximum amount of extractable work, which is the largest value of $W_\textup{ext} = E_\batt^k-E_\batt^j$ such that the state transition $\rho_\cold^0\otimes\rho_\batt^0\rightarrow\rho_\cold^1\otimes\rho_\batt^1$ is possible. The form of $W_\textup{ext}$ (derived in the~\suppl~Section \ref{subsub:highcertainty}) is
		\begin{align}
		W_{\rm ext} =& \inf_{\alpha\geq 0} ~W_\alpha,\label{eq:Maindefinitions1 no broduct bath}\\
		W_\alpha =& \frac{1}{\beta_\hot (\alpha-1)} [\ln (A-\varepsilon^\alpha)-\alpha\ln (1-\varepsilon)],\label{eq:Maindefinitions2 no broduct bath}\\
		A=& \frac{\sum_i p_i^\alpha q_i^{1-\alpha}}{\sum_i p_i'^\alpha q_i^{1-\alpha}},
		\end{align}
		where $p_i= \frac{e^{-\beta_\cold E_i}}{Z_{\beta_\cold}}$, $q_i= \frac{e^{-\beta_\hot E_i}}{Z_{\beta_\hot}}$ are probabilities of the thermal state of the cold bath at temperatures $\beta_\cold, \beta_\hot$ respectively, 
		and $p_i'$ are the probability amplitudes of state $\rho_\cold^1$ when written in the energy eigenbasis of $\hat H_\cold$. 
		The quantity $W_\textup{ext}$ is dependent on the initial and final cold bath $\rho_\cold^0, \rho_\cold^1$, the hot bath inverse temperature $\beta_\hot$, and the allowed failure probability $\varepsilon$. The difficulty of evaluating $W_\textup{ext}$ comes from the infimum over $\alpha$, which is completely dependent on $\beta_\hot, \beta_\cold, \hat H_\cold$ and other parameters.
		
		The efficiency $\eta$ defined in Eq. \eqref{eq:efficiency main text}, however, is not determined by the maximum extractable work, but rather by an optimal tradeoff between $W_{\rm ext}$ and the energy drawn from the hot bath, $\Delta \hot$. 
		Since $\hat H_\total = \hat H_\cold+\hat H_\hot+\hat H_\mach+\hat H_\batt$ is void of interaction terms, and since total energy is preserved, 
		we can also write the change of energy in the hot bath, in terms of the energy change in the remaining systems, 
		\begin{equation}
		\Delta \hot = \Delta \cold +\Delta W.
		\end{equation}
		where
		$ \Delta \cold:= \tr \left[ H_\cold\rho_\cold^1 \right]-\tr \left[ H_\cold\rho_\cold^0 \right], $
		and 
	$ 	\Delta W:= \tr(\hat H_\batt\rho_\batt^1)-\tr(\hat H_\batt\rho_\batt^0) $
		are \emph{changes} in average energy of the cold bath and battery. 
		We thus see that the efficiency can be described solely in terms of the battery and the cold bath.
		
		\noindent
		\textbf{Macroscopic second law}
		We first analyze the efficiency invoking only the standard second law,
		 namely assuming that the free energy ($\alpha= 1$) fully dictates if a certain state transition is possible. The question is then: given an initial cold bath Hamiltonian $\hat H_\cold$, what is the maximum attainable efficiency considering all possible final states $\rho_\cold^1$? In both cases of perfect and near perfect work, we find that the efficiency is only maximized whenever $\rho_\cold^1$ is (A) a thermal state, and (B), when it is a thermal state, can only achieve the Carnot efficiency when the inverse temperature $\beta_f$ is arbitrarily close to $\beta_\cold$. We refer to this situation as a \emph{quasi-static} heat engine. Moreover, we find that the Carnot efficiency can be achieved by any given $\hat H_\cold$. These results rigorously prove Carnot's findings when only the usual free energy is relevant.
		
		\noindent
		\textbf{Nanoscale second laws} 
		Here, when considering perfect work under the constraints of all generalized free energies coming into play, we are immediately faced with an obstacle: the constraint at $\alpha=0$ implies that $W_\textup{ext}>0$ is not possible, whenever $\rho_\cold^0$ is of full rank. This is due to the discontinuity of $D_0$ in Eq. \eqref{eq:renyi alfa divs Main text} w.r.t. the quantum state, and is similar to effects observed in information theory in lossy vs. lossless compression: no compression is possible if no error however small is allowed.
		However, when considering the limit $\Delta S/W_{\rm ext} \rightarrow 0$, i.e. near perfect work, the $D_0$ constraint is satisfied automatically.
		
		The results for the macroscopic second law form an upper bound for both the maximum extractable work and efficiency for nanoscale second laws, since the constraint of generalized free energy at $\alpha= 1$ is one of the many constraints described by all $\alpha\geq 0$. We can thus show that the results from the standard second law have the following implication in the nanoregime: if we can achieve the Carnot efficiency, we can only do so when $ \rho_\cold^1\rightarrow\tau_\cold^0 $. We analyze the quasi-static regime. Furthermore, we specialize to the case where the cold bath consists of multiple identical two-level systems, each of which are described by a Hamiltonian with energy gap $E$.
		
		\begin{figure}
			
			\includegraphics[width=0.7\linewidth]{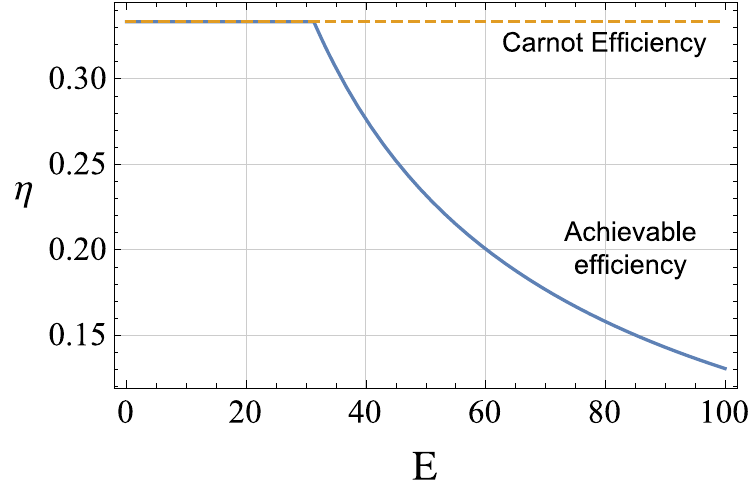}
			\label{fig:eff_plots_E}
			
			\includegraphics[width=0.7\linewidth]{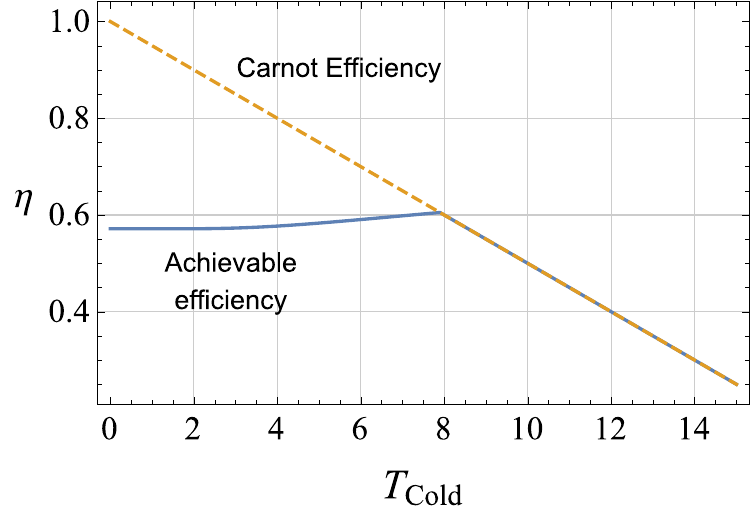}
			\label{fig:eff_plots_T_Cold}
			
			\includegraphics[width=0.7\linewidth]{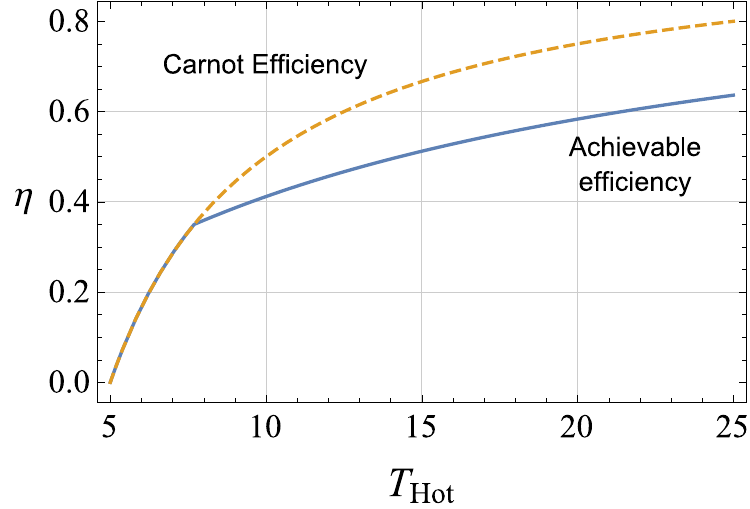}
			\label{fig:eff_plots_T_Hot}
			\caption{
				{\bf Comparison of the quasi-static nanoscale efficiency versus Carnot efficiency.} 
				\textbf{Top}: efficiency vs. the energy gap $E_\textup{min}$ of $\hat{H}_{\cold}$ (recall that $ k_{\rm B} =1$).
				According to Eqs.~\eqref{eq:carnotOmegaleq1},  \eqref{eq:subCarnotOmega>1}, for any $T_\cold<T_\hot$ one can achieve Carnot efficiency when $ E_\textup{min} $ is sufficiently small. Otherwise, we find a reduced efficiency. This has been plotted for $T_\hot=15$ and $T_\cold=10$. 
				\textbf{Middle}: efficiency vs. $T_\cold$, where $E_\textup{min}=15$. For every $\hat{H}_\cold$, there exists a temperature regime ($T_\cold$ vs. $T_\hot$)
				such that Carnot efficiency cannot be achieved. 
				This happens as $T_\cold$ gets further from the temperature of the hot bath $T_\hot=20$. 
				\textbf{Bottom}: efficiency vs. $T_\hot$, where $E_\textup{min}=15$. Similarly, we see the inability of attaining Carnot efficiency, as $ T_\hot $ increases relative to $T_\cold = 5$.}\label{fig:comparison}
		\end{figure}
	
		Firstly, we identify characteristics that the failure parameter of work extraction $\varepsilon$ should have, such that near perfect work is extracted in the limit $\beta_f\rightarrow\beta_\cold$ (i.e. when (A) and (B) are satisfied). We then show two technical results:
		\begin{enumerate}[leftmargin=*]
			\item The choice of $\varepsilon$ (as a function of $\beta_f$) simplifies the minimization problem in Eq.~\eqref{eq:Maindefinitions1 no broduct bath}, by reducing the range of the variable $\alpha$ appearing in the optimization of $W_{\rm ext}$. Under the consideration of near perfect work, $\varepsilon$ can be chosen such that the optimization of $\alpha$ is over $\alpha\geq\kappa$ for some $\kappa\in(0,1]$, instead of $\alpha\geq 0$. The larger $\kappa$ is for a chosen $\varepsilon$, the slower $\Delta S/W_\textup{ext}$ converges to zero. 
			\item We analyze the following cases separately: 
			\begin{itemize}[leftmargin=*]
				\item For $\Omega \leq 1$, $\varepsilon$ can always be chosen so that Eq.~\eqref{eq:Maindefinitions1 no broduct bath} is obtained in the limit $\alpha\rightarrow 1$. Evaluating the efficiency in the limit $\alpha\rightarrow 1$ corresponds to the Carnot efficiency. 
				\item For $\Omega > 1$, we show that for the best choice of $\varepsilon$, the infimum in Eq.~\eqref{eq:Maindefinitions1 no broduct bath} for $W_\textup{ext}$ is obtained at $\alpha\rightarrow\infty$. Furthermore, $\Omega>1$ means that up to leading order terms (with regards to the quasi-static limit parameters), $W_1 > W_\infty$ for $W_\alpha$ defined in Eq.~\eqref{eq:Maindefinitions2 no broduct bath}. But we know that the quantity $W_1$ gives us Carnot efficiency. Therefore, the  maximum  efficiency is strictly less than the Carnot value.
			\end{itemize}
		\end{enumerate}

\section{Discussions and Conclusion}\label{sec:conclusion}
		Our work establishes a fundamental result for the operation of nanoscale heat engines. 
		We find all cold baths can be used in heat engines; and --- remarkably --- that even when one of the heat baths is as small as a single qubit, as long as certain conditions of the bath Hamiltonian are met, Carnot efficiency can still be achieved in the quasi-static limit. However, Carnot efficiency is not necessarily always achieved for all baths; instead, achievability depends on the Hamiltonian structures of the baths. In the case where the cold bath consists of multiple qubits, we find that for all values of inverse temperatures $\beta_{\cold}$ and $\beta_{\hot}$ of the cold and hot baths, there exists an energy gap $E_\textup{min}(\beta_\cold,\beta_\hot)$ of the qubits forming the cold bath above, in which the optimal quasi-static efficiency is reduced below the Carnot efficiency. 
		Viewed from another direction, for a fixed energy gap $E_\textup{min}(\beta_\cold,\beta_\hot)$, whether the Carnot efficiency can be achieved depends on the relation between $T_\hot$ and $T_\cold$ as illustrated in Fig.~\ref{fig:comparison}.
		Loosely speaking, the Carnot efficiency can be achieved whenever the two temperatures are unequal but not too far apart. 
		One might wonder why this restriction has not been observed before in the classical scenario. There, the energy spectrum of the bath is continuous or forms a quasi-continuum, and hence we always have access to the regime where Carnot efficiency is achievable. 
		
		Nanoscale heat engines are starting to be constructed experimentally~\cite{scully2003extracting,nanoscaleHE}. Not all of these nanoscale heat engines will be able to obtain the Carnot efficiency due to sub-optimality of the heat engines. Our results may influence future nanoscale machine designs, since engineers may wish to use thermal baths that have small energy gaps such that $\Omega$ is not too large, depending on the temperature difference between the two baths involved. 

		Our result is a consequence of the fact that the second law takes on a more complicated form in the nanoregime. 
		It has been known for some time now that in addition to the standard second law, many other second laws become relevant and lead to additional restrictions. However, the implications of these restrictions on heat engines that operate in the quantum nanoregime have never been addressed in full until this paper.
		From a statistical perspective, small numbers require more refined descriptions than provided by averages, and as a result thermodynamics becomes more complicated when considering systems comprised of few particles. 
		Similar effects can also be observed in information theory, where averaged quantities as given by the Shannon entropy need to be supplemented with refined quantities when we consider finitely many channel uses. 

		In the macroscopic regime, for completeness, we ruled out the possibility that the observed limitations on efficiency is a consequence of our demand for near perfect work, or the fact that we are using systems with discrete (sufficiently large spaced) spectra. 
		This verification was achieved by showing that the Carnot efficiency can indeed always be attained (regardless to the size of an energy gap if present) when extracting near perfect work, when we are in such large systems that only the standard second law is relevant. 
		One might wonder whether heat engines that do not operate by extracting an infinitesimal amount of work, or employing quantum coherences would allow us to achieve the Carnot efficiency independent of the structure of the cold bath. 
		As we show in the~\suppl, both do not help.

It would furthermore be satisfying to derive the explicit form of a hot bath, and machine attaining Carnot --- or our reduced Carnot --- efficiency.
One might wonder whether a non-trivial machine $(\rho_\mach,\hat{H}_\mach)$ is needed at all in this case. 
To illustrate the dependence on the bath, it was sufficient to consider a bath comprising solely of qubits. The tools proposed in the~\suppl~can also be used to study other forms of bath structures, yet it is a non-trivial question to derive efficiencies for such cold baths. 

Most interestingly, there is the extremely challenging question of deriving a statement that is analogous to the Carnot efficiency, but which makes explicit the trade-off between information and energy for all possible starting situations. 
In a heat engine, we obtain energy from two thermal baths about which we have minimal information. It is clear that the Carnot efficiency is thus a special case of a more general statement in which we start with two systems
of a certain energy about which we may have some information, and we want to extract work by combining them. Indeed, the form that such a general statement should take is by itself a beautiful conceptual challenge, since what we understand as efficiency may not only be a matter of work obtained vs. energy wasted.
Instead, we may want to take a loss of information about the initial states into account when formulating such a fully general efficiency.

\myacknowledgments
We thank J. Oppenheim, M. Horodecki, R. Renner and A. Winter for useful comments and discussions. This research was funded by Singapore's MOE Tier 3A Grant MOE2012-T3-1-009, and STW, Netherlands. M.P.W. acknowledges funding from the Swiss National Science Foundation (Ambizione Fellowship PZ00P2\_179914). N.H.Y.N. acknowledges support from the Alexander von Humboldt foundation.




\printbibliography

\newtheorem{corollary}{Corollary}





\hypersetup{
	bookmarksnumbered,
	pdfstartview={FitH},
	citecolor={darkgreen},
	linkcolor={darkred},
	urlcolor={darkblue},
	pdfpagemode={UseOutlines}}
\definecolor{darkgreen}{RGB}{50,190,50}
\definecolor{darkblue}{RGB}{0,0,190}
\definecolor{darkred}{RGB}{238,0,0}

\renewcommand{\d}{\dagger}
\newcommand{\Tr}{\mathrm{Tr}}
\renewcommand{\Re}{\mathrm{Re}}
\renewcommand{\Im}{\mathrm{Im}}

\newcommand{\corr}{\textup{corr}}
\newcommand{\ncorr}{\textup{no corr}}
\newcommand{\qm}{\textup{qm}}

\newcommand{\hatrhocone}{\hat\rho_\textup{Cold}^\textup{a}}
\newcommand{\hatrhoctwo}{\hat\rho_\textup{Cold}^\textup{b}}

%


\onecolumn\newpage
\appendix
\begin{center}
	\textbf{{\large The maximum efficiency of nano heat engines depends on more than temperature: Supplementary Material}}
\end{center}

\tableofcontents

\vspace{0.5cm}
In this~\suppl , we detail our findings. Sections \ref{section:The setting}-\ref{section:efficiency_def} are aimed at giving the reader an overview of the important concepts regarding heat engines, and to introduce the quantities of interest. 
Firstly, in Section \ref{section:The setting} we describe the setup of our heat engine, the systems involved, and how work is extracted and stored. 
By using this general setup, we then proceed in Section \ref{sec:macro and nano themodynamics} to introduce conditions for thermodynamical state transitions in a cycle of a heat engine. 
In Section \ref{section:efficiency_def}, we introduce the formal definition of efficiency, and specify how can this quantity be maximized over a set of free parameters (involving the bath Hamiltonian structure). 

After providing these guidelines, we start in Section \ref{section:standardthermo} to apply the macroscopic law of thermodynamics. We have performed the analysis with the generalization of allowing for an arbitrarily small probability of failure. The results in this section might be familiar and known to the reader, however from a technical perspective, their establishment is helpful for proving our main results (in Section \ref{section:Nano/quantum scale heat engine at maximum efficiency}) about nanoscale systems. 
In Section \ref{section:Nano/quantum scale heat engine at maximum efficiency}, we apply the recently discovered generalizations of the second law for small quantum systems. 
The results in Section \ref{section:standardthermo} and Section \ref{section:Nano/quantum scale heat engine at maximum efficiency} are summarized at the beginning of each section, for the reader to have a concise overview of the distinction between thermodynamics of macroscopic and nanoscopic systems. 
Finally, in Section \ref{Extensions to the setup}, we show that even when considering a more general setup, these results obtained in Section \ref{section:Nano/quantum scale heat engine at maximum efficiency} remain unchanged.
\section{The general setting for a heat engine}\label{section:The setting}
A heat engine is a procedure for extracting work from a temperature difference. It comprises of four basic elements: two thermal baths at distinct temperatures $T_\hot$ and $T_\cold$ respectively, a machine, and a battery. 
The machine interacts with these baths in such a way that utilizes the temperature difference between the two baths to perform work extraction. 
The extracted work can then be transferred and stored in the battery, while the machine returns to its original state. 

In this section, we describe a fully general setup, where all involved systems and changes in energy are accounted for explicitly. Let us begin with the total Hamiltonian 
\begin{equation}\label{eq:totalH}
\hat H_t = \hat H_\cold +\hat H_\hot + \hat H_\mach +\hat H_\batt,
\end{equation}
where the indices Hot, Cold, M, W represent a hot thermal bath (Hot), a cold thermal bath (Cold), a machine (M), and a battery (W) respectively. 
Let us also consider an initial state
\begin{align}\label{eq:rho0}
	\rho^{0}_\total=\tau_\cold^{0}\otimes \tau_\hot^{0}\otimes \rho_\mach^{0}\otimes \rho_\battery^{0}.
\end{align}
The state $\tau_\hot^{0}$ ($\tau_\cold^{0}$) is the initial thermal state at temperature $T_\hot$ ($T_\cold$), corresponding to the hot (cold) bath Hamiltonians $\hat H_\hot, \hat H_\cold$. 
More generally, we have the following definition 
\begin{definition}(Thermal state)
	Given any Hamiltonian $\hat H$ and temperature $T$, the \textup{thermal state} is defined as $\tau=\frac{1}{\tr (e^{-\hat H/k_B T})}	e^{-\hat H/k_B T}$. For notational convenience, we shall often use inverse temperatures, defined as $\beta_h:=1/k_B T_\hot$ and $\beta_c:=1/k_B T_\cold$ where $k_B$ is the Boltzmann constant. 
\end{definition}
We will assume throughout that $T_\cold<T_\hot$. This is what is meant by ``hot''; namely that it is at a higher temperature than the ``cold'' bath. 
The initial machine $(\rho_\mach^0, \hat H_\mach)$ can be chosen arbitrarily, as long as its final state is preserved (and therefore the machine acts like a catalyst). 

The aim is to achieve a final \emph{reduced} state $\rho^1_\total$, such that
\begin{align}\label{eq:rho1}
\rho_\CMW^1 = \tr_\hot (\rho^{1}_\total)=\rho_\cold^{1}\otimes  \rho_\mach^{1}\otimes\rho^1_\battery,\qquad\qquad
\end{align}
where $\rho_\mach^{1}=\rho_\mach^{0}$, i.e. the machine is preserved, and $\rho_\cold^1,$ $\rho_\battery^{1}$ are the final states of the cold bath and battery. In Section \ref{Extensions to the setup}, we will consider the case in which there are correlations between the final state of the cold bath, hot bath, battery and or machine. We will find that the correlations do not change our results. For any bipartite state $\rho_{\rm AB}$, we use the notation of reduced states $\rho_{\rm A}:=\tr_\textup{B}(\rho_{\rm AB})$, $\rho_{\rm B}:=\tr_\textup{A}(\rho_{\rm BA})$.

Finally, we describe the battery such that the state transformation from $\rho^{0}_\total$ to $\rho^{1}_\total$ stores work in the battery. This is done as follows: consider the battery which has a Hamiltonian (written in its diagonal form)
\begin{align}\label{eq:W general 0}
\hat H_\battery:=\sum_{i=1}^{n_\battery} E^\battery_i|E_i\ra \la E_i|_\battery,\qquad\qquad\qquad
\end{align}
where $\{E_i^\battery\in\rr\}_{i=1}^{n_\battery}$ are arbitrary while $n_\battery\in\nn^+$ is an arbitrarily large fixed integer.
For some parameter $\varepsilon\in [0,1)$, we consider the initial and final states of the battery to be
\begin{align}
\rho_\battery^{0}&=|E_j\ra\la E_j|_\battery\label{eq:battery initial state}\\
\rho_\battery^{1}&=(1-\varepsilon)|E_k\ra\la E_k|_\battery+\varepsilon|E_j\ra\la E_j|_\battery\label{eq:battery final state}
\end{align}
respectively. 
The parameter $W_{\rm ext}$ is defined as the energy difference 
\be\label{W general}
W_{\rm ext}:=E_k^\battery-E_j^\battery.\qquad\qquad\qquad\qquad\quad
\ee
where we define $E_k^\battery>E_j^\battery$ such that $W_{\rm ext}>0$. 
We refer to the parameter $\varepsilon$ as the probability of failure of work extraction. Note that $\varepsilon$ in Eq. \eqref{eq:battery final state} is also the trace distance 
\begin{equation}\label{eq:trace dist def}
d(\rho,\sigma) = \frac{1}{2} \|\rho-\sigma\|_1={\frac  {1}{2}}{\tr}\left[{\sqrt  {(\rho -\sigma )^{\dagger }(\rho -\sigma )}}\right]
\end{equation}
between $\rho^1_\battery$ and $\ketbra{E_k}{E_k}_\battery$. In Section \ref{Extensions to the setup}, we will generalize this definition to include \textit{all} final states of the battery $\rho_\battery^1$, which are a trace distance $\varepsilon$ from the ideal final battery state $\ketbra{E_k}{E_k}_\battery$. We show that our findings regarding the achievability of C.E. remains unchanged.

Throughout our analysis, we deal with two distinct scenarios of work extraction as defined below. 
\begin{definition}(Perfect work)\label{def:perfect work}
An amount of work extracted $W_\textup{ext}$ is referred to as \textup{perfect work} when $\varepsilon=0$.
\end{definition}
The next definition of work involves a condition regarding the \emph{von Neumann entropy} of the final battery state. Let $\Delta S$ be the von Neumann entropy of the final battery state. When the initial state $\rho_\batt^0$ is pure, we have
\be\label{eq:von neu}
\Delta S:=-\tr(\rho_\battery^1\ln\rho_\battery^1).
\ee
When the final battery state is given by Eq. \eqref{eq:battery final state}, its probability distribution has its support on a two-dimensional subspace of the battery system, this definition also coincides with the binary entropy of $\varepsilon$,
\begin{equation}\label{def:binaryentropy}
\htwo (\varepsilon)=-\varepsilon\ln \varepsilon -(1-\varepsilon)\ln (1-\varepsilon)=\Delta S.\qquad
\end{equation}
We will see that no heat-engine can actually achieve the value of Carnot efficiency exactly, but moreover, that under certain conditions; some achieve it as a limiting process.\footnote{This subtlety is not unique to nano scale sized systems, but moreover is also true at the macroscopic level. In the standard formulation of Carnot's famous results about heat engine efficiency, the Carnot efficiency can only be achieved in the so-called ``quasi-static limit''.} For this reason, it is convenient to introduce the notion of a partially ordered set of heat engines. Roughly speaking, we will say later that a heat engine can achieve the Carnot efficiency when a heat engine in the closure of the set can achieve said efficiency. Let $\mathcal{S}^\textup{H.E.}=\{(W_\textup{ext},\varepsilon)\}$ denote a partially ordered set\footnote{The partial order is given by the condition $W_\textup{ext}\leq W_\textup{ext}'$ for two set elements $(W_\textup{ext},\varepsilon), (W_\textup{ext}',\varepsilon')\in \mathcal{S}^\textup{H.E.}$.} of heat engines with extracted work $W_\textup{ext}$ and corresponding failure parameter $\varepsilon$, introduced above.
\begin{definition}(Near perfect work)\label{def:near perfect work}
We say that a partially ordered set of heat engines $\mathcal{S}^\textup{H.E.}$ (introduced above) can achieve \textup{near perfect work} when
\begin{itemize}
	\item[1)] $0< \varepsilon\leq l,~$ for some fixed $l<1$ and
	\item[2)] For all $1>p>0$ there exists a non trivial subset $\mathcal{S}^\textup{H.E.}_p\subset \mathcal{S}^\textup{H.E.}$ such that when $(W_\textup{ext},\varepsilon)\in \mathcal{S}^\textup{H.E.}_p$,
	\be  0<\frac{\Delta S}{W_\textup{ext}}<p,\ee
	where recall $\Delta S=\Delta S(\varepsilon)$ \textup{[}see Eq. \eqref{def:binaryentropy}\textup{]}.
\end{itemize}
\end{definition}
Note that when this definition is used in lemmas and theorems, the precise type of ``heat engine'' will be specified, e.g. heat engines satisfying the macroscopic laws of thermodynamics (which are defined later).
In the main text, we have provided a detailed discussion regarding the physical meaning of perfect work and near perfect work, and the necessity for considering these quantities. In particular, it is discussed how it can only be achieved as a limiting process. Why we are interested in such limits will become apparent when we discuss the macroscopic case, even before we derive the efficiency in the nano regime.
As we will see later in the proof to Lemma \ref{lem:cannot do better than carnot with porb failuer}, 1) and 2) in Def. \ref{def:near perfect work} are both satisfied if and only if 
\be\label{eq:alternative equiv to n.p.w. def}
\displaystyle\lim_{\varepsilon\rightarrow 0^+} \frac{\Delta S}{W_\textup{ext}} =0.
\ee
Since the initial state $\rho_\sttotal^0$ is diagonal in the energy eigenbasis, and since catalytic thermal operations do not create coherences between energy eigenstates, therefore $\rho_\CMW^1$ has to be diagonal in the energy eigenbasis. Furthermore, (as already stated above) in Section \ref{Extensions to the setup}, we extend the setup to include correlation in the final state between the battery, cold bath and machine and more general final battery states.

Note that in our model we allow the battery to have arbitrarily many (but finite) eigenvalues. One can compare this to the two-dimensional battery used in \cite{2ndlaw}, referred to as the wit. Having a minimal dimension, the wit is a conceptually very useful tool to visualize work extraction. However, it has the disadvantage that the energy spacing, i.e. the amount of work to be extracted, has to be known a priori to the work being extracted in order to tune the energy gap of the wit. The more general battery, which we describe in Eq.~\eqref{eq:W general 0}, requires a higher system dimension, but has the advantage that it can form a quasi-continuum and thus effectively any amount of work (i.e. any $W_{\rm ext} >0$) can be stored in it without prior knowledge of the work extraction process. We will see that our results are independent of $n_\battery\geq 2$.
\begin{figure}
\begin{center}
\label{fig1}
\includegraphics[scale=0.33]{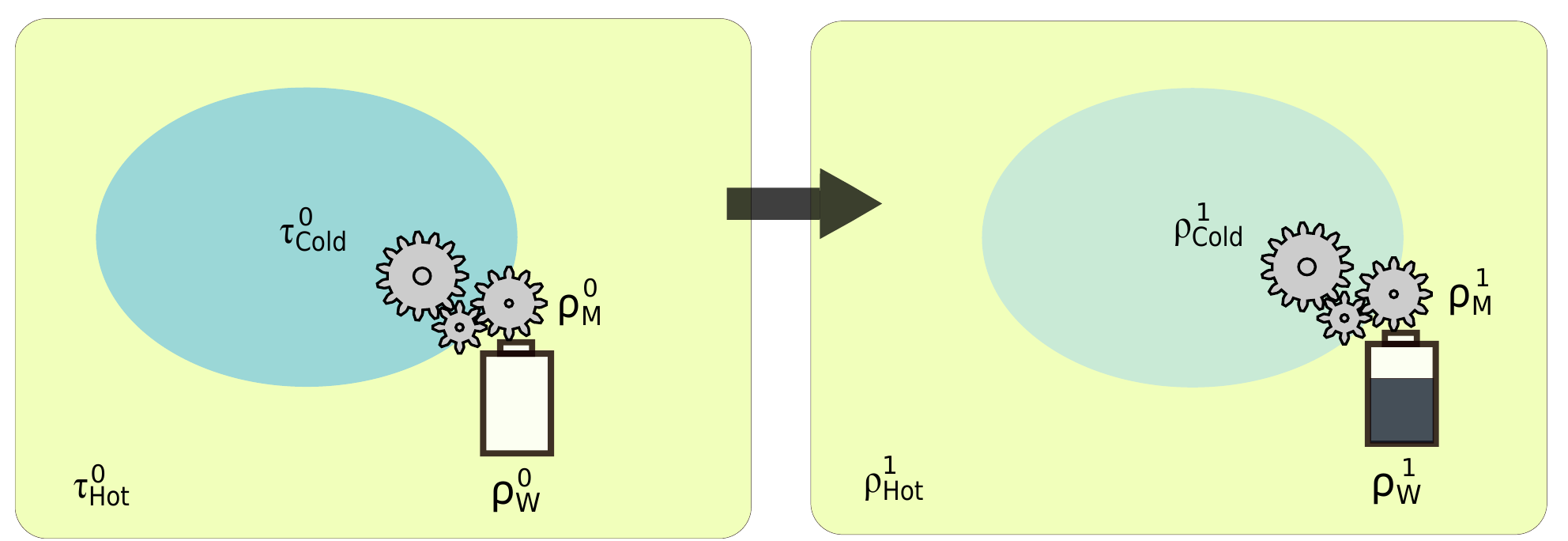}
\caption{The setting of a working heat engine.}
\end{center}
\end{figure}

To summarize, so far we have made the following minimal assumptions:
\begin{itemize}
\item [\textbf{ (A.1)}]
Product state: There are no initial nor final correlations between the cold bath, machine and battery. Initial correlations we assume do not exist, since each of the initial systems are brought independently into the process. This is an advantage of our setup, since if one assumed initial coherence, one would then have to use unknown resources to generate them in the first place. We later also show that correlations between the final cold bath and battery do not provide improvements in maximum extractable work or efficiency.
\item [\textbf{ (A.2)}]
Perfect cyclicity: The machine undergoes a cyclic process, i.e. $\rho_\mach^0=\rho_\mach^{1}.$ 
\item[\textbf{ (A.3)}]
Isolated quantum system: The heat engine as a whole, is isolated from and does not interact with the world. This assumption ensures that all possible resources in a work extraction process has been accounted for.
\item[\textbf{ (A.4)}] Finite dimension: The Hilbert space associated with $\rho_\sttotal^0$ is finite dimensional but can be arbitrarily large. Moreover, the Hamiltonians $\hat H_\cold,$ $\hat H_\hot,$ $\hat H_\mach$ and $\hat H_\battery$ all have bounded pure point spectra, meaning that these Hamiltonians have eigenvalues which are bounded.
\end{itemize}

After defining the set of allowed operations, and describing the desired state transformation process, one can then ask: what conditions should be fulfilled such that there exists a hot bath $(\tau_\hot^0,\hat H_\hot)$, and a machine $(\rho_\mach^0,\hat H_\mach)$ such that $\rho_\CW^0\rightarrow\rho_\CW^1$ is possible? Throughout this document we use ``$\rightarrow$'' to denote a state transition via catalytic thermal operations. 

In Section \ref{section:standardthermo}, by assuming the macroscopic law of thermodynamics governs the heat engine, we derive the efficiency of a heat engine, and verify the long known Carnot efficiency as the optimal efficiency. We do this for both cases where $\varepsilon=0$ and when $\varepsilon$ is arbitrarily small. In Section \ref{section:Nano/quantum scale heat engine at maximum efficiency}, we analyze the same problem under recently derived second laws, which hold for small quantum systems. We show that these new second laws lead to fundamental differences to the efficiency of a heat engine.

Throughout our analysis, a particular notion that describes thermodynamical transitions will be important towards achieving maximum efficiency. We therefore define this technical term, which will be used throughout the manuscript. 
\begin{definition}(Quasi-static)\label{def:quasi static}
A heat engine is \textup{quasi-static} if the final state of the cold bath is a thermal state and its inverse temperature $\beta_f$ only differs infinitesimally from the initial cold bath temperature, i.e. $\beta_f=\beta_c-g$, where $0<g\ll 1$.
\end{definition}

Since throughout this analysis we frequently deal with arbitrarily small paramaters $\varepsilon, g$, we also introduce beforehand the notation of order function $\bo(x)$, $\soo(x)$,
which denotes the growth of a function. 
\begin{definition}(Big $\bo$, small $\soo$ notation
~\cite{NTandAlgebra})\label{def:bigosmallo}
Consider two real-valued functions $P(x), Q(x)$. We say that \\
1. 
$P(x)=\bo(Q(x))$ in the limit $x\rightarrow a$ iff there exists $c_1,c_2 > 0$ and $\delta >0$ such that for all $|x-a|\leq\delta$, $c_1 \leq \left|\frac{P(x)}{Q(x)}\right| \leq c_2$. \\
2. $P(x)=\soo(Q(x))$ in the limit $x\rightarrow a$ iff there exists $c_3\geq 0$ such that  $~\displaystyle\lim_{x\rightarrow a} \left|\frac{P(x)}{Q(x)}\right|=c_3$. \\
\end{definition}
\begin{remark}
In Def.\ref{def:bigosmallo}, if the limit of $x$ is unspecified, by default we take $a=0$. In \cite{NTandAlgebra}, these order terms were only defined for $x\rightarrow\infty$. However, choosing a general limit $x\rightarrow a$ can be done by simply defining the variable $x'=1/(x-a)$, and $x\rightarrow a^+$ is the same as taking $x'\rightarrow\infty$.
\end{remark}
We also list a few properties of these functions here for $x\rightarrow 0$, which will help us throughout the proof:\\
a) For any $c\neq 0$, $\bo(c\cdot P(x))=\bo(P(x))$.\\
b) For any functions $P_1(x)$ and $P_2(x)$, $\bo(P_1(x))+\bo(P_2(x))= \bo\left(\max\left\lbrace |P_1(x)|,|P_2(x)|\right\rbrace\right)$.\\
c) For any functions $P_1(x)$ and $P_2(x)$, $\bo(P_1(x))\cdot\bo(P_2(x))= \bo (P_1(x)P_2(x))$.\\
d) For any functions $P_1(x)$ and $P_2(x)$, $\bo(P_1(x))/\bo(P_2(x))= \bo (P_1(x)/P_2(x))$.\\

Definition \ref{def:quasi static} has two direct implications for a quasi-static heat engine:
\begin{itemize}
\item[(i)] The temperature of the final state of the cold bath $T_f$, only increases w.r.t. its initial temperature by an infinitesimal amount, i.e. $T_f= T_\cold+T_\cold^2 \;g+\bo(g^2)$.

\item[(ii)] The amount of work extracted is infinitesimal: as we shall see later, the extractable perfect and near perfect work $W_\mathrm{ext}>0$ (see Defs. \eqref{def:perfect work}, \eqref{def:near perfect work}) is of order $\bo(g)$. This follows from using Eq. \eqref{eq:Wext} for the case where $\rho_\stcold^1$ is a thermal state with inverse temperature $\beta_f=\beta_c-g$, and calculating the Taylor expansion of $W_\mathrm{ext}$ about $g=0$.
\end{itemize}

\section{The conditions for thermodynamical state transitions}\label{sec:macro and nano themodynamics}
In this section, we state the laws which govern the transitions from initial, $\rho^{0}_\total$ to final, $\rho^{1}_\total$ states for one cycle of our heat engine. By applying these laws, the amount of extractable work $W_\textup{	ext}$ can be quantified and expressed as a function of the cold bath. We distinguish between two cases, the standard macroscopic regime, and the quantum regime.

\subsection{Second law for macroscopic systems}\label{sub:macro}
The cold bath, machine and battery form a \textit{closed} but not isolated thermodynamic system. This means only heat exchange (and not mass exchange) occurs between these systems and the hot bath. Therefore, a transition from $\rho^0_\CMW$ to $\rho^1_{\CMW}$ will be possible if and only if the \helmholtz~free energy, $F$ does not increase
\be \label{eq:Hemholts def}
F(\rho^{0}_\CMW)\geq F(\rho^{1}_\CMW),
\ee
where
\begin{equation}
F(\rho):=\la \hat H \ra_\rho-\frac{1}{\beta}S(\rho),
\end{equation}
and $S(\rho):= - \tr(\rho\ln \rho)$ and $\langle \hat H\rangle_\rho :=\tr (\hat{H}\rho)$ being the entropy and the mean energy of state $\rho$ respectively. Throughout the manuscript, whenever the state is a thermal state at temperature $\beta$, we use the shorthand notation $\langle\hat H_\cold\rangle_{\beta}$ and $S(\beta)$.

The \helmholtz~free energy bears a close relation to the \emph{relative entropy}, 
\begin{equation}\label{eq:defrelativeent}
D (\rho\|\sigma) = \tr (\rho\ln\rho) - \tr(\rho\ln\sigma).
\end{equation}
Whenever $\rho$ and $\sigma$ are diagonal in the same basis, the relative entropy can be written as
\begin{equation}\label{eq:defrelativeentdiag}
D(\rho\|\sigma) = \sum_i p_i\ln\frac{p_i}{q_i},
\end{equation}
where $p_i,q_i$ are the eigenvalues of $\rho$ and $\sigma$ respectively. Now, for any Hamiltonian $\hat H$, consider $\tau_\beta=\me^{-\beta \hat H}/Z_\beta$, which is the thermal state at some inverse temperature $\beta$, with partition function $Z_\beta = \tr[\me^{-\beta \hat H}]$, and denote its eigenvalues as $q_i$. Then for any diagonal state $\rho$ with eigenvalues $p_i$, and denoting $\lbrace E_i \rbrace_i$ as the eigenvalues of $\hat H$,
\begin{align}
D (\rho\|\tau_\beta) = \sum_i p_i\ln\frac{p_i}{q_i} = -S(\rho) + \sum_i p_i (\beta E_i + \ln Z_\beta) = \beta F(\rho) + \ln Z_\beta.
\end{align}
This implies that 
\begin{equation}\label{eq:helmholtzF}
F(\rho)=\frac{1}{\beta} [D(\rho\|\tau_\beta)- \ln Z_\beta].
\end{equation}
In Section \ref{section:standardthermo} we will solve Eq.~\eqref{eq:Hemholts def} in order to evaluate the maximum efficiency.

\subsection{Second laws for nanoscopic systems}\label{sub:nano}
In the microscopic quantum regime, where only a few quantum particles are involved, it has been shown that macroscopic thermodynamics is not a complete description of thermodynamical transitions. More precisely, not only the \helmholtz~free energy, but a whole other family of generalized free energies have to decrease during a state transition \cite{2ndlaw}. This places further constraints on whether a particular transition is allowed.
In particular, these laws also give necessary and sufficient conditions, when a system with initial state $\rho_\CW^0$ can be transformed to final state $\rho_\CW^1$ (both diagonal in the energy eigenbasis), with the help of any catalyst/machine which is returned to its initial state after the process.

Formally, these laws correspond to the following case:
A transition from the initial state $\rho^{0}_\CMW$ to the final state $\rho^{1}_\CMW$, is possible iff there exists an energy-preserving unitary $U(t)$ on the global system,  (i.e. a unitary that obeys $[U(t),\hat H_\total]=0$), where states $\rho^{0}_\CMW$, $\rho^{1}_\CMW$ are of the form Section described in \ref{section:The setting} (i.e. the state of the machine $\rho_\mach^0$ is preserved). 
If $(\tau_\hot^0, \hat H_\hot)$ and $(\rho_\mach^0, \hat H_\mach)$ can be arbitrarily chosen, these correspond to the set of \emph{catalytic thermal operations} \cite{2ndlaw} one can perform on the joint state $\CW$. This implies that the cold bath is used as a resource state.

We can apply these second laws to our scenario by associating the catalyst with $\rho^0_\mach$, and considering the state transition $\rho^0_\battery\otimes\tau_\cold^0 \rightarrow \rho^1_\battery\otimes\rho_\stcold^1$ as described in Section \ref{section:The setting}. Note that the initial state $\rho^0_\battery\otimes\tau_\cold^0 $ is block-diagonal in the energy eigenbasis (for the battery by our choice, and for the cold bath because it is a thermal state). By catalytic thermal operations, the final state is also block-diagonal in the energy eigenbasis. Furthermore, according to the second laws in \cite{2ndlaw}, the transition from $\rho^0_\battery\otimes\tau_\cold^0 \rightarrow \rho^1_\battery\otimes\rho_\stcold^1$ is then possible iff
\begin{align}\label{2nd law eq}
	F_\alpha(\tau_\cold^0\otimes\rho^0_\battery,\tau_\CW^h)\geq F_\alpha(\rho_\stcold^1\otimes\rho^1_\battery,\tau_\CW^h)\quad \forall \alpha\geq 0,
\end{align}
where $\tau_\CW^h$ is the thermal state of the system at temperature $T_\hot$ of the surrounding bath. The quantity $F_\alpha(\rho,\sigma)$ for $\alpha\geq 0$ corresponds to a family of free energies defined in \cite{2ndlaw}, which can be written in the form
\begin{align}\label{eq:generalfreeenergy}
	F_\alpha(\rho,\tau)=\frac{1}{\beta_h} \left[D_\alpha(\rho\|\tau)-\ln Z_h\right],
\end{align}
 where $D_\alpha(\rho\|\tau)$ are known as $\alpha$-R{\'e}nyi divergences. Sometimes we will use the short hand $F_\infty:=\lim_{\alpha\rightarrow \infty}F_\alpha$. On occasion, we will refer to a particular transition as being possible/impossible according to the $F_\alpha$ free energy constraint. By this, we mean that for that particular value of $\alpha$ and transition, Eq. \eqref{2nd law eq} is satisfied/not satisfied. The $\alpha$-R{\'e}nyi divergences can be defined for arbitrary quantum states, giving us necessary (but insufficient) second laws for state transitions \cite{2ndlaw,LJR2015description}. However, since we are analyzing states which are diagonal in the same eigenbasis, the R{\'e}nyi divergences can be simplified to
\begin{align}\label{w no ep}
	D_\alpha(\rho\|\tau)=\frac{1}{\alpha-1}\ln \sum_i p_i^\alpha q_i^{1-\alpha},
\end{align}
where $p_i,$ $q_i$ are the eigenvalues of $\rho$ and the state $\tau$. The cases $\alpha=0$ and $\alpha\rightarrow1$ are defined by continuity, namely
\begin{align}\label{eq:reyi in limits}
	D_0(\rho\|\tau)&=\lim_{\alpha\rightarrow 0^+}D_\alpha(\rho\|\tau)=-\ln \sum_{i:p_i\neq 0}q_i,\\ 
	D_1(\rho\|\tau)&=\lim_{\alpha\rightarrow 1}D_\alpha(\rho\|\tau)=\sum_{i}p_i\ln \frac{p_i}{q_i},\label{eq:reyi in limits 2}
\end{align}
and we also define $D_\infty$ as
\begin{align}\label{eq:D_inft def}
	D_\infty(\rho\|\tau)&=\lim_{\alpha\rightarrow \infty^+}D_\alpha(\rho\|\tau)=\ln \max_{i}\frac{p_i}{q_i}.
\end{align}
The quantity $D_1(\rho\|\tau)$ coincides with $D(\rho\|\tau)$, as we have defined in Eq.~\eqref{eq:defrelativeent}, and evaluated in Eq.~\eqref{eq:defrelativeentdiag} for diagonal states. We will often use this convention. 
Furthermore, since we are considering initial states which are block-diagonal in the energy eigenbasis, these generalized second laws are both necessary and \textit{sufficient} conditions for state transformations. 
Therefore, in Section \ref{subsub:highcertainty} we will solve Eq. \eqref{2nd law eq} explicitly to find an expression for $W_\textup{ext}$ with the ultimate goal of evaluating the maximum efficiency in this regime.

The reader should note that for both Section \ref{sub:macro} and \ref{sub:nano}, the conditions for state transformation place upper bounds on the quantity $W_\textup{ext}$. In particular, this allows us to express the maximum values $W_\textup{ext}$ can take (such that the joint state transformation of cold bath and battery is possible) in terms of quantities related to the cold bath, and the error probability $\varepsilon$. It is also worth comparing the conditions for state transformation in Section \ref{sub:macro} and \ref{sub:nano}, which are stated in Eqs.~\eqref{eq:Hemholts def} and \eqref{2nd law eq}. In particular, Eq.~\eqref{eq:Hemholts def} is but a particular instance of Eq.~\eqref{2nd law eq}, and therefore the nanoscopic second laws always place a stronger upper bound on $W_\textup{ext}$ compared to the macroscopic second law.

\section{Efficiency, maximum efficiency and how to evaluate it}\label{section:efficiency_def}
The central quantity of interest in this article is the efficiency of heat engines. Since we have already introduced the notion of a heat engine in Section \ref{section:The setting}, and the rules which govern the possibility of thermodynamical transitions of one cycle of a heat engine in Section \ref{sec:macro and nano themodynamics}, it is timely to define the efficiency. After defining this quantity, we demonstrate how to go about calculating its maximum value under different conditions, such as for perfect work, near perfect work, in both the macroscopic and nanoscopic regimes. This will prepare the scene for Sections \ref{section:standardthermo} and \ref{section:Nano/quantum scale heat engine at maximum efficiency}, where we evaluate the maximum efficiency more explicitly.
\subsection{Definition of efficiency and maximum efficiency}
As stated in the main text, the efficiency of a particular heat engine (recall that a heat engine is defined by its initial and final states $\rho_\total^0, \rho_\total^1$ as described in Section \ref{section:The setting}) is defined as
\be\label{eq:efficiency}
\eta:=\frac{W_{\rm ext}}{\Delta H},
\ee
where $W_{\rm ext}$ is the amount of work extracted which is defined in Eq. \eqref{W general}, and $\Delta H$ is the amount of mean energy drawn from the hot bath, namely $\Delta H:= \tr(\hat H_\hot \rho_\sthot^0)-\tr(\hat H_\hot \rho_\sthot^1)$, where $\rho_\sthot^1$ is the reduced state of the hot bath. 

Now, consider the set of conditions on state transformations given by Eq.~\eqref{2nd law eq} for nanoscale systems.
As discussed in Section \ref{sec:macro and nano themodynamics}, these conditions place a restriction on the range of values $W_\textup{ext}$ can take. Therefore, for any fixed $\rho_\cold^1$, we define $\eta^\textup{nano}(\rho_\cold^1)$ as the maximum achievable efficiency as a function of the final state of the cold bath. More precisely,
\begin{align}
&\eta^\textup{nano}(\rho_\cold^1)\label{eq:max nano as function fo cold bath}\\
&=\sup_{W_\textup{ext}} \eta(\rho_\cold^1) \quad \textup{ subject to }\quad F_\alpha(\rho^0_\battery\otimes\tau_\cold^0,\tau_\CW^h)\geq F_\alpha(\rho^1_\battery\otimes\rho_\stcold^1,\tau_\CW^h)\quad \forall \alpha\geq 0.\label{eq:1max nano as function fo cold bath}
\end{align}
In Eq.~\eqref{eq:max nano as function fo cold bath}, we have written the quantity in Eq.~\eqref{eq:efficiency} as $\eta=\eta(\rho_\cold^1)$ to remind ourselves of its explicit final cold bath state dependency. 
Therefore, the maximum efficiency will correspond to maximizing over the final state of the cold bath:
\be \label{eq:max eff general}
\eta_\textup{max}=\sup_{\rho_\cold^1\in\mathcal{S}}\eta^\textup{nano}(\rho_\cold^1),
\ee 
 where $\mathcal{S}$ is the space of all quantum states in $\mathcal{H}_\cold$. 
By analyzing this quantity in Section \ref{section:Nano/quantum scale heat engine at maximum efficiency}, we show that perfect work cannot be extracted. Therefore, when we calculate the maximization in Eq. \eqref{eq:max eff general} we will consider near perfect work (see Def. \ref{def:near perfect work}).
 
In the macro regime, we have to satisfy a less stringent requirement, namely the macroscopic second law of thermodynamics. And hence we have that for fixed $\rho_\cold^1$, $\eta^\textup{mac}(\rho_\cold^1)$ is the maximum efficiency as a function of $\rho_\cold^1$ 
\begin{align} \label{eq:max eff macro function of rho_C^1}
\eta^\textup{mac}(\rho_\cold^1)=\sup_{W_\textup{ext}} \eta(\rho_\cold^1) \quad &\textup{ subject to } \quad F(\rho^0_\CMW)\geq F(\rho^1_\CMW)\\
\quad &\textup{ and } \quad
\tr(\hat H_t \rho_\sttotal^0)=\tr(\hat H_t \rho_\sttotal^1),
\end{align}
where $\hat H_t$ is defined in Eq.~\eqref{eq:totalH}.
Similarly to the nanoscale setting, the maximum efficiency is
\be \label{eq:max eff macro}
\eta_\textup{max}=\sup_{\rho_\cold^1\in\mathcal{S}}\eta^\textup{mac}(\rho_\cold^1).
\ee
We can also define the maximum quasi-static efficiencies for the macro and nano scale. The maximum efficiency of a quasi-static heat engine (see Def. \ref{def:quasi static}), is
\begin{align}
\eta_\textup{max}^\textup{stat}=\lim_{g\rightarrow 0^+}\eta^\textup{nano}(\tau(g)),\label{def:quasi static eff nano}\\
\eta_\textup{max}^\textup{stat}=\lim_{g\rightarrow 0^+}\eta^\textup{mac}(\tau(g)),\label{def:quasi static eff mac}
\end{align}
for the \emph{nanoscopic} and \emph{macroscopic} cases respectively. $\tau(g)\in\mathcal{H}_\cold$ is the thermal state with Hamiltonian $\hat H_\cold$ at temperature $\beta_f=\beta_c-g$ and $\eta^\textup{nano},\eta^\textup{mac}$ are defined in Eqs. \eqref{eq:max nano as function fo cold bath} and \eqref{eq:max eff macro function of rho_C^1} respectively.
Since we can extract perfect and near perfect work in the macroscopic setting, we will derive the efficiency for both cases in Section \ref{section:standardthermo}.

\subsection{Finding a simplified expression for the efficiency}\label{sub:finding_simplified_expression}
We can find a more useful expression for $\Delta H$ appearing in Eq.~\eqref{eq:efficiency}. This can be obtained by observing that since only energy preserving operations are allowed, we have
\begin{equation}\label{eq:conservetotalenergy}
\tr(\hat H_t \rho_\sttotal^0)=\tr(\hat H_t \rho_\sttotal^1),
\end{equation}
where $\hat H_t =\hat H_\hot+\hat H_\cold+\hat H_\mach+\hat H_\battery$. Since the Hamiltonian does not contain interaction terms between these systems, the mean energy depends only on the \emph{reduced states} of each system. Mathematically, it means that Eq.~\eqref{eq:conservetotalenergy} can be written as 
\begin{align}
\tr(\hat H_\hot \rho_\sthot^0)+\tr(\hat H_\cold \rho_\stcold^0)+\tr(\hat H_\mach \rho_\mach^0)+\tr(\hat H_\battery \rho_\battery^0)
=\\ \tr(\hat H_\hot \rho_\sthot^1)+\tr(\hat H_\cold \rho_\stcold^1)+\tr(\hat H_\mach \rho_\mach^1)+\tr(\hat H_\battery \rho_\battery^1).
\end{align}
Also, note that since $\rho_\mach^0=\rho_\mach^1$, therefore $\tr(\hat H_\mach \rho_\mach^0)=\tr(\hat H_\mach \rho_\mach^1)$. This implies that we have 
\be \label{eq:Delta H in terms of C and W}
\Delta H = \Delta C + \Delta W,
\ee
where
\begin{align} \label{eq:delta C def}
\Delta C&:= \tr \left[ \hat H_\cold\rho_\stcold^1 \right]-\tr \left[\hat H_\cold\tau_\cold^c \right],
\end{align}
and 
\begin{align}\label{eq:Delta W def as average}
\Delta W&:= \tr(\hat H_\battery\rho_\battery^1)-\tr(\hat H_\battery\rho_\battery^0).
\end{align}
are the change in average energy of the cold bath and battery. We can thus write Eq. \eqref{eq:efficiency} as
\be 
\eta=\frac{W_\textup{ext}}{\Delta W+\Delta C}.
\ee
Furthermore, from Eqs. \eqref{eq:battery initial state}, \eqref{eq:battery final state}, \eqref{W general}  and \eqref{eq:Delta W def as average}, we have $\Delta W=(1-\varepsilon)W_\textup{ext}$, and hence we can write the inverse efficiency as
\be \label{eq:eff explicit function of cold bath}
\eta^{-1}(\rho_\cold^1)=1-\varepsilon + \frac{\Delta C(\rho_\cold^1)}{W_\textup{ext}(\rho_\cold^1)},
\ee
where we have made explicit the $\rho_\cold^1$ dependency. We already know from the setting that $\rho_\cold^0$ is thermal. If $\rho_\cold^1$ is also a thermal state at some temperature $\beta$  according to the cold bath Hamiltonian $\hat H_\cold$, we will sometimes use the shorthand notation $\eta(\beta)$ for $\eta(\rho^1_\cold)$ and $\Delta W(\beta)$, $\Delta  C(\beta)$ for $\Delta W(\rho^1_\cold)$, $\Delta  C(\rho^1_\cold)$ respectively.

In Section \ref{section:standardthermo}, we will derive an expression for $W_\textup{ext}$ and solve Eqs. \eqref{eq:max eff macro function of rho_C^1}, \eqref{eq:max eff macro}. In Section \ref{section:Nano/quantum scale heat engine at maximum efficiency}, we will derive a new expression for $W_\textup{ext}$ in the nanoscopic regime, and solve Eqs.  \eqref{eq:max nano as function fo cold bath}, \eqref{eq:max eff general}.

\section{Efficiency of a heat engine according to macroscopic thermodynamics}\label{section:standardthermo}

In this section, we study the efficiency of the setup detailed in Section \ref{section:The setting} under the constraints of macroscopic thermodynamics, as described in Section \ref{sub:macro}. This implies that the \helmholtz~free energy solely dictates whether $\rho_\CW^0\rightarrow\rho_\CW^1$ is possible. We find that in both cases of extracting perfect and near perfect work,
\begin{itemize}
\item [(1)] The maximum achievable efficiency is the Carnot efficiency.
\item [(2)] The Carnot efficiency can be achieved for any  cold bath Hamiltonian.
\item [(3)] For any $\Delta C$, the maximum efficiency achievable for the particular value of $\Delta C$, is achieved iff the final state of the cold bath is thermal (according to a different temperature $T_f$).
\item [(4)] When the final state of the cold bath is thermal, the Carnot efficiency is achieved iff we take the limit corresponding to a quasi-static heat engine (Eq. \eqref{def:quasi static eff mac}). Roughly speaking, this means that there is only infinitesimal change in the final temperature of the cold bath, compared to its original state.
\end{itemize}

This section can be summarized as follows: in Section \ref{subsec:maxextwork_macro}, we first apply the macroscopic law of thermodynamics, namely the fact that \helmholtz~free energy is non-increasing, to our heat engine setup. By making use of energy conservation, we can derive the amount of maximum extractable work as shown in Eq.~\eqref{eq:Helmholts as D_1}. Next, in Section \ref{sec:Maximum efficiency} we show that when considering the extraction of perfect work, we show the points (1)-(4) as stated above. In Section \ref{sub:macroeffNPW}, we show that points (1)-(4) hold also when considering near perfect work.

The main results can be found in Theorem \ref{theorem:classical CE} and Lemma \ref{lem:CEmaximum_generalbatt}.
One may think points (1)-(4) are obvious since it has long been known that the optimal achievable efficiency of a heat engine operating between two thermal baths is the Carnot efficiency, and that this efficiency can be achieved quasi-statically. The motivations for proving these results here are two-fold. Firstly, this is a rigorous and mathematical proof of optimality, while usually one encounters arguments such as reversibility, or that the heat engine must remain in thermal equilibrium at all times during the working of the heat engine. Secondly, we will find later on at the nano/quantum scale that the Carnot efficiency can be achieved but observation (2) does not hold anymore. For these reasons, it is worthwhile proving that one can actually achieve points (1)-(4) in this setting for any cold bath Hamiltonian according to macroscopic thermodynamics. From a practical point of view, many of the technical results proved here will be needed in the proofs of Section \ref{section:Nano/quantum scale heat engine at maximum efficiency}, where we derive results involving a more refined set of generalized free energies, which describes thermodynamic transitions for nanoscale quantum systems.\\

\subsection{Maximum extractable work according to macroscopic law of thermodynamics}\label{subsec:maxextwork_macro}
Our first task is to find an expression for $W_\textup{ext}$ in the macro regime. We do so by solving Eq. \eqref{eq:Hemholts def} for $W_\textup{ext}$ such that
\be\label{eq:Helmholts free energy ineq}
\langle \hat H_\CMW\rangle_{\rho^1_\CMW}-\frac{1}{\beta_h} S(\rho^1_\CMW) \leq \langle \hat H_\CMW\rangle_{\rho^0_\CMW}- \frac{1}{\beta_h}S(\rho^0_\CMW).
\ee 
The entropy is an additive quantity under tensor product, meaning that $S(\rho_1\otimes\rho_2)=S(\rho_1)+S(\rho_2)$ for any states $\rho_1,\rho_2$. Furthermore, since the joint Hamiltonian does not contain interaction terms, therefore the mean energy also depends only on the reduced states. In summary, both $S$ and $\langle \hat H\rangle$ are additive under a tensor product structure of $\rho^0_\CMW$ and $\rho^1_\CMW$ as described in Eqs.~\eqref{eq:rho0} and \eqref{eq:rho1}. This means one can rewrite Eq.~\eqref{eq:Helmholts free energy ineq} by expanding its terms,
\begin{align}\label{eq:expansionofeq5}
&\langle \hat H_\cold\rangle_{\rho^1_\cold}+\langle \hat H_\mach\rangle_{\rho^1_\mach}+\langle \hat H_\battery\rangle_{\rho^1_\battery}-\frac{1}{\beta_h} \left[S(\rho^1_\cold)+S(\rho^1_\mach)+S(\rho^1_\battery)\right] \leq  \\
&\qquad\langle \hat H_\cold\rangle_{\rho^0_\cold}+\langle \hat H_\mach\rangle_{\rho^0_\mach}+\langle \hat H_\battery\rangle_{\rho^0_\battery}-\frac{1}{\beta_h} \left[S(\rho^0_\cold)+S(\rho^0_\mach)+S(\rho^0_\battery)\right], \nonumber
\end{align}
Furthermore, note that $\rho_\mach^0=\rho_\mach^1$, and therefore $S(\rho_\mach^0), \langle \hat H_\mach\rangle_{\rho_\mach^0}$ are common terms on both sides of Eq. \eqref{eq:expansionofeq5} which can be cancelled out. Furthermore, by our construction of the battery in Eqs.~\eqref{eq:W general 0}-\eqref{W general}, we have that $S(\rho_\battery^0)=0$, $S(\rho_\battery^1)=\Delta S=\htwo(\varepsilon)$ and $\langle \hat H_\battery\rangle_{\rho^0_\battery}=E_j^\battery$ and  $\langle \hat H_\battery\rangle_{\rho^1_\battery} = (1-\varepsilon)E^{W}_k+\varepsilon E^{W}_j$. Thus, Eq.~\eqref{eq:expansionofeq5} can be simplified to
\begin{equation}\label{eq:Wext01}
(1-\varepsilon) W_{\rm ext}+ \langle \hat H_\cold\rangle_{\rho^1_\cold} - \frac{1}{\beta_h}S(\rho^1_\cold)  \leq  \langle \hat H_\cold\rangle_{\rho^0_\cold}-\frac{1}{\beta_h}S(\rho^0_\cold)+\frac{1}{\beta_h}\htwo(\varepsilon) ,
\end{equation}
where $W_{\rm ext}$ has been defined in Eq.~\eqref{W general}. In other words, $(1-\varepsilon) W_\textup{ext}\leq F(\rho_\cold^0)-F(\rho_\cold^1)+\frac{1}{\beta_h}\htwo(\varepsilon)$.

We can also express $W_\text{ext}$ with the \emph{relative entropy} instead, by using Eq.~\eqref{eq:helmholtzF}. 
We can apply this identity to Eq.~\eqref{eq:Wext01} whenever the initial and final states are diagonal in the energy eigenbasis. Note that the initial $\rho_\stcold^0$ is a thermal state (of some temperature), and therefore diagonal in the energy eigenbasis. Since we start with a state $\tau_\cold^0\otimes\rho_\battery^0$ which is diagonal w.r.t. the Hamiltonian, and since catalytic thermal operations can never increase coherences between energy eigenstates (or in the macro setting, since we only demand mean energy conservation), we know that the final state $\rho_\stcold^1\otimes\rho_\battery^1$ is also diagonal in the energy eigenbasis. 
Therefore, Eq.~\eqref{eq:Wext01} can be rewritten w.r.t. the relative entropies as follows
\begin{equation}\label{eq:Helmholts as D_1}
(1-\varepsilon)W_{\rm ext}\leq F(\rho_\cold^0)-F(\rho_\cold^1)+\frac{1}{\beta_h}\htwo(\varepsilon)= \frac{1}{\beta_h} \left[D (\rho_\stcold^0\|\tau_\cold^h)-D (\rho_\stcold^1\|\tau_\cold^h)+\htwo(\varepsilon)\right].
\end{equation}


\subsection{Maximum efficiency for perfect work is Carnot efficiency}\label{sec:Maximum efficiency}
In this section, we want to find the maximum efficiency according to Eqs.~\eqref{eq:efficiency}, \eqref{eq:max eff macro function of rho_C^1} and \eqref{eq:max eff macro}, for the case of $\varepsilon=0$ which implies $\htwo(\varepsilon)=0$. 
We do this by the following steps:
\begin{enumerate}
\item \textit{Evaluate $W_\textup{ext}$.} According to Eq.~\eqref{eq:Helmholts as D_1}, we know that 
\begin{equation}\label{eq:Wext}
W_{\rm ext}= F(\rho_\cold^0)-F(\rho_\cold^1)=\frac{1}{\beta_h} \left[D (\rho_\stcold^0\|\tau_\cold^h)-D (\rho_\stcold^1\|\tau_\cold^h)\right],
\end{equation}
where recall that we have defined $\tau_\cold^h$ previously as the thermal state of system $\cold$ with temperature $T_\hot$. Note that here equality can be achieved because in macroscopic thermodynamics, satisfying the free energy constraint is a necessary and sufficient condition for the possibility of a state transformation. 
Note that since by construction the initial and final states of the battery are pure energy eigenstates, namely $\varepsilon=0$ and therefore
\begin{equation}\label{eq:deltaWeqWext}
W_{\rm ext}=\Delta W.
\end{equation} 
\item \textit{Write inverse maximum efficiency as optimization problem.} By substituting the simplified expression for efficiency derived in Eq.~\eqref{eq:eff explicit function of cold bath} into Eq.~\eqref{eq:max eff macro}, we have
\begin{equation}\label{eq:def_effinv}
\eta_{\rm max}^{-1} = \inf_{\rho_\stcold^1} (\eta^{\rm mac})^{-1} = 1+\inf_{\rho_\stcold^1} \frac{\Delta C}{W_\textup{ext}}.
\end{equation}
\item \textit{Maximize $W_\textup{ext}$ given a fixed value of $\Delta C$.} This is done in Lemma \ref{lem:minimizingD1}, where we show that given a fixed $\Delta C$, the final cold bath state that maximizes $W_\textup{ext}$ is uniquely a thermal state, corresponding to a certain inverse temperature $\beta'$.
\item \textit{Show that 3) implies that efficiency is maximized by a thermal state of the cold bath.} This is proven in Lemma \ref{lemma: optimum eff is for themal state}. Therefore, this implies one only needs to optimize Eq.~\eqref{eq:def_effinv} over one variable, i.e. $\beta_f$, the final temperature of the cold bath.
\item \textit{Show that the efficiency is strictly increasing with $\beta_f$.} This is done first by proving several identities, which are summarized in Corollary \ref{cor:specialid}. Using these identities, we prove in Lemma \ref{lemma:optimaleff is for infinitesimal temp change} that the first derivative of efficiency w.r.t. $\beta_f$ is always positive over the range where $W_\textup{ext}>0$. This leads us to conclude, in Theorem \ref{theorem:classical CE}, that maximum efficiency is achieved in the limit $\beta_f\rightarrow\beta_c$, and evaluating the efficiency at this limit gives us the Carnot efficiency. 
\end{enumerate}

Firstly, let us develop a technical Lemma \ref{lem:minimizingD1}, which concerns the unique solution towards maximizing $W_\textup{ext}$ for a fixed $\Delta C$. By applying Lemma \ref{lem:minimizingD1}, we show in Lemma \ref{lemma: optimum eff is for themal state} that the maximal efficiency is achieved when $\rho_\stcold^1$ is a thermal state. The reader can easily find similar proofs in \cite{cover2012elements}.
\begin{lemma}\label{lem:minimizingD1}
Given any Hamiltonian $\hat H_\cold$, a corresponding thermal state $\tau_\cold^h$ of some temperature $\beta_h$, and a fixed initial state $\rho_\stcold^0$, consider the maximization over final states $\rho_\stcold^1$,
\begin{equation}\label{eq:maxDeltaW}
\max_{\rho_\stcold^1} W_\textup{ext} \, = \frac{1}{\beta_h} \left[D (\rho_\stcold^0\|\tau_\cold^h)- \min_{\rho_\stcold^1} D (\rho_\stcold^1\|\tau_\cold^h)\right].
\end{equation}
over all states $\rho_\stcold^1$ which are diagonal in the energy eigenbasis, subject to the constraint that $\Delta C$ is a constant. Then the solution for $\rho_\stcold^1$ is unique, and $\rho$ is a thermal state according to the Hamiltonian $\hat H_\cold$ at a certain temperature $\beta '$.
\end{lemma}
\begin{proof}
Firstly, from Eq.~\eqref{eq:delta C def} we see that the constraint $\Delta C$ being a constant, is the same as $\tr \left[ \hat H_\cold\rho_\stcold^1 \right]$ being a constant. This is because they differ only by a constant term.
On the other hand, from Eq.~\eqref{eq:Delta W def as average} and \eqref{eq:deltaWeqWext}, we can see that $\max_{\rho_\stcold^1}  W_\textup{ext}$ is equal to the R.H.S. of Eq. \eqref{eq:maxDeltaW}. 
Since $\rho_\stcold^1$ and $\tau$ are both diagonal in the energy eigenbasis ($\rho_\stcold^1$ by the statement in the lemma, and $\tau$ by it being a thermal state), one can evaluate the relative entropy by using Eq.~\eqref{eq:defrelativeent}. Denote the eigenvalues of our variable $\rho_\stcold^1$ to be $\lbrace p_i \rbrace_i$, and the eigenvalues of the thermal state $\tau$ to be $\lbrace q_i \rbrace_i$. We can then write the optimization problem as
\begin{align*}
\min_{\lbrace p_i \rbrace}	&\sum_i p_i(\ln p_i-\ln q_i);\quad \text{ subject to } \sum_i p_iE_i=c\quad {\rm constant, and} ~\sum_i p_i= 1. \\
&{\rm where} \quad q_i=\frac{e^{-\beta E_i}}{Z_\beta};\quad Z_\beta=\sum_ie^{-\beta E_i}.
\end{align*}

We can now employ techniques of Lagrange multipliers to solve this optimization. The constrained Lagrange equation is
\begin{align}
	L(\{p_i\},\lambda)&=\sum_i p_i(\ln p_i-\ln q_i)+\lambda \left(\sum_i E_i p_i-c\right)+\mu \left(\sum_i p_i -1\right),\\
	\frac{dL}{dp_i}&=(\ln p_i-\ln q_i+1+\lambda E_i +\mu)=0,\\
	\frac{dL}{d\lambda}&=\sum_i E_i p_i-c=0.\\
	\frac{dL}{d\mu}&=\sum_i p_i-1=0.
\end{align}
We find that the normalized solution is 
\begin{align}\label{sol lagrange}
p_i=\frac{e^{-\beta' E_i}}{Z_{\beta'}},\quad Z_{\beta'}=e^{(1+\mu)}Z_\beta,	
\end{align}
and $p_i$ are probabilities corresponding to the Boltzmann distribution, according to inverse temperature $\beta'=\beta+\lambda$. Depending on the mean energy constraint $c$ and normalization condition, one can solve for the Lagrange multipliers $\lambda$ and $\mu$. With this we conclude that the state $\rho$ which maximizes  $D(\rho_\stcold^1\|\tau)$ is a thermal state, where its temperature is such that the constraint on mean energy is satisfied. 
\end{proof}

\begin{lemma}\label{lemma: optimum eff is for themal state}
Consider the work extraction process described by the state transformation $\rho_\CMW^0\rightarrow\rho_\CMW^1$, where $\rho_\stcold^0$, $\rho_\battery^0$ and $\rho_\battery^1$ have been described in Section \ref{section:The setting}. Denote $\mathcal{H}_\cold$ as the Hilbert space of the cold bath. Then the maximal efficiency in Eq.~\eqref{eq:def_effinv} is obtained for a final state of the cold bath $\rho_\stcold^1$, which is thermal:
\begin{equation}
\eta_{\rm max}^{-1} =1+\inf_{\rho_\stcold^1\in \mathcal{S}_\tau}\frac{\Delta C}{ W_\textup{ext}},
\end{equation}
where $\mathcal{S}_\tau$ the set of all thermal states (for $\hat{H}_\cold$ with any temperature $T>0$) in $\mathcal{H}_\cold$. Furthermore, all non-thermal 
 states do not achieve the maximum efficiency, i.e.
\begin{equation}
\eta_{\rm max}^{-1}<1+\frac{\Delta C}{ W_\textup{ext}}\Big{|}_{\rho_\stcold^1}\quad\text{ for any } \rho_\stcold^1\in \mathcal{S}
\setminus\mathcal{S}_\tau.
\end{equation}
where $\mathcal{S}$ is the space of all quantum states in $\mathcal{H}_\cold$
\end{lemma}
\begin{proof}
First of all, note that without loss of generality we can always consider only diagonal states, as explained in the paragraph before Eq.~\eqref{eq:Helmholts as D_1} that catalytic thermal operations do not increase coherences between energy eigenstates.
We begin by substituting Eqs.~\eqref{eq:delta C def} and \eqref{eq:Wext} into Eq.~\eqref{eq:def_effinv}, and finding
\begin{align}
\eta_{\rm max}^{-1} &= 1+\inf_{\rho_\stcold^1} \frac{\Delta C}{ W_\textup{ext}}\\
&= 1+ \inf_{\rho_\stcold^1} \frac{\beta_h\Delta C}{D_1(\tau_\cold^c\|\tau_\cold^h)-D_1 (\rho_\stcold^1\|\tau_\cold^h)}\\
&= 1+\beta_h\left[\sup_{\rho_\stcold^1} \frac{D_1(\tau_\cold^c\|\tau_\cold^h)-D_1 (\rho_\stcold^1\|\tau_\cold^h)}{\tr (\hat H_\cold\rho_\stcold^1)-\tr(\hat H_\cold\tau_\cold^c)}\right]^{-1}.\label{eq:etainv}
\end{align}
In the last line of Eq.~\eqref{eq:etainv}, we see that only two terms depend on the maximization variable $\rho_\stcold^1$. This means we can perform the maximization in two steps:
\begin{equation}\label{eq:2maxproblems}
\sup_{\rho_\stcold^1} \frac{D_1(\tau_\cold^c\|\tau_\cold^h)-D_1 (\rho_\stcold^1\|\tau_\cold^h)}{\tr (\hat H_\cold\rho_\stcold^1)-\tr(\hat H_\cold\tau_\cold^c)} = \sup_{A > 0} \frac{D_1(\tau_\cold^c\|\tau_\cold^h)- B(A)}{A}
\end{equation}
where $B(A)$ is the optimal value of a separate minimization problem:
\begin{equation}\label{eq:submaxprob}
B(A) = \displaystyle\inf_{\substack{\rho_\stcold^1 \in\mathcal{S} \\ \tr (H_\cold\rho_\stcold^1)-\tr(\hat H_\cold\tau_\cold^c)=A}}D_1 (\rho_\stcold^1\|\tau_\cold^h)
\end{equation}
From Lemma \ref{lem:minimizingD1}, we know that the solution of the sub-minimization problem in Eq.~\eqref{eq:submaxprob} has a unique form, namely $\rho_\stcold^1=\tau_\cold^f$ is a thermal state of some temperature $\beta_f$.  Therefore,  Eq.~\eqref{eq:2maxproblems} can be simplified to 
\begin{equation}\label{eq:max_over_beta}
\sup_{\rho_\stcold^1} \frac{D_1(\tau_\cold^c\|\tau_\cold^h)-D_1 (\rho_\stcold^1\|\tau_\cold^h)}{\tr (\hat H_\cold\rho_\stcold^1)-\tr(\hat H_\cold\tau_\cold^c)} = \sup_{\beta_f}
 \frac{D_1(\tau_\cold^c\|\tau_\cold^h)-D_1 (\tau_\cold^f\|\tau_\cold^h)}{\tr (\hat H_\cold\tau_\cold^f)-\tr(\hat H_\cold\tau_\cold^c)}.
 \end{equation}
Whats more, for every constant $A$, the function
\be 
f(x)= \left(1+\beta_h\left[\frac{D_1(\tau_\cold^c\|\tau_\cold^h)-x}{A}\right]^{-1}\right)^{-1}
\ee
 is bijective in $x\in\rr$ and thus due to the uniqueness of the sub-minimization problem in Eq.~\eqref{eq:submaxprob}, we conclude that for all non-thermal states $\rho_\stcold^1$, the corresponding efficiency will be strictly less than that of Eq.~\eqref{eq:etainv}.
Thus from Eq. \eqref{eq:max_over_beta} and \eqref{eq:etainv} we conclude the lemma.
\end{proof}

After establishing Lemma \ref{lemma: optimum eff is for themal state}, we can continue to solve the optimization problem in Eq.~\eqref{eq:def_effinv} by only looking at final states which are thermal (according to some final temperature $\beta_f$ which we optimize over). In the next Lemma \ref{lem:identities} and Corollary \ref{cor:specialid}, we derive some useful and interesting identities. These identities concern quantities such as the derivatives of mean energy and entropy of the thermal state (with respect to inverse temperature), and relates them to the variance of energy. 
The reader can find similar proofs in any standard thermodynamic textbook (For example in Sections 6.5, 6. of \cite{reif1965fundamentals}), but we derive them here for completeness.

\begin{lemma}\label{lem:identities}
For any cold bath Hamiltonian $\hat H_\cold$, consider the thermal state $\tau_\beta = \frac{1}{Z_\beta}e^{-\beta \hat H_\cold}$ with inverse temperature $\beta$. Define $\langle \hat H_\cold\rangle_{\beta} = \tr (\hat H_\cold \tau_\beta)$, and $S(\beta) = -\tau_\beta\ln\tau_\beta$ to be the mean energy and entropy of $\tau_\beta$. Then the following identities hold:
\begin{align}
\frac{d \langle \hat H_\cold \rangle_{\beta}}{d\beta} &= -{\rm var} (\hat H_\cold)_\beta\\
\frac{d S(\beta)}{d\beta} &= - \beta\cdot {\rm var} (\hat H_\cold)_\beta,
\end{align}
where ${\rm var}(\hat H_\cold)_{\beta} = \langle \hat H_\cold^2\rangle_{\beta}-\langle \hat H_\cold\rangle_{\beta}^2$ is the variance of energy for $\tau_\beta$.
\end{lemma}
\begin{proof}
Intuitively we know that the expectation value of energy increases as temperature increases (or as the inverse temperature decreases). More precisely, consider the probabilities of $\tau_\beta$ for each energy level of the Hamiltonian $E_i$,
\begin{align}
p_i &= \frac{e^{-\beta E_i}}{Z_\beta}, \quad {\rm where}~ Z_\beta=\sum_i e^{-\beta E_i}\nonumber\\
\frac{d p_i}{d \beta} &= \frac{1}{Z_\beta^2} \left[ -E_i e^{-\beta E_i}\cdot Z_\beta - \frac{d Z_\beta}{d \beta}\cdot e^{-\beta E_i} \right]=-p_i E_i - \frac{1}{Z_\beta} \frac{d Z_\beta}{d \beta} p_i=-p_i E_i +p_i \langle \hat H_\cold \rangle_{\beta}.\label{eq:piderivative}
\end{align}
The last equality holds because of the following identity:
\begin{equation}\label{eq:identity_derivative_Z}
\frac{-1}{Z}\frac{d Z}{d \beta} = \frac{-1}{Z} \sum_i (-E_i) e^{-\beta E_i} = \sum_i p_i E_i=\langle \hat H_\cold \rangle_{\beta}.
\end{equation} 
Therefore, we have
\begin{align}\label{eq:derivativeDeltaC}
\frac{d \langle \hat H_\cold \rangle_{\beta}}{d \beta} &= \sum_i \frac{d \langle \hat H_\cold \rangle_{\beta}}{d p_i} \frac{d p_i}{d \beta}= \sum_i E_i\cdot \left[ -p_iE_i+p_i \langle \hat H_\cold \rangle_{\beta}\right]\\
&= - \langle \hat H_\cold^2 \rangle_{\beta} + \langle \hat H_\cold \rangle_{\beta}^2 = -{\rm var} (\hat H_\cold)_\beta.
\end{align}
On the other hand, similarly, one can prove the second identity by writing down the expression of entropy for the thermal state,
\begin{align}
S(\beta)&=-\sum_i \frac{e^{-\beta E_i}}{Z_\beta}\ln \frac{e^{-\beta E_i}}{Z_\beta} = \sum_i \beta E_i \frac{e^{-\beta E_i}}{Z_\beta} + \ln Z_\beta \sum_i \frac{e^{-\beta E_i}}{Z_\beta} = \beta \langle \hat H_\cold \rangle_{\beta} + \ln Z_\beta.
\end{align}
Therefore, the derivative of $S(\beta)$ w.r.t. $\beta$ is
\begin{align}\label{eq: S div}
\frac{d S(\tau_{\beta})}{d\beta} &= \langle \hat H_\cold \rangle_{\beta} + \beta \frac{d \langle \hat H_\cold \rangle_{\beta}}{d \beta} + \frac{1}{Z_\beta} \frac{d Z_\beta}{d \beta}= \beta \cdot \frac{d \langle \hat H_\cold \rangle_{\beta}}{d \beta}=-\beta \cdot {\rm var} (\hat H_\cold)_\beta.
\end{align}
\end{proof}

By using Lemma \ref{lem:identities} in a special case, we obtain the following corollary:
\begin{corollary}\label{cor:specialid}
Given any Hamiltonian $\hat H_\cold$, consider the quantities
\begin{align}\label{eq:DCdefinition}
\Delta C(\beta_f)  = \tr (\hat H_\cold \tau_{\beta_f})-\tr(\hat H_\cold\tau_{\beta_c}) = \langle \hat H_\cold \rangle_{\beta_f}-\langle \hat H_\cold \rangle_{\beta_c}
\end{align}
and
\begin{align}\label{eq:DWdefinition}
 W_\textup{ext}(\beta_f)  = F(\tau_{\beta_c})-F(\tau_{\beta_f})=\frac{1}{\beta_h}\left[ D(\tau_{\beta_c}\|\tau_{\beta_h})-D(\tau_{\beta_f}\|\tau_{\beta_h}) \right],
\end{align}
where $\tau_{\beta}$ corresponds to the thermal state defined by $\hat H_\cold$ at inverse temperature $\beta$. Then
\begin{align}
\frac{d \Delta C(\beta_f)}{d\beta_f} &= -{\rm var} (\hat H_\cold)_{\beta_f}\\
\frac{d  W_\textup{ext}(\beta_f)}{d\beta_f} &= \frac{\beta_h-\beta_f}{\beta_h} {\rm var} (\hat H_\cold)_{\beta_f}.
\end{align}
\end{corollary}

\begin{proof}
For $\Delta C(\beta_f)$, it is straightforward from Lemma \ref{lem:identities} that
\begin{align}\label{eq:dev Delta C in terms of var}
\frac{d \Delta C(\beta_f)}{d\beta_f} = \frac{d \langle \hat H_\cold \rangle_{\beta_f}}{d\beta_f} = -{\rm var} (\hat H_\cold)_{\beta_f}.
\end{align}
On the other hand, $\Delta W(\beta_f)$ can be simplified by substituting Eq.~\eqref{eq:helmholtzF} into Eq.~\eqref{eq:DWdefinition},
\begin{align}\label{eq:simplifyDeltaW}
 W_\textup{ext}(\beta_f) &= F(\tau_{\beta_c})-F(\tau_{\beta_f}) = \langle \hat H_\cold \rangle_{\beta_c}-\langle \hat H_\cold \rangle_{\beta_f} - \frac{1}{\beta_h} \left[S(\tau_{\beta_c})-S(\tau_{\beta_f})\right].
\end{align}
With this, we can evaluate the derivative
\begin{align*}
\frac{d W_\textup{ext}(\beta_f)}{d\beta_f} &= - \frac{d\langle \hat H_\cold \rangle_{\beta_f}}{d\beta_f} + \frac{1}{\beta_h} \frac{dS(\tau_{\beta_f})}{d\beta_f}\nonumber\\
&= {\rm var} (\hat H_\cold)_{\beta_f} - \frac{\beta_f}{\beta_h}{\rm var} (\hat H_\cold)_{\beta_f} \nonumber\\
&= \frac{\beta_h-\beta_f}{\beta_h}{\rm var} (\hat H_\cold)_{\beta_f}.
\end{align*}
The second equality is obtained by Lemma \ref{lem:identities} for $\frac{d\langle \hat H_\cold \rangle_{\beta_f}}{d\beta_f}$, and the third by grouping common factors together. 
\end{proof}

In the next step, by using Corollary \ref{cor:specialid}, we show that when the final state of the cold bath is thermal, the optimal efficiency is achieved only in the quasi-static limit, i.e. in the limit $\beta_f\rightarrow\beta_c$ when the efficiency is optimised over all final thermal states of the cold bath. 

\begin{lemma}\label{lemma:optimaleff is for infinitesimal temp change}
Evaluate the efficiency expressed in Eq.~\eqref{eq:eff explicit function of cold bath} for the situation where the final state of the cold bath is a thermal state at inverse temperature $\beta_f$:
\be 
\eta(\beta_f)=\frac{ W_\textup{ext}(\beta_f)}{\Delta C(\beta_f)+ W_\textup{ext}(\beta_f)}.
\ee
Then for all $\beta_f < \beta_c$, $\frac{d \eta(\beta_f)}{d\beta_f} > 0.$
\end{lemma}
\begin{proof}
To prove this, we show that $\frac{d \eta^{-1}}{d \beta_f}< 0$, where $\eta^{-1}=1+\frac{\Delta C}{ W_\textup{ext}}$. 
Evaluating the derivative of $\eta^{-1}$ w.r.t. $\beta_f$, we obtain
\begin{align}
\frac{d\eta^{-1}}{d\beta_f} & = \frac{1}{ W_\textup{ext}^2}\cdot \left[\frac{d\Delta C(\beta_f)}{d\beta_f} W_\textup{ext} - \frac{d W_\textup{ext}(\beta_f)}{d\beta_f} \Delta C\right]\\
&=\frac{{\rm var} (\hat H_\cold)_{\beta_f}}{W_\textup{ext}^2} \cdot \left[-W_\textup{ext} - \frac{\beta_h-\beta_f}{\beta_h}\Delta C\right]\\
&=\frac{{\rm var} (\hat H_\cold)_{\beta_f}}{W_\textup{ext}^2} \cdot \left[\Delta C + \frac{1}{\beta_h}\left[S(\tau_{\beta_c})-S(\tau_{\beta_f})\right]- \frac{\beta_h-\beta_f}{\beta_h}\Delta C\right]\\
&=\frac{{\rm var} (\hat H_\cold)_{\beta_f}}{W_\textup{ext}^2} \frac{\beta_f}{\beta_h}\cdot \left[\Delta C - \frac{1}{\beta_f}[S(\tau_{\beta_f})-S(\tau_{\beta_c})]\right].
\end{align}

The first equality is obtained by invoking the chain rule of differentiation. The second equality is obtained by  substituting $\frac{d W_\textup{ext}}{d\beta_f}, \frac{d\Delta C}{d\beta_f},$ as evaluated earlier in Corollary \ref{cor:specialid}. The third equality is obtained by expressing $W_\textup{ext}$ according to Eq.~\eqref{eq:simplifyDeltaW}, plus recognizing that $\langle \hat H_\cold \rangle_{\tau_{\beta_f}}-\langle\hat H_\cold \rangle_{\tau_{\beta_c}}=\Delta C$. The last inequality is obtained, simply by taking out a common term $\beta_f/\beta_h$. We then make the following observations: \\
1) The factor
\be
\frac{\beta_f}{\beta_h W_\textup{ext}^2}> 0,
\ee
2) The variance of energy for any positive temperature
\be\label{Delta C div}
{\rm var} (\hat H_\cold)_{\beta_f}> 0,
\ee
3) and the last term $\Delta C - \frac{1}{\beta_f}[S(\tau_{\beta_f})-S(\tau_{\beta_c})]$ can be written as $F(\tau_{\beta_f})-F(\tau_{\beta_c})$, where $F$ is the free energy of a system w.r.t. a bath with inverse temperature $\beta_f$. But then, since $\tau_{\beta_f}$ is the thermal state with the same inverse temperature, this means that $\tau_{\beta_f}$ is the \emph{unique} state that minimizes free energy. Therefore, $F(\tau_{\beta_c})-F(\tau_{\beta_f})> 0$ for any $\tau_{\beta_c}$. 

\end{proof}
From Lemma \ref{lemma: optimum eff is for themal state} and Lemma \ref{lemma:optimaleff is for infinitesimal temp change}, we conclude that the maximization of efficiency for any Hamiltonian $\hat H$ happens for a final state which is thermal, and the greater its inverse temperature $\beta_f$, the higher efficiency is. With these lemmas we can now prove the main result of this section (Theorem \ref{theorem:classical CE}). 

In the next theorem, we evaluate the efficiency at the limit $\beta_f\rightarrow\beta_c^-$, and show that it corresponds to the Carnot efficiency.

\begin{theorem}[Carnot Efficiency]\label{theorem:classical CE}
Consider all heat engines which extract perfect work (see Definition \ref{def:perfect work}). Then according to the macroscopic second law of thermodynamics, the maximum achievable efficiency (see Eq. \eqref{eq:max eff macro}) is the Carnot efficiency
\be 
\eta_\textup{max}=1-\frac{\beta_h}{\beta_c}.
\ee
It can be obtained for \textup{all} cold bath Hamiltonians $\hat H_\cold$, but when the final state of the cold bath $\rho^1_\cold$ is thermal, then \textup{only} for quasi-static heat engines (as defined in Def. \ref{def:quasi static} and Eq. \eqref{def:quasi static eff mac} for quasi-static maximum efficiency). In this quasi-static limit, an infinitesimal amount of work is extracted.
\end{theorem}
\begin{proof}
From Eq. \eqref{eq:max eff macro}, we have an expression for the optimal efficiency in terms of a maximization over final cold bath states $\rho_\stcold^1\in\mathcal{S}$. By Lemma \ref{lemma: optimum eff is for themal state}, we know that the optimal solution is obtained only for thermal states. Subsequently, by Lemma \ref{lemma:optimaleff is for infinitesimal temp change}, it is shown that when the final cold bath is of temperature $\beta_f$, the corresponding efficiency is strictly increasing w.r.t. $\beta_f$. 
Also note that since by definition $W_{\rm ext} > 0$, this implies that $\beta_f<\beta_c$. Intuitively, this is because heat cannot flow from a cold to hot system without any work input. One can also see this mathematically, by showing that for any $\beta \geq \beta_h$, 
\begin{align}
\frac{d F(\tau_\beta)}{d\beta} & = \frac{d}{d\beta} \left[\langle \hat H_\cold \rangle_{\beta} -\frac{1}{\beta_h} S(\beta)\right] = \left(\frac{\beta}{\beta_h}-1\right) {\rm var} (\hat H_\cold)_\beta \geq 0.
\end{align}
This implies that if $\beta_f\geq\beta_c\geq\beta_h$, then $F(\beta_f)\geq F(\beta_c)$, and according to Eq.~\eqref{eq:DWdefinition} $W_\textup{ext}\leq 0$.
Therefore, when the final state of the cold bath $\rho_\stcold^1$ is thermal, the optimal efficiency must be achieved only when its inverse temperature $\beta_f$ approaches $\beta_c$ from below. Let $\beta_f = \beta_c -g$, where $g>0$. Then we have
\be\label{eta inv *}
\eta_\textup{max}^{-1}=\lim_{g\rightarrow 0^+}(\eta^\textup{mac})^{-1}(\beta_c-g),\quad (\eta^\textup{mac})^{-1}(\beta_c-g)=1+\frac{\Delta C}{W_\textup{ext}}\Big{|}_{\rho_\stcold^1=\tau_{(\beta_c-g)}}.
\ee
Since as $g\rightarrow 0^+$, both the numerator and denominator vanish, we can evaluate this limit by first applying L'H\^ospital rule, the chain rule for derivatives (for any function $F$, $\frac{dF}{dg}=-\frac{dF}{d\beta_f}$), and then Corollary \ref{cor:specialid} to obtain
\begin{align*}
\lim_{g\rightarrow 0^+} \frac{\Delta C}{W_\textup{ext}} = \lim_{g\rightarrow 0^+}\frac{\frac{d\Delta C}{dg}}{\frac{dW_\textup{ext}}{dg}}= \lim_{\beta_f\rightarrow \beta_c^-}\frac{\frac{d\Delta C}{d\beta_f}}{\frac{dW_\textup{ext}}{d\beta_f}} = \frac{\beta_h}{\beta_c-\beta_h}.
\end{align*}
This implies that
\begin{align}\label{eta inv}
	\eta_\mathrm{max}^{-1}=\lim_{g\rightarrow 0^+}(\eta^\textup{mac})^{-1}(\beta_c-g)=1+\frac{\beta_{h}}{\beta_{c}-\beta_{h}}=\frac{\beta_{c}}{\beta_{c}-\beta_{h}}
\end{align}
and hence $\eta_\mathrm{max}=1-\frac{\beta_{h}}{\beta_{c}}$.
\end{proof}

\subsection{Maximum efficiency for near perfect work is still Carnot efficiency}\label{sub:macroeffNPW}
In this section, we show that even while allowing a non-zero failure probability $\varepsilon >0$ in the near perfect work scenario, the maximum achievable efficiency is still the Carnot efficiency. It is worth noting that this result is also important later, as an upper bound to maximum efficiency in the nanoscopic regime. We first prove it in Lemma \ref{lem:cannot do better than carnot with porb failuer} for the case where the final state of the battery is fixed as in Eq.~\eqref{eq:battery final state}. Then later, we show in Lemma \ref{lem:CEmaximum_generalbatt} that Carnot efficiency  is still the maximum, even if we allow a more general final battery state. Before we present the proof, it is useful for the reader to recall the definition of near perfect work (Def.~\ref{def:near perfect work}) and quasi-static heat engines (Def.~\ref{def:quasi static}).

\begin{lemma}\label{lem:cannot do better than carnot with porb failuer}
Consider all heat engines which extract near perfect work (see Def. \ref{def:near perfect work}).
Then according to the macroscopic second law of thermodyanmics, the maximum efficiency of a heat engine, $\eta_\textup{max}$ is the Carnot efficiency
\be 
\eta_\textup{max}=\sup_{\rho_\cold^1\in\mathcal{S}}\eta^\textup{mac}(\rho_\cold^1)=1-\frac{\beta_h}{\beta_c},
\ee
and the supremum is achieved for quasi-static heat engines (see Def. \eqref{def:quasi static} and Eq. \eqref{def:quasi static eff mac}).

\end{lemma}
\begin{proof}
The ideas in this proof are very similar to that of Section \ref{sec:Maximum efficiency}, and the main complication comes from proving that even if we allow $\varepsilon >0$, as long as $\Delta S/W_\textup{ext}$ is arbitrarily small, the maximum efficiency cannot surpass the Carnot efficiency.

Let us begin by establishing the relevant quantities for near perfect work extraction. The amount of work extractable from the heat engine, when we have a probability of failure, according to the standard free energy can be obtained by solving Eq. \eqref{eq:Helmholts as D_1}. 
We thus have that the maximum $W_{\rm ext}$ is
\be\label{eq:W ext std}
W_\textup{ext}= \beta_h^{-1}(1-\varepsilon)^{-1} \left[ D(\tau_{\beta_c}\| \tau_{\beta_h})-D(\rho_\cold^1\|\tau_{\beta_h})+\Delta S\right],
\ee
where $\Delta S$ is defined in Eq. \eqref{eq:von neu}. 

Before we continue with the analysis, we will note a trivial consequence of Eq. \eqref{eq:W ext std}. Condition 1) in Def \ref{def:near perfect work} implies that $(1-\varepsilon)^{-1}$ is upper bounded. The terms in square brackets in Eq. \eqref{eq:W ext std} are also clearly upper bounded for finite $\beta_c,\beta_h$. Hence $W_\textup{ext}$ is bounded from above. $\Delta S$ is solely a function of $\varepsilon$ and only approaches zero in the limits $\varepsilon\rightarrow 0^+$, $\varepsilon\rightarrow 1^-$; and $\varepsilon\rightarrow 1^-$ is forbidden by 1) in Def \ref{def:near perfect work}. Thus if 1) and 2) in Def \ref{def:near perfect work} are satisfied,
\be \label{eq: equiv to def near perfect work}
\lim_{\varepsilon\rightarrow 0^+}\frac{\Delta S}{W_\textup{ext}}=0.
\ee
In turn, if Eq. \eqref{eq: equiv to def near perfect work} is satisfied, then we have near perfect work by Def. \ref{def:near perfect work}. Thus Eq. \eqref{eq: equiv to def near perfect work} is satisfied iff we have near perfect work. We will use this result later in the proof.

Extracting a positive amount of near perfect work implies that we can rule out all states $\rho_\cold^1$ such that $D(\tau_{\beta_c}\| \tau_{\beta_h})\leq D(\rho_\cold^1\|\tau_{\beta_h})$ from the analysis. This can be proven by contradiction: if $D(\tau_{\beta_c}\| \tau_{\beta_h})\leq D(\rho_\cold^1\|\tau_{\beta_h})$, then from Eq. \eqref{eq:W ext std} $\beta_h W_\textup{ext}\leq \Delta S/(1-\varepsilon)$ and together with 2) in Def \ref{def:near perfect work} this would imply 
\be\label{eq:not possible}
0<\beta_h (1-\varepsilon)\leq \frac{\Delta S}{W_\textup{ext}}<p.
\ee
However, since from 1) Def. \ref{def:near perfect work} we have $\varepsilon\leq l$, Eq. \eqref{eq:not possible} cannot be satisfied for all $p>0$, leading to a contradiction. 

From Eq. \eqref{eq:eff explicit function of cold bath} we have
\be\label{eq:inf rho c}
\eta_\textup{max}^{-1}=  1-\varepsilon+ \inf_{\rho_\cold^1\in\mathcal{S}}\frac{\Delta C}{W_\textup{ext}}= (1-\varepsilon)\cdot \left[1+ \frac{\beta_h \Delta C}{ D(\tau_{\beta_c}\| \tau_{\beta_h})-D(\rho_\cold^1\|\tau_{\beta_h}) +\Delta S} \right],
\ee
where $\Delta C=\Delta C(\rho_\cold^1)$ and is defined in Eq. \eqref{eq:delta C def}. 

Firstly, let us show that with a similar analysis as shown in Lemma \ref{lemma: optimum eff is for themal state}, the maximum efficiency occurs when $\rho_\cold^1$ is a thermal state. From Eq.~\eqref{eq:inf rho c}, we have
\begin{align}
\eta_{\rm max}^{-1} &= (1-\varepsilon) \left[1+\beta_h \inf_{\rho_\cold^1\in\mathcal{S}}\frac{\Delta C}{D(\tau_{\beta_c}\| \tau_{\beta_h})-D(\rho_\cold^1\|\tau_{\beta_h}) +\Delta S}\right]\label{eq:min0}\\
&= (1-\varepsilon) \left[1+\beta_h \inf_{A>0}\frac{A}{D(\tau_{\beta_c}\| \tau_{\beta_h})- B(A) +\Delta S}\right]\label{eq:min1}
\end{align}
where 
\begin{align}\label{eq:min2}
B(A) = \displaystyle\inf_{\substack{\rho_\cold^1 \in\mathcal{S} \\ \tr (\hat H_\cold\rho_\cold^1)-\tr(\hat H_\cold\tau_{\beta_c})=A}} D(\rho_\cold^1\|\tau_{\beta_h}).
\end{align}
We can split this minimization problem to Eqs.~\eqref{eq:min1} and \eqref{eq:min2} because $D(\tau_{\beta_c}\| \tau_{\beta_h})$ and $\Delta S$ do not depend on the variable $\rho_\cold^1$. Furthermore, when $\rho_\cold^1$ is a thermal state of inverse temperature $\beta_f$, we have seen in the beginning of the proof in Theorem \ref{theorem:classical CE} that for $W_\textup{ext}>0$, $\beta_f<\beta_c$. This implies that the variable $A = \Delta C = \tr (\hat H_\cold \tau_{\beta_f})-\tr (\hat H_\cold \tau_{\beta_c}) >0$.

By Lemma \ref{lem:minimizingD1}, for any fixed $A>0$ we conclude that the infimum in Eq.~\eqref{eq:min2} is achieved \emph{uniquely} when $\rho_\cold^1$ is a thermal state. Therefore, our optimization problem is simplified to optimization over final temperatures $\beta_f$ (or $g=\beta_c-\beta_f$),
\begin{align}
\eta_{\rm max}^{-1} 
&= (1-\varepsilon) \cdot \left[1+\beta_h \inf_{\substack{\beta_f\\\Delta C>0}}\frac{\Delta C}{D(\tau_{\beta_c}\| \tau_{\beta_h})- D(\tau_{\beta_f}\| \tau_{\beta_h}) +\Delta S}\right]\label{eq:min3}
\end{align}

Consider cases of $\beta_f$, where $D(\tau_{\beta_c}\| \tau_{\beta_h})- D(\tau_{\beta_f}\| \tau_{\beta_h})$ is non-vanishing (finite), i.e. which are not \emph{quasi-static}. Note that this always corresponds to extracting near perfect work, since when $\varepsilon\rightarrow 0^+$, we have $\varepsilon, \Delta S\rightarrow 0$ and these contributions disappear from Eq.~\eqref{eq:min3}. However, by Lemma \ref{lemma: optimum eff is for themal state} we also know that the infimum over $\beta_f$ occurs uniquely at the quasi-static limit, when $g\rightarrow 0^+$. 

What remains, is then to consider the quasi-static heat engine, namely the limit $g\rightarrow 0^+$. Extracting near perfect work in this case corresponds to requiring that $\lim_{g\rightarrow 0^+} \frac{\Delta S}{W_\textup{ext}} =0$, where $\varepsilon=\varepsilon(g)$ and $\lim_{g\rightarrow 0^+} \varepsilon(g)=0$. Equivalently 
\begin{equation}\label{eq:limitW/S}
\lim_{g\rightarrow 0^+} \frac{W_\textup{ext}}{\Delta S} =\infty.
\end{equation}
Substituting Eq.~\eqref{eq:W ext std} into Eq.~\eqref{eq:limitW/S},
\begin{align}
\lim_{g\rightarrow 0^+} (1-\varepsilon(g))^{-1} \left[1+\frac{D(\tau_{\beta_c}\| \tau_{\beta_h})- D(\tau_{\beta_f}\| \tau_{\beta_h})}{\Delta S}\right] =\infty
\end{align}
which implies that $\displaystyle\lim_{g\rightarrow 0^+} \frac{D(\tau_{\beta_c}\| \tau_{\beta_h})- D(\tau_{\beta_f}\| \tau_{\beta_h})}{\Delta S}=\infty,$ or equivalently, 
\begin{align}\label{eq:limitS/D-D}
\lim_{\varepsilon\rightarrow 0^+}\lim_{g\rightarrow 0^+} \frac{\Delta S}{D(\tau_{\beta_c}\| \tau_{\beta_h})- D(\tau_{\beta_f}\| \tau_{\beta_h})}=0.
\end{align}

Finally, we evaluate the inverse efficiency at the quasi-static limit,
\begin{align}
\eta^{-1} &= \lim_{g\rightarrow 0^+} (1-\varepsilon(g)) \cdot \left[1+\beta_h \frac{\Delta C}{D(\tau_{\beta_c}\| \tau_{\beta_h})- D(\tau_{\beta_f}\| \tau_{\beta_h}) +\Delta S}\right]\\
&= 1+\beta_h \lim_{g\rightarrow 0^+}\frac{\Delta C}{D(\tau_{\beta_c}\| \tau_{\beta_h})- D(\tau_{\beta_f}\| \tau_{\beta_h}) +\Delta S}\\
&= 1+\beta_h \lim_{g\rightarrow 0^+}\frac{\Delta C}{\left[D(\tau_{\beta_c}\| \tau_{\beta_h})- D(\tau_{\beta_f}\| \tau_{\beta_h})\right]}\cdot \left(1+\frac{\Delta S}{D(\tau_{\beta_c}\| \tau_{\beta_h})- D(\tau_{\beta_f}\| \tau_{\beta_h})}\right)^{-1}\label{eq:3rdlasteq}\\
&= 1+\beta_h \lim_{g\rightarrow 0^+} \frac{d\Delta C (\tau_{\beta_f})/dg}{d D(\tau_{\beta_f}\|\tau_{\beta_h})/dg}\label{eq:2ndlasteq}\\
&= 1- \frac{\beta_h}{\beta_h-\beta_c}\label{eq:lasteq},
\end{align}
where from Eq.~\eqref{eq:3rdlasteq} to \eqref{eq:2ndlasteq}, we make use of Eq.~\eqref{eq:limitS/D-D} : the second term within the limit is simply 1, and the first term depends only on $g$, which we can obtain Eq.~\eqref{eq:2ndlasteq} by invoking the L'H\^ospital rule. 
The last equality in Eq.~\eqref{eq:lasteq} follows directly from the identities we derived for $\frac{dW_\textup{ext}}{d\beta_f}$ and $\frac{d\Delta C}{d\beta_f}$ in Corollary \ref{cor:specialid},
\begin{align}
\frac{d\Delta C}{dg} &= -\frac{d\Delta C}{d\beta_f} = - {\rm var} (\hat H_\cold)_{\beta_f}\\
\frac{dD(\tau_{\beta_f}\|\tau_{\beta_h})}{dg} &= -\frac{dD(\tau_{\beta_f}\|\tau_{\beta_h})}{d\beta_f} = \beta_h \frac{dW_\textup{ext}}{d\beta_f} = (\beta_h-\beta_f) {\rm var} (\hat H_\cold)_{\beta_f},
\end{align} 
 while in the limit $g\rightarrow 0,~\beta_f = \beta_c$.

Finally, we now see that the quasi-static efficiency is
\begin{equation}
\eta = \left( \frac{\beta_h-\beta_c-\beta_h}{\beta_h-\beta_c}\right)^{-1} = \frac{\beta_c-\beta_h}{\beta_c} = 1 - \frac{\beta_h}{\beta_c}
\end{equation}
which is exactly the Carnot efficiency. 
\end{proof}

Later, in Section \ref{A more general final battery state} we will need Lemma \ref{lemma: optimum eff is for themal state} to hold in a more general scenario, i.e. instead of the final battery state being $\rho_\batt^1 = (1-\varepsilon)\ketbra{E_k}{E_k}_\batt + \varepsilon \ketbra{E_j}{E_j}_\batt$, we want to allow the final battery state to be any energy block-diagonal state with trace distance $\varepsilon$. Next we state and prove this generalized lemma.

\begin{lemma}\label{lem:CEmaximum_generalbatt}
Consider all heat engines which extract near perfect work (see Definition \ref{def:near perfect work}), but allowing for any final battery state with a trace distance $\varepsilon$ to the ideal final pure state $\ketbra{E_k}{E_k}_\batt$. 
Then according to the macroscopic second law of thermodynamics, the maximum efficiency of a heat engine, $\eta_\textup{max}$ is the Carnot efficiency
\be 
\eta_\textup{max}=\sup_{\rho_\cold^1\in\mathcal{S}}\eta^\textup{mac}(\rho_\cold^1)=1-\frac{\beta_h}{\beta_c},
\ee
and when the final state of the cold bath $\rho^1_\cold$ is thermal, the supremum is \textup{only} achieved for quasi-static heat engines (see Def. \eqref{def:quasi static} and Eq. \eqref{def:quasi static eff mac}).
\end{lemma}
\begin{proof}
Firstly, let us note that since the initial state $\rho_\CW^0$ which we start out with is energy block-diagonal, the final state has to also be block-diagonal. Therefore, given the product structure between the cold bath and battery, it is sufficient to consider the case when the final battery state is energy block-diagonal. Next, let us note that any final state $\rho_\batt^2$ which is energy block-diagonal, and has trace distance $\varepsilon$ with $\ketbra{E_k}{E_k}_\batt$ can be written as,
\begin{equation}\label{eq:lem6proof_generalbatt}
\rho_\batt^2 = (1-\varepsilon)\ketbra{E_k}{E_k}_\batt + \varepsilon\rho_\batt^\textup{junk}, ~\textup{where } \rho_\batt^\textup{junk} = \sum_{i} p_i \ketbra{E_i}{E_i}_\batt, ~~\sum_i p_i = 1 ~\textup{and}~ p_k = 0.
\end{equation}
Next, one can calculate $W_\textup{ext}$ given by the standard free energy	condition, i.e.
\begin{equation}
F(\tau_{\beta_c})+F(\rho_\batt^0) \geq F(\rho_\cold^1)+F(\rho_\batt^2). 
\end{equation}
Using the identity $F(\rho) = \tr(\hat H\rho) - \beta^{-1} S(\rho)$, we have that
\begin{align}
F(\tau_{\beta_c})+ E_j &\geq F(\rho_\cold^1)+ (1-\varepsilon) E_k + \varepsilon \tr(\hat H_\batt \rho_\batt^\textup{junk}) - \beta^{-1}_h S(\rho_\batt^2).
\end{align}
Substituting $W_\textup{ext} = E_k - E_j$, and rearranging terms, we have
\begin{equation}
(1-\varepsilon)W_\textup{ext}  \leq F(\tau_{\beta_c})-F(\rho_\cold^1)+ \beta^{-1}_h\Delta S - \varepsilon [\tr(\hat H_\batt \rho_\batt^\textup{junk})-E_j].
\end{equation}
Finally, by using the identity (in Eq.~\eqref{eq:helmholtzF}) that $F(\rho) = \beta^{-1}_h [D(\rho\|\tau_{\beta_h})-\ln Z_{\beta_h}]$, the maximum amount of extractable work is given by 
\begin{equation}
W_\textup{ext}  = (1-\varepsilon)^{-1} \beta^{-1}_h\cdot [ D(\tau_{\beta_c}\|\tau_{\beta_h})-D(\rho_\cold^1\|\tau_{\beta_h})+\Delta S - \varepsilon\tilde E],
\end{equation}
where $\tilde E =  \tr(\hat H_\batt \rho_\batt^\textup{junk})-E_j$.

Following the steps in Lemma \ref{lem:cannot do better than carnot with porb failuer}, in particular the derivations in Eq.~\eqref{eq:min0} and \eqref{eq:min1}, we have
\begin{align}\label{eq:etamaxn1}
\eta_{\rm max}^{-1} 
&= (1-\varepsilon) \cdot \left[1+\beta_h \inf_{\substack{\beta_f\\\Delta C>0}}\frac{\Delta C}{D(\tau_{\beta_c}\| \tau_{\beta_h})- D(\tau_{\beta_f}\| \tau_{\beta_h}) +\Delta S - \varepsilon\tilde E}\right].
\end{align}
To show Eq.~\eqref{eq:etamaxn1} gives the Carnot efficiency, we show that 1) for non quasi-static cases where $\beta_f < \beta_c$, Carnot efficiency is not attained, and 2) in the quasi-static limit, Carnot efficiency is attained.

Let us first consider the case of extracting a non-vanishing amount of near perfect work, i.e. for all cases where $\beta_f < \beta_c$. Then near perfect work, by Def. \ref{def:near perfect work}, corresponds to the limit $\varepsilon\rightarrow 0$,
\begin{align}
\eta^{-1} 
&= \lim_{\varepsilon\rightarrow 0}~(1-\varepsilon) \cdot \left[1+\beta_h \frac{\Delta C}{D(\tau_{\beta_c}\| \tau_{\beta_h})- D(\tau_{\beta_f}\| \tau_{\beta_h}) +\Delta S - \varepsilon\tilde E}\right]\\
&= 1 + \beta_h \frac{\Delta C}{D(\tau_{\beta_c}\| \tau_{\beta_h})- D(\tau_{\beta_f}\| \tau_{\beta_h})}.
\end{align}
In this limit, all terms involving $\varepsilon$ vanish, and the inverse efficiency has the same expression as the efficiency for perfect work. We already know from Lemma \ref{lemma:optimaleff is for infinitesimal temp change} that the infimum over $\beta_f$ cannot be obtained in this regime, since the inverse efficiency is strictly decreasing with $\beta_f$. 

Therefore, again we are left with analyzing the quasi-static limit for this problem. Following the derivation in Eq.~\eqref{eq:3rdlasteq} for the quasi-static limit, we obtain
\begin{align}\label{eq:generalbattstate_opteff}
\eta^{-1}_{\rm max} &= 1+\beta_h \lim_{g\rightarrow 0^+}\frac{\Delta C}{\left[D(\tau_{\beta_c}\| \tau_{\beta_h})- D(\tau_{\beta_f}\| \tau_{\beta_h})\right]}\cdot \left(1+\frac{\Delta S-\varepsilon\tilde E}{D(\tau_{\beta_c}\| \tau_{\beta_h})- D(\tau_{\beta_f}\| \tau_{\beta_h})}\right)^{-1},\end{align}
where $\varepsilon=\varepsilon(g)$ and note that requiring near perfect work implies that 
\begin{equation}\label{eq:lim_DeltaSoverWexteq0}
\lim_{g\rightarrow 0^+}\frac{\Delta S}{D(\tau_{\beta_c}\| \tau_{\beta_h})- D(\tau_{\beta_f}\| \tau_{\beta_h})} = 0.
\end{equation}
Next, we observe the relationship between $\varepsilon$ and $\Delta S$, in the regime where $\varepsilon$ is small. Given any $\varepsilon >0$ denoting the trace distance $d(\rho_\batt^2, \ketbra{E_k}{E_k}_\batt)=\varepsilon$, the smallest amount of entropy that can be produced corresponds to $\Delta S = h_2 (\varepsilon)$. This is because if we try to distribute the weight $\varepsilon$ over more energy eigenvalues, then by majorization the entropy only increases. But we also know that $\varepsilon \leq h_2 (\varepsilon)$ for small values of $\varepsilon$, in particular over the regime $\varepsilon\in [0,\frac{1}{2}]$. Therefore, we have that in this regime, $\varepsilon \leq h_2 (\varepsilon) \leq \Delta S$ holds.
Therefore, we also know that
\begin{equation}\label{eq:eps/deltaD}
\lim_{g\rightarrow 0^+} \frac{\varepsilon\tilde E}{D(\tau_{\beta_c}\| \tau_{\beta_h})- D(\tau_{\beta_f}\| \tau_{\beta_h})}= 0,
\end{equation} 
where $\varepsilon=\varepsilon(g)$. Plugging Eqns.~\eqref{eq:lim_DeltaSoverWexteq0} and \eqref{eq:eps/deltaD} into Eq.~\eqref{eq:generalbattstate_opteff}, we have that the quasi-static efficiency is $\eta = 1 - \frac{\beta_h}{\beta_c}$.
\end{proof}


\section{Efficiency of a nanoscopic quantum heat engine}\label{section:Nano/quantum scale heat engine at maximum efficiency}
In this section, we will be applying the conditions for state transitions for nanoscale systems, as detailed in Section \ref{sub:nano}. The reader will see that due to these extra constraints from the generalized free energies, the fundamental limitations on efficiency will differ greatly from those observed in Section \ref{section:standardthermo}.

Firstly, in Section \ref{subsec:impossible}, we show that the extraction of a positive amount of perfect work is \emph{impossible} using the setup. In Section \ref{subsec:highcertainty}, we show that this can be resolved by considering near perfect work instead. Then we find that:
\begin{itemize}
\item [(1)] The maximum achievable efficiency is still the Carnot efficiency. This is proven in Section \ref{upper bound carnot eff}.
\item [(2)] However, in the case of quasi-static heat engines, the Carnot efficiency \emph{cannot} be achieved for all cold bath Hamiltonians. \textbf{This is our main result, which is stated in Theorem \ref{th:EfficiencyRandomEnergyGaps}, found in Section \ref{subsub:main}.} The results in Section \ref{subsect:LBeff} and \ref{subsect:solvinginfimum} are more technical proofs, that pave the way for deriving this main result.
\end{itemize}

\subsection{Impossibility of extracting perfect work}\label{subsec:impossible}
We will first show that with the general setup as described in Section \ref{section:The setting}, no perfect work can ever be extracted. By this we mean that whenever $\varepsilon$ as defined in Eq.\eqref{eq:battery final state} equals zero, then for any value of $W_{\rm ext}>0$, and for any final state $\rho_\stcold^1$, the transition $|E_j\ra\la E_j|_{\battery\!}\otimes\tau_\cold^0\rightarrow|E_k\ra\la E_k|_{\battery\!}\otimes\rho_\stcold^1$ is not possible. Intuitively speaking, this occurs because the cold bath is initially in a state of full rank. Since thermal operations cannot decrease the rank of the system, therefore the final state of the cold bath $\rho_\stcold^1$ must also be of full rank. By directly solving Eq.~\eqref{2nd law eq}, for $W_{\rm ext}$, we can find the an equation governing the amount of extractable work. To solve Eq.~\eqref{2nd law eq} for $W_{\rm ext}$, we start by first using Eq. \eqref{eq:generalfreeenergy} to find an expression solely in term of the R\'enyi divergences and the fact that the bipartite thermal states are product states,
\begin{align}\label{D bee D}
D_\alpha(\tau_\cold^0\otimes\rho^0_\battery\|\tau_\CW^h)\geq D_\alpha(\rho_\stcold^1\otimes\rho^1_\battery\|\tau_\CW^h)\quad \forall \alpha\geq 0.
\end{align}
Now using the additivity of the R\'enyi divergences, from \eqref{D bee D} it follows 
\begin{equation}
 D_\alpha(\rho^1_\battery\|\tau_\batt^h) - D_\alpha(\rho^0_\battery\|\tau_\batt^h) \leq D_\alpha(\tau_\cold^0\|\tau_\cold^h)- D_\alpha(\rho_\cold^1\|\tau_\cold^h)\quad \forall \alpha\geq 0.
\end{equation}
By directly solving the L.H.S., we find that $W_{\rm ext}$ must satisfy
\begin{align}
W_{\rm ext}\leq kT_\hot  \left[ D_\alpha(\tau^{0}_\cold\|\tau^h_\cold)-D_\alpha(\rho^{1}_\cold\|\tau^h_\cold) \right]\quad \forall \alpha\geq 0.
\end{align}
Hence
\begin{align}\label{W general not prod}
	W_{\rm ext}\leq kT_\hot \inf_{\alpha \geq 0} \left[ D_\alpha(\tau^{0}_\cold\|\tau^h_\cold)-D_\alpha(\rho^{1}_\cold\|\tau^h_\cold) \right],
\end{align}
where $\tau^h_\cold$ is the thermal state of the cold bath (according to the cold bath Hamiltonian $\hat H_\cold$), at temperature $T_{\hot}$ (since the surrounding hot bath is of temperature $T_\hot$). However from Eq.~\eqref{eq:generalfreeenergy}, $D_0 (\tau_\cold^0\|\tau_\cold^h) = D_0 (\rho_\stcold^1\|\tau_\cold^h)$. Therefore according to Eq. ~\eqref{W general not prod}, the amount of work extractable satisfies $W_{\rm ext}\leq 0$.

We phrase this with more rigor in the following Lemmas \ref{lem: lem 0} and \ref{lem: lem 1}, which proves that for perfect work, $W_\textup{ext}>0$ is impossible. The proof holds for general initial states $\rho_\cold^0$ of full rank, in particular, they need not even be diagonal in the energy eigenbasis.
\begin{lemma}\label{lem: lem 0}
For any $W_{\rm ext}>0$, consider the Hamiltonian $\hat H_\battery$ given by Eq. \eqref{eq:W general 0}. Then for any inverse temperature $\beta_h>0$, the thermal state $\tau_\battery^h = \frac{1}{\tr[e^{-\beta_{h}\hat H_\battery}]}	e^{-\beta_{h}\hat H_\battery}$ satisfies
\begin{equation}
\tr \left[\left(|E_j\ra\la E_j|_\textup{W}-|E_k\ra\la E_k|_\textup{W}\right)\tau_\battery^h\right] > 0.
\end{equation}
\begin{proof}
Follows directly from the definitions. Since $W_{\rm ext} >0$, we know that $E_j^\battery < E_k^\battery$. Evaluating the quantity above gives $\frac{1}{\tr[e^{-\beta_{h}\hat H_\battery}]}\cdot \left(e^{-\beta_h E_j^\battery}-e^{-\beta_h E_k^\battery}\right) > 0$.
\end{proof}

\end{lemma}

\begin{lemma}\label{lem: lem 1}
Consider any general quantum state $\rho^0_\cold$ of full rank. Then for any $\rho^1_\cold$, the transition from $\rho^0_\cold\otimes\rho^0_\battery\rightarrow\rho^1_\cold\otimes\rho^1_\battery$ is not possible via catalytic thermal operations if
\begin{equation}
\tr \left[\left(\Pi_{\rho^0_\battery}-\Pi_{\rho_\battery^1} \right)\tau_\battery^h\right]>0,
\end{equation}
where $\Pi_{\rho}$ is the projector onto the  support of state $\rho$, and $\tau_\battery^h$ is the thermal state of the battery at the initial hot bath temperature. 
\end{lemma}

\begin{proof}
One can show this by invoking the quantum second law for $\alpha=0$ \cite{2ndlaw}, which says that if $\rhoin\rightarrow\rhoout$ is possible via catalytic thermal operations, then
\begin{equation}
 D_0 (\rhoin\|\tau)\geq  D_0(\rhoout\|\tau),
\end{equation}
where $\tau$ is the thermal state of the system at bath temperature, and 
\begin{eqnarray}\label{eq:F_0}
	& D_0(\rho\|\sigma)=\displaystyle\lim_{\alpha \rightarrow 0^+}\frac{1}{\alpha-1}\ln \tr[\rho^\alpha \sigma^{1-\alpha}]=-\ln\tr[\Pi_\rho \sigma],
\end{eqnarray}
is defined for arbitrary quantum states $\rho,\sigma$.
Applying this law with $\rhoin=\rho^0_\battery\otimes\rho_\stcold^0$ and $\rhoout=\rho^1_\battery\otimes\rho_\stcold^1$, we arrive at
\begin{equation}\label{lem 1 1}
	 D_0(\rho^0_\battery\|\tau^h_\battery)- D_0(\rho^1_\battery\|\tau^h_\battery)\geq  D_0(\rho_\stcold^1\|\tau_\cold^h)- D_0(\rho_\stcold^0\|\tau_\cold^h),
\end{equation}
where $\tau_\cold^h$ and $\tau^h_\battery$ are thermal states of the cold bath and battery at the temperature of surrounding hot bath ($T_\hot$) respectively.
Since $\rho_\stcold^0$ have full rank, and since  $\tau_\cold^h$ is normalized, therefore according to Eq.~ \eqref{eq:F_0}, $ D_0(\rho_\stcold^0\|\tau_\cold^h)=0$. Furthermore, since the $\alpha-$R{\'e}nyi divergence $ D_0$ is non-negative, therefore the r.h.s. of Eq.~\eqref{lem 1 1} is lower bounded by 0. Thus, we have 
\begin{equation}\label{eq:condition}
\tr [(\Pi_{\rho^0_\battery}-\Pi_{\rho^1_\battery} )\tau_\battery^h]\leq 0.
\end{equation}
Since this is a necessary condition for state transformations, we arrive at the conclusion that: when Eq.~ \eqref{eq:condition} is violated, state transformations are not possible. But from Lemma \ref{lem: lem 0}, any type of perfect work extraction violates Eq.~\eqref{eq:condition}. Therefore, in this setting, perfect work extraction is always impossible.\\
\end{proof}

To summarize, Lemma \ref{lem: lem 1} implies that if the initial state of the cold bath is thermal, and therefore of full rank, then any work extraction scheme via thermal operations bringing $\rho^0_\battery=|j\ra\la j|_\textup{W}$ to $\rho^1_\battery=|k\ra\la k|_\textup{W}$ where $W_{\rm ext} = E_k^\battery - E_j^\battery >0$ is not possible. In general, we see that if $\Pi_{\rho^0_\battery}\neq\Pi_{\rho^1_\battery},$ then when transition $\rho^0_\battery$ to $\rho^1_\battery$ is possible, transition $\rho^1_\battery$ to $\rho^0_\battery$ is not.
Consequentially, we will have to consider near perfect work at the nano regime.

\subsection{Efficiency for extracting near perfect work}\label{subsec:highcertainty}
As we have just seen in the previous Section \ref{subsec:impossible}, we cannot extract perfect work. Due to the impossibility result, we consider the relaxation of extracting near perfect work in the nanoscale setting. 

\begin{itemize}
\item We begin by evaluating the expression for efficiency according to the nanoscopic laws of thermodynamics, given a final state of the cold bath, and comparing it to the expression according to macroscopic laws of thermodynamics. This is done in Sections \ref{subsub:highcertainty} and \ref{upper bound carnot eff}, and the relation between two efficiencies are summarized in Eq.~\eqref{eq:std eff def}. Since the nanoscopic efficiency is always smaller than the macroscopic efficiency, which attains Carnot efficiency only in the limit $\Delta C \rightarrow 0$, it will be a necessary condition to consider this limit if we want to achieve the Carnot efficiency, when considering nanoscopic laws of thermodynamics.
\item We analyze the quasi-static regime, focusing on the special case where the cold bath consists of $n$ qubits. Since the quasi-static limit corresponds to the case of small $g>0$, and $\varepsilon$ also has to be arbitrarily small for near perfect work extraction, we perform Taylor expansion of the analytical expressions for $W_\textup{ext}$ and $\Delta C$ w.r.t. $g$ and $\varepsilon$. This is done in Section \ref{subsect:evaluatingWextquasi}.
\item In Section \ref{subsect:LBeff}, we identify how to choose $\varepsilon (g)$ such that it corresponds to drawing near perfect work in the quasi-static limit. We first begin by observing that any continuous function $\varepsilon (g)$ that vanishes in the limit $g\rightarrow 0$ can be characterized with a real-valued parameter $\kappabar$ that determines how quickly $\varepsilon$ goes to zero. This is shown in Lemma \ref{lemma:existence kappabar}. In Lemma \ref{lemma: bar gamma}, we show that near perfect work is drawn only if $\kappabar\in [0,1]$.
\item Lemma \ref{lemma: bar gamma} gives us the analytical expression and minimization range in order to evaluate $W_\textup{ext}$, according to Eq.~\eqref{eq:work extractable to first order in g}. In Section \ref{subsect:solvinginfimum}, we show how one can evaluate this optimization problem, by comparing the stationary points and endpoints of the function $\frac{\alpha B_\alpha}{\alpha-1}$ that gives the leading term in Eq.~\eqref{eq:work extractable to first order in g}. Lemma \ref{zeros lemma} proves a technical property of the first derivative of this function. Using it, we prove in Lemma \ref{lim lim lem} that one can always choose $\varepsilon (g)$ with some $\kappabar <1$ such that the infimum of $\frac{\alpha B_\alpha}{\alpha-1}$ is obtained at either $\alpha=\kappabar$ or $\alpha\rightarrow\infty$.
\item Finally, in Section \ref{subsub:main}, we use the results in Section \ref{subsect:solvinginfimum} regarding the evaluation of $W_\textup{ext}$ to find the efficiency in the quasi-static limit.
\end{itemize}

\subsubsection{An explicit expression for $W_\textup{ext}$}\label{subsub:highcertainty}
Our first task is to work out an explicit expression for $W_\textup{ext}$ depending on the initial and final states of the cold bath, $\varepsilon$ and hot bath (inverse) temperature $\beta_h$. Such as expression is found by applying the generalized second laws as detailed in Section \ref{sub:nano}.
\begin{lemma}\label{eq:explicit sol for W ext}
Consider the transition 
\begin{equation}\label{eq:transition no broduct bath}
\tau_\cold^0\otimes\rho_\textup{W}^0\rightarrow \rho^1_\cold\otimes\rho_\textup{W}^1\quad \textup{with}\quad\varepsilon>0.
\end{equation}
where $\rho_\textup{W}^0$ and $\rho_\textup{W}^1$ are defined in Eqs. \eqref{eq:battery initial state}, \eqref{eq:battery final state} respectively. Let $W_{\rm ext}$ denote the \textup{maximum} possible value such that Eq.~\eqref{eq:transition no broduct bath} is possible via catalytic thermal operations, with a thermal bath of inverse temperature $\beta_h$. Let $\beta_c >\beta_h$. 
Then the final state $\rho_\cold^1=\sum_i p_i'|E_i\ra\la E_i|_\cold\;$ is block-diagonal in the energy eigenbasis, and
\begin{align}
W_{\rm ext} &= \inf_{\alpha\geq 0} ~W_\alpha,\label{eq:definitions1 no broduct bath}\\
W_\alpha &= \frac{1}{\beta_h (\alpha-1)} [\ln (A-\varepsilon^\alpha)-\alpha\ln (1-\varepsilon)],\label{eq:definitions2 no broduct bath}\\
A&= \frac{\sum_i p_i^\alpha q_i^{1-\alpha}}{\sum_i p_i'^\alpha q_i^{1-\alpha}},\label{eq:definitions3 no broduct bath}
\end{align}
where $p_i= \frac{e^{-\beta_c E_i}}{Z_{\beta_C}}$, $q_i= \frac{e^{-\beta_h E_i}}{Z_{\beta_h}}$, and $p_i'$ are the probability amplitudes of state $\rho_\cold^1$ when written in the energy eigenbasis of $\hat H_\cold$. 
The quantities $W_1$ and $W_\infty$ are defined by taking the limit $\alpha\rightarrow 1, +\infty$ respectively.
\end{lemma}
\begin{proof}
Eq. \eqref{2nd law eq} is necessary and sufficient for Eq. \eqref{eq:transition no broduct bath} to be satisfied. We can apply the additivity property of the R{\'e}nyi divergence, to Eq. \eqref{2nd law eq} to find
\begin{equation}\label{eq:dalphas  no broduct bath}
D_\alpha (\rho_\battery^0\|\tau_\textup{W}^h) +  D_\alpha (\tau_{\beta_c}\|\tau_{\beta_h}) \geq D_\alpha (\rho_\battery^1\|\tau_\textup{W}^h) +  D_\alpha (\rho_\cold^1\|\tau_{\beta_h}),
\end{equation}
where $\tau_{\beta_h}$ and $\tau_\textup{W}^h$ are the thermal states with Hamiltonians $\hat H_\cold$ and $\hat H_\battery$ respectively, at inverse temperature $\beta_h$. From Eq. \eqref{eq:dalphas  no broduct bath} it follows,
\begin{align}\label{eq:dalphas  no broduct bath 2}
\frac{-1}{\alpha-1} \ln\left[ \varepsilon^\alpha + (1-\varepsilon)^\alpha e^{-\beta (E_k-E_j)(1-\alpha)} \right] &\geq  D_\alpha (\rho_\cold^1\|\tau_{\beta_h}) - D_\alpha (\tau_{\beta_c}\|\tau_{\beta_h})\\
&= \frac{-1}{\alpha-1} \ln A.
\end{align}
By examining the three cases $\alpha<1$, $\alpha>1$ separately, have find
\begin{align}
\begin{split}
 \varepsilon^\alpha + (1-\varepsilon)^\alpha e^{-\beta (E_k-E_j)(1-\alpha)} \geq A   &\quad \text{ if }\, \alpha<1\\
 \varepsilon^\alpha + (1-\varepsilon)^\alpha e^{-\beta (E_k-E_j)(1-\alpha)} \leq A 	&\quad\text{ if }\, \alpha>1
 \end{split}\label{eq:alpha 1 case}
\end{align}
where note that $A$ is independent of $E_k$ and $E_j$. Now note that the largest value of $E_k-E_j$ in Eq. \eqref{eq:alpha 1 case}, is obtained when the inequalities hold with equality. Thus, since these equations have to hold for all $\alpha\in[0,\infty]$, the largest amount of work we can extract corresponds to the value of $E_k-E_j$ for which Eq. \eqref{eq:alpha 1 case} holds for all $\alpha\in[0,\infty]$, with an equality for at least one particular value of $\alpha\in[0,\infty]$. In other words $W_{\rm ext}$ is given by Eq. \eqref{eq:definitions1 no broduct bath}, and $W_\alpha$ satisfies,
\begin{equation}
\varepsilon^\alpha + (1-\varepsilon)^\alpha e^{-\beta W_\alpha (1-\alpha)} = A.\label{eq:W alpah def eq}
\end{equation}
Note that due to the continuity of the R\'enyi divergences in the neighbourhood of one, this case follows by continuity. Solving Eq. \eqref{eq:W alpah def eq} for $W_\alpha$ gives us Eq. \eqref{eq:definitions2 no broduct bath}.
\end{proof}
As we will see later there exist $\rho_\cold^1$ such that, $W_\textup{ext}$ given by Eq. \eqref{eq:definitions1 no broduct bath}, has a solution (i.e. $W_\textup{ext}>0$) for \textit{any} $\varepsilon>0$. We can use this to write down an explicit solution to the maximization problem Eq. \eqref{eq:max nano as function fo cold bath}. Using Eqs. \eqref{eq:max nano as function fo cold bath}, \eqref{eq:eff explicit function of cold bath} and Lemma \ref{eq:explicit sol for W ext}, we conclude
\be \label{eq: max nano eff as function of cold bath no sup}
\eta^\textup{nano}(\rho_\cold^1)=\left( 1-\varepsilon +\frac{\Delta C(\rho_\cold^1)}{\inf_{\alpha\geq 0}W_\alpha(\rho_\cold^1)} \right)^{-1}
\ee 
where $W_\alpha$ is given by Eqs. \eqref{eq:definitions2 no broduct bath}, \eqref{eq:definitions3 no broduct bath} and recall $\Delta C$ can be found in Eq. \eqref{eq:delta C def}.
From Eqs. \eqref{eq: max nano eff as function of cold bath no sup}, \eqref{eq:definitions2 no broduct bath}, \eqref{eq:definitions3 no broduct bath}, we see that the optimization problem $\sup_{\rho_\cold^1} \eta^\textup{nano}(\rho_\cold^1)$ is still a formidable task. In the next section, see will show that we can use the results from Section \ref{section:standardthermo}, to drastically simplify the problem.

\subsubsection{An upper bound for the efficiency}\label{upper bound carnot eff}
Before moving on to solving the nanoscale efficiency explicitly, we will first use the results of Section \ref{sub:macroeffNPW} to find upper bounds for the efficiency in the nanoscale regime, in the context of extracting near perfect work (Def. \ref{def:near perfect work}).


Recall how we have discussed in comparing Sections \ref{sub:macro} and \ref{sub:nano}, that the solution for the family of free entropies $F_\alpha$, in the case of $F_1$ is simply the \helmholtz~free energy. Therefore, from Lemma \ref{eq:explicit sol for W ext}, it follows that $W_1$ is simply the maximum amount of extractable work according to Eq. \eqref{eq:Hemholts def}. From Eqs. \eqref{eq:max eff macro function of rho_C^1}, \eqref{eq:eff explicit function of cold bath},
\be\label{eq:std eff def}
\eta^\textup{mac}(\rho^1_\cold)=\left( 1-\varepsilon+ \frac{\Delta C(\rho^1_\cold)}{W_1(\rho^1_\cold)} \right)^{-1}.
\ee
One can now compare Eq.~\eqref{eq:std eff def} with Eq. \eqref{eq: max nano eff as function of cold bath no sup}, and note that for any $\rho^1_\cold\in\mathcal{S}$, we have $W_1 (\rho_\cold^1) \geq \inf_{\alpha\geq 0} W_\alpha (\rho_\cold^1)$. Therefore, we conclude that for any $\rho^1_\cold\in\mathcal{S}$,
\be\label{eq:std eff up bound gen eff}
\eta^\textup{nano}(\rho^1_\cold)\leq \eta^\textup{mac}(\rho^1_\cold).
\ee
Eq. \eqref{eq:std eff up bound gen eff} in conjunction with Lemma \ref{lem:cannot do better than carnot with porb failuer} has an important consequence. Namely, 
\begin{equation}
\sup_{\rho_\cold^1\in\mathcal{S}}\eta^\textup{nano}(\rho^1_\cold)\leq 1-\beta_h/\beta_c 
\end{equation}
with equality only if the state $\rho^1_\cold$ that solves the supremum is the limiting case where it tends to the initial state of the cold bath, $\rho^1_\cold\rightarrow \rho^0_\cold$. 
Therefore, in order to see whether we can still achieve the Carnot efficiency, we will consider the quasi-static regime in the rest of Section \ref{section:Nano/quantum scale heat engine at maximum efficiency}.

\subsubsection{Evaluating near perfect work in the quasi-static heat engine}\label{subsect:evaluatingWextquasi}
In light of the results from the previous section, we will now calculate the near perfect work $W_\textup{ext}$ for quasi-static heat engines, i.e. the case where $\varepsilon, g\ll 1$. Specifically, we make the following assumption about the cold bath Hamiltonian:\\
\textbf{(A.5)} The Hamiltonian is taken to be of $n$ qubits:
\begin{equation}\label{eq:tensorprod_Hamiltonian}
\hat H_\cold=\sum_{k=1}^n \id^{\otimes (k-1)}\otimes \hat H_{\onecold,k}\otimes\id^{\otimes (n-k)},\quad \textup{where}\quad\hat H_{\onecold,k}=\bar E_k|\bar E_k\ra\la \bar E_k|,
 \end{equation}
 and $\bar E_k>0$ is the energy gap of the $k$-th qubit. Here for simplify, we have chosen w.l.o.g. the ground state of each qubit to have an eigenvalue equal to zero.\\
The tensor product structure in Assumption (A.5) allows us to simplify $\rho_\cold^0$, to
\be \label{eq:rho0 for tensor H c ham}
\rho_\cold^0=\bigotimes_{i=1}^n\tau_{i,\beta_c},
\ee
where $\tau_{i,\beta_c}$ is the thermal state of $i$th qubit Hamiltonian $\hat H_{i,\onecold}$ at inverse temperature $\beta_c$.
For the simplicity of following proofs, we present them in the special case of identical qubits, i.e. that $\bar E_i=E$ for all $1\leq i\leq n$. This means Eq.~\eqref{eq:rho0 for tensor H c ham} can be reduced to 
\be \label{eq:rho0 for tensor H c ham 2}
\rho_\cold^0 = \tau_{\beta_c}^{\otimes n}.
\ee
Furthermore, since we consider quasi-static heat engines, the output state is
\be \label{eq:rho1 for tensor H c ham}
\rho_\cold^1 = \tau_{\beta_f}^{\otimes n},
\ee
with $\beta_f=\beta_c-g$ ,where $0<g\ll 1$. Eq.~\eqref{eq:tensorprod_Hamiltonian} together with Eq.~\eqref{eq:rho1 for tensor H c ham} allows us to further simplify Eq. \eqref{eq:definitions3 no broduct bath} to
\begin{align}
A&= \left(\frac{\sum_i p_i^\alpha q_i^{1-\alpha}}{\sum_i p_i'^\alpha q_i^{1-\alpha}}\right)^n,\label{eq:definitions3}
\end{align}
where $p_i= \frac{e^{-\beta_cE \delta_{i,0}}}{Z_{\beta_c}}$, $p_i'= \frac{e^{-\beta_fE \delta_{i,0}}}{Z_{\beta_f}}$, $q_i= \frac{e^{-\beta_hE \delta_{i,0}}}{Z_{\beta_h}}$, with $\delta_{1,0}=0,$  $\delta_{1,1}=1,$ are the probabilities of thermal states (different temperatures) for the qubit Hamiltonian $\hat H_\onecold$.
The proof follows along the same lines as the proof to Lemma \ref{eq:explicit sol for W ext}, but now noting that in Eq. \eqref{eq:dalphas  no broduct bath} we can replace $D_\alpha (\tau_\cold\|\tau_{\beta_h})$ and $D_\alpha (\rho_\cold^1\|\tau_{\beta_h})$ with $n D_\alpha (\tau_{\beta_c}\|\tau_{\beta_h})$ and $n D_\alpha (\tau_{\beta_f}\|\tau_{\beta_h})$ respectively. This follows from the additivity property of the R{\'e}nyi divergences. After proving the special case of identical qubits, we show in Theorem \ref{th:EfficiencyRandomEnergyGaps} that it can be extended to non-identical qubits as generally described by Assumption (A.5).

Since we are dealing with near perfect work and quasi-static heat engines, both $g>0$ and $\varepsilon>0$ are infinitesimally small. Thus with the goal in find of finding a solution for $W_\textup{ext}$ from Eqs. \eqref{eq:definitions1 no broduct bath}, \eqref{eq:definitions2 no broduct bath}, and \eqref{eq:definitions3}; we will proceed to find an expansion of $W_\alpha$ for small $\varepsilon$ and $g$.\\ 
\vspace{0.1cm}
\begin{center}
\underline{\textit{i) The expansion of $A$ in a quasi-static heat engine}}
\end{center}
To simplify our calculations of $W_{\rm ext}$, especially that of efficiency, it is important to express $A$ in Eq.~\eqref{eq:definitions3} in terms of its first order expansion w.r.t. the parameter $g$. Recall that this parameter $g = \beta_c-\beta_f$ is the difference of inverse temperature between the initial and final state of the cold bath.

Firstly, note that for any integer $n$, the expression in Eq.~\eqref{eq:definitions3} evaluates to $A|_{g=0} = 1$. This is because at $g=0$, $\beta_f=\beta_c$ and therefore the probabilities $p_i, p_i'$ are identical. To obtain an approximation in the regime $0<g\ll 1$, we derive
\begin{align}
\frac{dA}{dg} &= -n \left( \sum_i p_i^\alpha q_i^{1-\alpha} \right)^n \left( \sum_i p_i'^\alpha q_i^{1-\alpha} \right)^{-n-1} \left[ \sum_i \alpha p_i'^{\alpha-1} q_i^{1-\alpha}\frac{dp_i'}{dg} \right]\\
&=-\alpha nA \left( \sum_i p_i'^\alpha q_i^{1-\alpha} \right)^{-1} \left[\sum_i  p_i'^{\alpha} q_i^{1-\alpha} (\bar E_i-\langle \hat H_\onecold \rangle_{\beta_f}) \right]\label{eq:120secondeq}.
\end{align}
The first inequality holds by noticing that only the probabilities $p_i'$ depend on $g$, which means only the denominator in Eq.~\eqref{eq:definitions3} is differentiated, using the chain rule
\begin{equation}
\frac{dA(\lbrace p_i'\rbrace)}{dg} = \sum_i \frac{dA(\lbrace p_i' \rbrace)}{dp_i'} \frac{dp_i'}{dg}.
\end{equation}
The equality in Eq.~\eqref{eq:120secondeq} makes use of the fact that $\frac{dp_i'}{dg}=-\frac{dp_i'}{d\beta_f} = p_i' (\bar E_i-\langle \hat H_\onecold \rangle_{\beta_f})$ as derived in Eq.~\eqref{eq:piderivative}. Evaluated at $g=0$, implies that $p_i'=p_i$, and therefore this gives
\begin{align}
\frac{dA}{dg}\bigg|_{g=0} &= \alpha n B_\alpha, ~~\textup{where} \label{eq:Falpha}\\
~~B_\alpha &=\frac{1}{\displaystyle\sum_i p_i^\alpha q_i^{1-\alpha}} \displaystyle\sum_i p_i^\alpha q_i^{1-\alpha} \left(\langle \hat H_\onecold\rangle_{\beta_c} - \bar E_i \right). \label{eq:Balpha}
\end{align}
Recall that $p_i, q_i$ are probabilities of the thermal states of $\hat H_\onecold$, at inverse temperatures $\beta_c,\beta_h$ respectively. With this, we can write the expansion of $A$ with respect to $g$ as
\begin{equation}\label{eq:taylorexpandA}
A = 1+ \alpha n g B_\alpha + \bo (g^2).
\end{equation}

Later on, we will also need to evaluate the derivative of $B_\alpha$ w.r.t. $\alpha$. This quantity, when evaluated at $\alpha=1$, has a close relation to the change in average energy of the cold bath (per copy), $\frac{\Delta C}{n}$.
\begin{lemma}\label{lem:Falphap}
Let 
\be 
\Delta C' (\beta_c):= \frac{d}{dg}\Delta C(\beta_f)\bigg|_{g=0},
\ee
where recall $\beta_f = \beta_c -g$. Then 
\be 
B_1' = \frac{d B_\alpha}{d\alpha}\bigg|_{\alpha=1} = \frac{\beta_c-\beta_h}{n} \Delta C'(\beta_f)=(\beta_c-\beta_h)\cdot \textup{var}(\hat H_\onecold)_{\beta_c}.
\ee
\end{lemma} 
\begin{proof}
From the definition of $\Delta C$ (Eq. \eqref{eq:delta C def}) and using Eqs. \eqref{eq:tensorprod_Hamiltonian}, \eqref{eq:rho0 for tensor H c ham}, \eqref{eq:rho1 for tensor H c ham}, we have
\begin{align}\label{eq: Delta C per copy}
\frac{\Delta C}{n} = {\rm tr} [(\tau_{\beta_f}-\tau_{\beta_c})\hat H_\onecold].
\end{align}
Recalling that $\beta_f=\beta_c-g$ and using Eq. \eqref{eq:dev Delta C in terms of var}, from Eq. \eqref{eq: Delta C per copy} it follows
\begin{align}
\frac{1}{n}\Delta C' (\beta_c) = \frac{1}{n}\frac{d\Delta C}{d g}\bigg|_{g=0} = -\frac{1}{n}\frac{d\Delta C}{d\beta_f}\bigg|_{\beta_f=\beta_c} = {\rm var} (\hat H_\onecold)_{\beta_c}.
\end{align}
Now, let us evaluate the partial derivative of $B_\alpha$ w.r.t. $\alpha$. Denoting $r_i = \frac{p_i}{q_i}$, and invoking the chain rule of derivatives for Eq.~\eqref{eq:Balpha}
\begin{align}\label{eq:humongous_derivative}
\frac{d B_\alpha}{d\alpha} =& \left(\sum_i p_i^\alpha q_i^{1-\alpha}\right)^{-2} \Bigg{\lbrace} \left[\sum_i q_i r_i^\alpha \ln r_i  \left(\langle \hat H_\onecold\rangle_{\beta_c} - \bar E_i \right) \right]\left[\sum_i p_i^\alpha q_i^{1-\alpha}\right]\\
&-\left[\sum_i q_i r_i^\alpha \ln r_i\right]\left[ \sum_i p_i^\alpha q_i^{1-\alpha}  \left(\langle \hat H_\onecold\rangle_{\beta_c} -\bar  E_i \right)\right] \Bigg\rbrace.
\end{align}
Substituting $\alpha=1$ into Eq.~\eqref{eq:humongous_derivative}, we obtain that $\sum_i p_i^\alpha q_i^{1-\alpha}=1.$ Also, $\sum_i p_i^\alpha q_i^{1-\alpha}  \big(\langle \hat H_\onecold\rangle_{\beta_c} - \bar E_i \big)$ $=0$ while the factor multiplied in front is finite. Therefore, we are left with the terms
\begin{align}
B_1' &= \sum_i p_i \ln r_i  \left(\langle \hat H_\onecold\rangle_{\beta_c} - \bar E_i \right) \\
&= \sum_i p_i \left[\ln\frac{Z_h}{Z_c}+ (\beta_h-\beta_c) \bar E_i\right](\langle \hat H_\onecold\rangle_{\beta_c} - \bar E_i )\\
&= (\beta_c-\beta_h) {\rm var} (\hat H_\onecold)_{\beta_c}\\
&= \frac{\beta_c-\beta_h}{n} \Delta C'(\beta_c).\label{eq:F_1 dev in terms of standard}
\end{align}
The second equality comes from substituting $r_i = \frac{p_i}{q_i} = e^{(\beta_h-\beta_c)\bar E_i}\cdot Z_h/Z_c$. In the third equality, $\ln\frac{Z_h}{Z_c}$ is brought out of the summation, while the summation yields 0. Subsequently, we invoke $\sum_i p_i \bar E_i (\langle \hat H_\onecold\rangle_{\beta_c}-\bar E_i) = \langle \hat H_\onecold\rangle_{\beta_c}^2-\langle {\hat H_\onecold}^2\rangle_{\beta_c} = -{\rm var} (\hat H_\onecold)_{\beta_c}$.
\end{proof}
\vspace{0.1cm}
\begin{center}
\underline{\textit{ii) The expansion of $W_\alpha$ in the quasi-static heat engine}}
\end{center}
In the following we proceed to derive an expansion of $W_\alpha$ valid for small $g$, and $\varepsilon$. Note that $W_1$ is defined through continuity to be the limit of the R{\'e}nyi divergences at $\alpha\rightarrow 1$, and the small $\varepsilon$ and $g$ expansion does not hold for $\alpha= 0$, we shall have to examine $W_1$ and $W_0$ separately.

In the following and throughout the manuscript, we will use the notation $x\in(y,\infty]$ to indicate that the expression whose input $x$ in being referred to holds for $x\in(y,\infty)$ and for the limit case $\lim_{x\rightarrow +\infty}$. Similarly, we use the notation $x\in[y,\infty]$ when referring to an expression which holds for $x\in[y,\infty)$ and for $\lim_{x\rightarrow +\infty}$.\vspace{0.3cm}\\
(A) For $\varepsilon > 0$, $\alpha\in(0,1)\cup(1,\infty]$.\\
We start with the case $\varepsilon > 0$, $\alpha\in(0,1)\cup(1,\infty)$:
\begin{align}
W_\alpha &=  \frac{1}{\beta_h (\alpha-1)} [\ln (A-\varepsilon^\alpha)-\alpha\ln (1-\varepsilon)]\label{eq:approxWalpha line 1}\\
&=  \frac{1}{\beta_h (\alpha-1)} \left[\ln \left(1+\alpha ngB_\alpha+\bo(g^2)-\varepsilon^\alpha\right)-\alpha\ln (1-\varepsilon)\right]\\
&= \frac{1}{\beta_h (\alpha-1)} \left[\alpha n g B_\alpha + \bo(g^2)- \varepsilon^\alpha  +\bo(\varepsilon^{2\alpha})+\bo(g\varepsilon^\alpha) -\alpha\left(-\varepsilon+\bo\left(\varepsilon^2\right)\right)\right],\\
&= \frac{1}{\beta_h (\alpha-1)} \left[\alpha n g B_\alpha - \varepsilon^\alpha +\alpha\varepsilon\right] + \bo(g^2)+\bo(\varepsilon^{2\alpha})+\bo(g\varepsilon^\alpha)+\bo(\varepsilon^2).\label{eq:last eq in W alpha expansion A}
\end{align}
In the second equality, we have used the expansion of $A$ derived in Eq.~\eqref{eq:taylorexpandA}. In the third equality, we use the Mercator series
\begin{align}
\ln(1+x)=\sum_{k=1}^\infty \frac{(-1)^{k+1}}{k}x^k,\quad |x|<1,
\end{align}
to expand both of the natural logarithms in line Eq. \eqref{eq:approxWalpha line 1}. The order terms of $\bo(g^3),$ $\bo(g^4),$ $\bo(g^2 \varepsilon^\alpha)$ vanish because they are of higher order compared with $\bo(g^2)$ and $\bo(g\varepsilon^\alpha)$. The last equality occurs because $c\bo(g(x))=\bo(g(x))$ for any $c\in\mathbb{R}\backslash 0$. \\
Finally, we consider the limit case $\alpha\rightarrow\infty$. By direct calculation using the expression in line Eq. \eqref{eq:approxWalpha line 1}, we find
\be 
\lim_{\alpha\rightarrow +\infty} W_\alpha= \frac{1}{\beta_h} [n g \lim_{\alpha\rightarrow+\infty}B_\alpha +\varepsilon] + \bo(g^2)+\bo(\varepsilon^2),
\ee 
which is identical to the expression one obtains by taking the limit $\alpha\rightarrow+\infty$ in Eq. \eqref{eq:last eq in W alpha expansion A}. We thus conclude that for $\varepsilon > 0$, $\alpha\in(0,1)\cup(1,\infty]$,
\begin{align}
W_\alpha &= \frac{1}{\beta_h (\alpha-1)} \left[\alpha n g B_\alpha - \varepsilon^\alpha +\alpha\varepsilon\right] + \bo(g^2)+\bo(\varepsilon^{2\alpha})+\bo(g\varepsilon^\alpha)+\bo(\varepsilon^2).\label{eq:approxWalpha}
\end{align}
\vspace{0.3cm}\\
(B) For $\varepsilon > 0$, $\alpha =1$\\
We are now interested in finding a small $\varepsilon>0$, $g>0$ expansion for $W_1$, 
which is defined through continuity of the R{\'e}nyi divergences. 
Going back to Eq.~\eqref{eq:dalphas  no broduct bath}, note that $W_1$ is the maximum value such that Eq.~\eqref{eq:dalphas  no broduct bath} holds with equality, when all $D_\alpha$ terms in Eq.~\eqref{eq:dalphas  no broduct bath} are evaluated at $\alpha\rightarrow 1$. Recall that $\lim_{\alpha\rightarrow1} D_\alpha (\rho\|\tau) = D(\rho\|\tau)$ (see Eq. \eqref{eq:reyi in limits 2}), the relative entropy we have derived in Section \ref{section:standardthermo}. 
Therefore, one can write an equation for $W_1$ in a more compact form: $W_1$ is the value such that
\begin{equation}\label{eq:secondlawalp1}
n\cdot \left[\langle \hat H_\onecold\rangle_{\beta_c} -\frac{1}{\beta_h} S(\beta_c) \right]= n\cdot\left[\langle \hat H_\onecold\rangle_{\beta_f} -\frac{1}{\beta_h} S(\beta_f) \right]+ (1-\varepsilon) W_1 - \frac{1}{\beta_h} \htwo (\varepsilon), 
\end{equation}
where $\la \hat H_\onecold\rangle_{\beta_c}$ is the mean energy evaluated at temperature $T_\cold$, $S(\beta_c)$ is the von Neumann entropy of the state $\tau_{\beta_c}$, and $\htwo(\varepsilon)$ is the binary entropy function. Rearranging Eq.~\eqref{eq:secondlawalp1}, we get
\begin{equation}\label{eq:w1 for expa}
W_1 = \frac{1}{1-\varepsilon} \left[n\langle \hat H_\onecold\rangle_{\beta_c} -n\langle \hat H_\onecold\rangle_{\beta_f} -n\frac{1}{\beta_h} \left(S(\beta_c)-S(\beta_f)\right) + \frac{1}{\beta_h} \htwo(\varepsilon)\right].
\end{equation}
We can expand \eqref{eq:w1 for expa} using a power law expansion in $g$ and $\varepsilon$ for the terms in Eq.\eqref{eq:w1 for expa}, obtaining
\begin{align}\label{eq:w1expand}
W_1 &= \left[1+\varepsilon+\bo (\varepsilon^2)\right]\cdot \left[n\frac{d (-\langle \hat H_\onecold \rangle_{\beta_f} +\beta_{h}^{-1} S(\beta_f))}{dg}\bigg|_{g=0}\;g + \bo (g^2) + \frac{1}{\beta_h} \htwo(\varepsilon) \right].
\end{align}
 To proceed, we recall that $\beta_f=\beta_c-g$ and evaluate the term
\begin{align}
\frac{d (-\langle \hat H_\onecold \rangle_{\beta_f} +\beta_h^{-1} S(\beta_f))}{dg}\bigg|_{g=0} &= \frac{d (\langle \hat H_\onecold \rangle_{\beta_f} -\beta_h^{-1} S(\beta_f))}{d\beta_f}\bigg|_{\beta_f=\beta_c} = -{\rm var} (\hat H_\onecold)_{\beta_c} + \frac{\beta_c}{\beta_h}{\rm var} (\hat H_\onecold)_{\beta_c}\\
&= \frac{\beta_c-\beta_h}{\beta_h}{\rm var} (\hat H_\onecold)_{\beta_c}.
\end{align}
This implies that when fully expanded, Eq.~\eqref{eq:w1expand} reads as
\begin{align}
W_1  =& ng \frac{\beta_c-\beta_h}{\beta_h}{\rm var} (\hat H_\onecold)_{\beta_c} + \beta_h^{-1} \htwo (\varepsilon) + \bo (\varepsilon g) + \bo (\varepsilon) \htwo (\varepsilon) + \bo (g\varepsilon^2) + \bo (\varepsilon^2)\htwo (\varepsilon)\\
 &+ \bo (g^2)+\bo (\varepsilon g^2) + \bo (\varepsilon^2 g^2) \\
 =& ng \frac{\beta_c-\beta_h}{\beta_h}{\rm var} (\hat H_\onecold)_{\beta_c} +\beta_h^{-1} (-\varepsilon\ln\varepsilon+\varepsilon) + \bo (\varepsilon g)+ \bo (\varepsilon^2 \ln\varepsilon) + \bo (\varepsilon^2)+\bo (g^2),\label{eq:W 1 small ep and g expansion}
\end{align}
where we have used $\htwo(\varepsilon)=-\varepsilon\ln \varepsilon +\bo(\varepsilon)$, which follows from finding the power-law expansion of the second term in Eq. \eqref{def:binaryentropy}.

Although Eq. \eqref{eq:approxWalpha} is not defined for $\alpha=1$, we can evaluate it in the limit $\alpha\rightarrow 1$ to see if it coincides with the correct expression of $W_1$ (in Eq. \eqref{eq:W 1 small ep and g expansion}) at least for the leading order term (found in square brackets of Eq. \eqref{eq:approxWalpha}). For the leading order term of Eq. \eqref{eq:approxWalpha}, we find
\begin{align}
\lim_{\alpha\rightarrow 1}\frac{1}{\beta_h (\alpha-1)} \left[\alpha n g B_\alpha - \varepsilon^\alpha +\alpha\varepsilon\right]
&= \beta_{h}^{-1} \left[n g \lim_{\alpha\rightarrow 1}\frac{\alpha B_\alpha}{\alpha-1} - \lim_{\alpha\rightarrow 1}\frac{\varepsilon^\alpha-\alpha\varepsilon}{\alpha-1}\right]\\
&=\beta_{h}^{-1} \left[ ng \lim_{\alpha\rightarrow 1}\frac{\alpha B_\alpha}{\alpha-1} + (-\varepsilon\ln\varepsilon + \varepsilon)\right],\\
&=ng\frac{\beta_{c}-\beta_{h}}{\beta_{h}}  {\rm var} (\hat H_\onecold)_{\beta_c} + \beta_{h}^{-1} (-\varepsilon\ln\varepsilon + \varepsilon) .\label{eq:limWalpha}
\end{align}
The last equality holds because
\begin{align}
\lim_{\alpha\rightarrow 1}\frac{\alpha B_\alpha}{\alpha-1} & = \lim_{\alpha\rightarrow 1} \frac{d B_\alpha}{d \alpha}\label{eq:145}\\
&=(\beta_{c}-\beta_{h}) \cdot {\rm var} (\hat H_\onecold)_{\beta_c},\label{eq:last line quant}
\end{align}
where Eq.~\eqref{eq:145} is derived from L'H\^ospital rule ($B_1=0$ follows from the definition, see Eq. \eqref{eq:Falpha}), and Eq.~\eqref{eq:last line quant} comes by invoking Lemma \ref{lem:Falphap}. Thus noting that Eq. \eqref{eq:limWalpha} is simply the first two terms in Eq. \eqref{eq:last line quant}, we conclude that the small $g>0$ and $\varepsilon>0$ expansion of $W_\alpha$ for 
$\alpha\in(0,\infty]$ can be summarized as
\begin{flalign}\label{lim interchabge}
& W_\alpha=\\
&\begin{cases}
\frac{1}{\beta_h (\alpha-1)} \left[\alpha n g B_\alpha - \varepsilon^\alpha +\alpha\varepsilon\right] + \bo(g^2)+\bo(\varepsilon^{2\alpha})+\bo(g\varepsilon^\alpha)+\bo(\varepsilon^2) &\mbox{if } \alpha\in(0,1)\cup(1,\infty],\\
\displaystyle\lim_{\alpha\rightarrow 1^+}\frac{1}{\beta_h (\alpha-1)} \left[\alpha n g B_\alpha - \varepsilon^\alpha +\alpha\varepsilon\right] + \bo (\varepsilon g)+ \bo (\varepsilon^2 \ln\varepsilon) + \bo (\varepsilon^2)+\bo (g^2) &\mbox{if } \alpha=1.
\end{cases}
\end{flalign}
(C) For $\alpha =0$\\
We will now investigate the $\alpha=0$ case. This is also particularly important to understand the difference between perfect and near perfect work, since in Section \ref{subsec:impossible}, the impossibility of extracting perfect work arises from evaluating the allowed values of $W_\textup{ext}$ under the $\alpha=0$ constraint. We show that by allowing $\varepsilon >0$, $W_\textup{ext}>0$ is allowed once again. Recall $D_0(p\|q)=\lim_{\alpha\rightarrow 0}D_\alpha(p\|q)=\sum_{i:{p_i\neq 0}}q_i.$ Thus from Eq. \eqref{eq:dalphas  no broduct bath}
\begin{align}\label{eq:W constraint from D_0}
D_0 (\rho_\battery^0\|\tau_\textup{W})-D_0 (\rho_\battery^1\|\tau_\textup{W}) \geq   n D_0 (\tau_{\beta_f}\|\tau_{\beta_h})- n D_0 (\tau_{\beta_c}\|\tau_{\beta_h})=0.
\end{align}
where the last equality follows from the fact that thermal states have full rank. This inequality is satisfied for any value of $W_\textup{ext}$, since whenever $\varepsilon >0$, $\rho_\battery^1$ is a full rank state, and $D_0 (\rho_\battery^1\|\tau_\textup{W})=0$. Furthermore, $D_0 (\rho_\battery^0\|\tau_\textup{W}) \geq 0$ because all R{\'e}nyi divergences are non-negative. 
Therefore, taking into account Eqs. \eqref{lim interchabge} and \eqref{eq:W constraint from D_0}, for quasi-static heat engines which extract near perfect work, we only need to solve 
\be\label{eq:new inf W alpha}
W_\textup{ext}=\inf_{\alpha > 0} W_\alpha,
\ee
where $W_\alpha$ is given by Eq. \eqref{lim interchabge}.

\subsubsection{The choice of $\varepsilon$ determines the infimum to evaluating $W_\textup{ext}$}\label{subsect:LBeff}
In this section, we will show that the infimum over $\alpha > 0$ in Eq.~\eqref{eq:new inf W alpha} can be simplified to taking the infimum over $\alpha > \bar\kappa$ instead, where the parameter $\bar\kappa$ determines how quickly $\varepsilon$ goes to 0 w.r.t. the parameter $g$. We define $\kappa$ introduce $\bar\kappa$ in Lemma \ref{lemma:existence kappabar}, for functions of $\varepsilon (g)$.
\begin{lemma}\label{lemma:existence kappabar}
For every continuous function $\varepsilon(g)>0$ satisfying $\lim_{g\rightarrow 0^+}\varepsilon(g)=0,$ for which the limit  $\lim_{g\rightarrow 0^+} \varepsilon^\kappa(g)/ g$, $\kappa\in\rr$ exists, we have that $\exists$ $\bar \kappa\in\mathbb{R}_{\geq 0}$ s.t. 
\be \label{eq:existence lemma}
\delta (\kappa) = \lim_{g\rightarrow 0^+}\frac{\varepsilon^\kappa(g)}{g}=
\begin{cases}
0 &\mbox{if } \kappa>\bar \kappa\\
\sigma\geq 0 &\mbox{if } \kappa=\bar\kappa\\
\infty &\mbox{if } \kappa<\bar\kappa\\
\end{cases}
\ee
where $\bar \kappa=+\infty$ is allowed (that is to say, $\lim_{g\rightarrow 0^+}\frac{\varepsilon^\kappa(g)}{g}$ diverges for every $\kappa$) and $\sigma=+\infty$ is also allowed.
\end{lemma}
\begin{proof}
The main idea in this proof is to divide the non-negative real line into an infinite sequence of intervals in an iterative process. We specify the ends of these intervals by constructing a sequence $\lbrace \kappa_i \rbrace_{i=1}^\infty$, and evaluating $\delta$ at these points. We then prove that according to our construction, there are only two possibilities: \\
1) $\kappa_i$ forms a convergent sequence, where the limit $\lim_{n\rightarrow\infty}\kappa_n=\kappabar$, or \\
2) the ends of these intervals extend to infinity. In this case, $\kappabar=\infty$.
The way to construct this interval is as follows: in the first round, pick some $\kappa_1 >0$. The corresponding interval is $[0,\kappa_1]$. Evaluate $\delta (\kappa_1)$. If $\delta(\kappa_1)=\infty$, then proceed to look at the interval $[\kappa_1,\frac{3}{2}\kappa_1]$. Otherwise if $\delta(\kappa_1)<\infty$, choose $\kappa_2 = \frac{\kappa_1}{2}$ and evaluate $\delta(\kappa_2)$. Depending on whether $\delta (\kappa_2)$ goes to infinity, we pick one of the intervals $[0,\kappa_2]$ or $[\kappa_2,\kappa_1]$.

A general expression of choosing $\kappa_n$ can be written: during the $n$th round, define the sets $\mathcal{S}_n^{(0)}, \mathcal{S}_n^{(\infty)}$ such that
\begin{align*}
\mathcal{S}_n^{(0)} &= \lbrace \kappa_i| 1\leq i\leq n~ {\rm and} ~\delta (\kappa_i) = 0 \rbrace\\
\mathcal{S}_n^{(\infty)} &= \lbrace \kappa_i| 1\leq i\leq n ~{\rm and} ~\delta (\kappa_i) = \infty \rbrace.
\end{align*}
Note that if we find $\delta (\kappa_i)=c\neq 0$ for some finite constant $c$, then our job is finished, i.e. $\bar\kappa = \kappa_i$ (We prove this later). Subsequently, define for $n\geq 1$,
\begin{align*}
\kappa_n^{(0)} &= \min_{\kappa\in \mathcal{S}_n^{(0)} } \kappa \qquad {\rm and}\qquad\kappa_n^{(\infty)} = \max_{\kappa\in \mathcal{S}_n^{(\infty)} } \kappa.
\end{align*}
If either sets are empty, we use the convention that the corresponding minimization/maximization equals 0. Once these quantities are defined, we can choose the next interval by evaluating 
\vspace{-0.07cm}
\begin{equation}\label{eq:iterative}
\kappa_{n+1} = \kappa_n^{(\infty)} + \frac{|\kappa_n^{(\infty)} -\kappa_n^{(0)} |}{2}.
\end{equation}
In the $n$-th round, the corresponding interval is $[\kappa_n^{(\infty)},\kappa_{n+1}]$.

Let us now analyze why we can use this scheme to find $\bar\kappa$. Firstly, consider the case where $\delta (\kappa_i)$ whenever evaluated, produces infinity. This means that in each round, $\kappa_n^{(\infty)} = \kappa_n$ increases with $n$ (by the iterative scheme), and $\kappa_n^{(0)} = 0$ always stays at zero. Note that this scheme has been constructed in a way such that $\lim_{n\rightarrow\infty}\kappa_n=\infty$. Indeed, for all $n$, by using Eq.~\eqref{eq:iterative}, 
\begin{equation}
\kappa_{n+1} = \frac{3}{2}\kappa_n = \left(\frac{3}{2}\right)^2\kappa_{n-1} = \cdots = \left(\frac{3}{2}\right)^n\kappa_{1},
\end{equation}
which tends to infinity as $n$ goes to infinity, whenever $\kappa_1 >0$. Later we will prove a property of the function $\delta$, which combined with this scenario means that $\delta(\kappa)=\infty$ for every $\kappa\geq 0$. Therefore, $\kappabar=\infty$. 
\begin{figure}
\begin{center}
\includegraphics[scale=0.45]{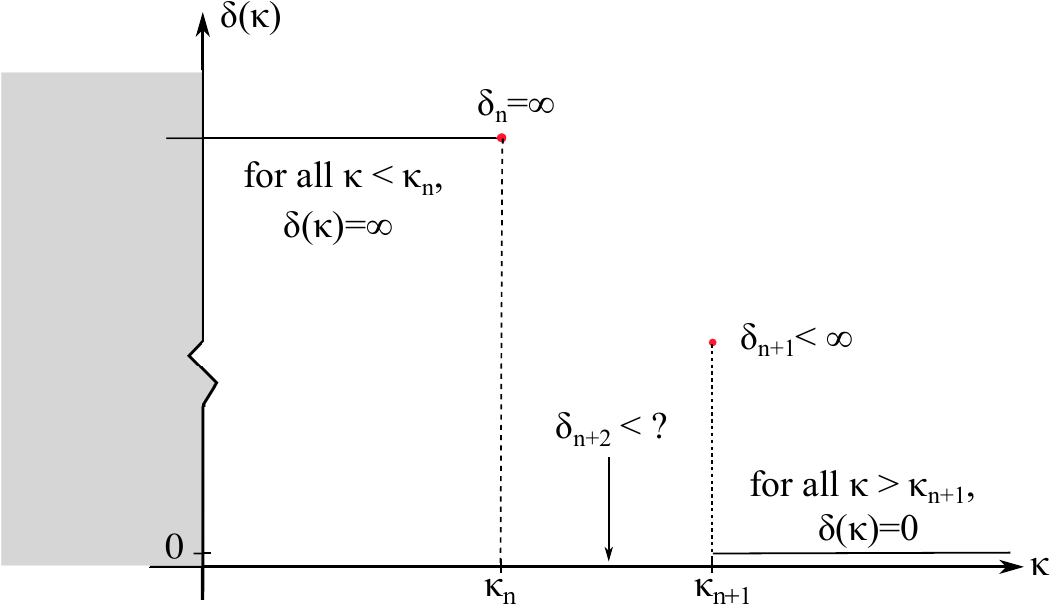}
\caption{Illustration of the scenario where $\delta (\kappa_n)=\infty$ and $\delta(\kappa_{n+1}) < \infty$.}
\label{fig:4}
\end{center}
\end{figure}

Next, suppose that there exist an $n$-th round, such that $\delta (\kappa_n)=\infty$ and $\delta(\kappa_{n+1}) < \infty$, as illustrated in Fig~\ref{fig:4}. Note that the function $\delta(\kappa)$ has a peculiar property, i.e. we know that if $\delta(\kappa_n)=\infty$, then for any $\kappa < \kappa_n$,
\be
\delta (\kappa)=\lim_{g\rightarrow 0^+}\underbrace{\varepsilon^{\kappa-\kappa_n}(g)}_{\rightarrow +\infty}\underbrace{\frac{\varepsilon^{\kappa_n}(g)}{g}}_{\rightarrow \infty}=+\infty.
\ee 
On the other hand, if $\delta(\kappa_{n+1})=0$, then we know that for any $\kappa > \kappa_{n+1}$,
\begin{align}
\delta (\kappa)=\lim_{g\rightarrow 0^+}\underbrace{\varepsilon^{\kappa-\kappa_{n+1}}(g)}_{\rightarrow 0}\underbrace{\frac{\varepsilon^{\kappa_{n+1}}(g)}{g}}_{\rightarrow 0}=0.
\end{align}
Moreover, if $\delta(\kappa_j)=c\neq 0$ for some positive, finite $c$, then following the same arguments, one can easily see that for all $\kappa<\kappa_j$, $\delta (\kappa) = \infty$ and for $\kappa > \kappa_j$, $\delta (\kappa) = 0$. In this case we find that $\bar\kappa=\kappa_j$. These observations are illustrated in Figure \ref{fig:4} for clarity.

One can now evaluate $\kappa_{n+2}$ (which is the midpoint of $\kappa_n$ and $\kappa_{n+1}$) and its corresponding value of $\delta(\kappa_{n+2})$. From this point on, in each iteration we either find $\bar\kappa$ exactly (whenever the function $\delta$ when evaluated produces a finite, non-zero number), or the length of the next interval gets halved, and goes to zero in the limit of $n\rightarrow\infty$. This, by Eq.~\eqref{eq:iterative}, also implies that $\lim_{n\rightarrow\infty}\kappa_n^{(\infty)} = \lim_{n\rightarrow\infty}\kappa_n^{(0)}$. We also know the following:\\
1) for all $\kappa < \kappa_n^{(\infty)}, \delta(\kappa)=\infty,$\\
2) for all $\kappa > \kappa_n^{(0)}, \delta(\kappa)=0.$\\
Therefore, we see that $\kappabar$ exists and $\kappabar = \lim_{n\rightarrow\infty}\kappa_n^{(\infty)} = \lim_{n\rightarrow\infty}\kappa_n^{(0)}.$
By this we conclude the proof. 
\end{proof}
To provide some intuition about how $\kappabar$ compares the rate of convergence $\varepsilon, g\rightarrow 0$, let us look at the following examples:\\
1) Consider $\varepsilon_1 (g) = \exp(-1/g)$. Then $\kappabar = 0$  with $\sigma=\infty$.\\
2) Consider $\varepsilon_2 (g) = g\ln g$. Then $\kappabar = 1$ with $\sigma=\infty$.\\
3) Consider $\varepsilon_3 (g) = c\cdot g^{1/k}$ for $k>0$. Then $\kappabar = k$ with $\sigma=c$.\\

In the next lemma, we consider the scenario of near perfect work, given in Def. \ref{def:near perfect work}, and show that this imposes a finite range of values $\kappabar$ should take. Given a particular $\kappabar$, we also show that the minimization of Eq.~\eqref{eq:new inf W alpha} changes with $\kappabar$.
\begin{lemma}\label{lemma: bar gamma}
Given any $\varepsilon (g) \in (0,1]$ as a continuous function of $g$, where $g>0$. If $\lim_{g\rightarrow 0^+} \varepsilon (g) =0$
and $\lim_{g\rightarrow 0^+} \frac{\Delta S}{W_\textup{ext}}=0$, then the following holds:
\begin{enumerate}
\item The quantity $\kappabar$ (defined in Lemma \ref{lemma:existence kappabar}) can only have any value in $\kappabar\in [0,1]$, where $\lim_{g\rightarrow 0^+} \frac{\varepsilon\ln\varepsilon}{g}=0$ has to hold if $\bar\kappa=1$. 
\item 
The extractable work can be written as
\be\label{eq:work extractable to first order in g}
\beta_h W_\textup{ext}=g\cdot \left[ \inf_{\alpha\geq \bar{\kappa}}\frac{n\alpha B_\alpha}{\alpha-1}+f(g) \right],
\ee
where $\lim_{g\rightarrow 0^+}f(g)=0$ and $\inf_{\alpha\geq \bar \kappa}$ can be exchanged for $\inf_{\alpha > \bar \kappa}$ if $\bar\kappa=0$.
\end{enumerate}

\end{lemma}
\begin{proof}
Firstly, let us use Eq. \eqref{lim interchabge} to simplify our expression for $W_\textup{ext}$: $W_\textup{ext}=\inf_{\alpha\geq 0} W_\alpha$, where
\be\label{eq: beta W=}
\beta_h W_\alpha=
\begin{cases}
 g\tilde{W}_\alpha + \bo(g^2)+\bo(\varepsilon^{2\alpha})+\bo(g\varepsilon^\alpha)+\bo(\varepsilon^2) &\mbox{if } \alpha\in(0,1)\cup (1,\infty]\\ 
 g\tilde{W}_1 + \bo (\varepsilon g)+ \bo (\varepsilon^2 \ln\varepsilon) + \bo (\varepsilon^2)+\bo (g^2) &\mbox{if } \alpha=1,
\end{cases}
\ee
and
\be\label{eq:alpha not 1}
\tilde{W}_\alpha:=\frac{1}{\alpha-1}\left( \alpha n B_\alpha+\alpha \frac{\varepsilon}{g}-\frac{\varepsilon^\alpha}{g} \right),
\ee
and for $\alpha=1$
\be\label{eq:alpha=1}
\tilde W_1=\lim_{\alpha\rightarrow 1} \tilde W_\alpha =\left(\lim_{\alpha\rightarrow 1} \frac{\alpha n B_\alpha}{\alpha-1} \right)+\frac{\varepsilon}{g}-\frac{\varepsilon}{g}\ln(\varepsilon).
\ee
From now on, the order terms in Eq.~\eqref{eq: beta W=} can be neglected, since it can be checked that all of them are of higher order compared to the terms we grouped in $\tilde W_\alpha$, in the limit of vanishing $g$. Even then, we note that due to the complicated form of $W_\textup{ext}$, it is not straightforward to begin our proof with the assumption $\lim_{g\rightarrow 0^+} \frac{\Delta S}{W_\textup{ext}}=0$. 

Instead, we begin by noting that given a function $\varepsilon (g)$ that satisfies the conditions of the above lemma, then one can invoke Lemma \ref{lemma:existence kappabar}, and therefore there exists a $\kappabar\in\mathbb{R}_{\geq 0}$ such that Eq.~\eqref{eq:existence lemma} holds. We then, for all possible $\kappa\in\mathbb{R}_{\geq 0}$, evaluate all $\tilde W_\alpha$ to take the infimum and obtain $W_\textup{ext}$. Given $W_\textup{ext}$, we then evaluate the quantity $\lim_{g\rightarrow 0^+} \frac{\Delta S}{W_\textup{ext}}=0$.

The value of $\kappabar$ determines how the limits of quantities like $\frac{\varepsilon}{g}, \frac{\varepsilon^\alpha}{g}$ behave. Therefore, we need to split the analysis into three different regimes: $\bar\kappa\in[0,1)$, $\bar\kappa=1$, $\bar\kappa\in(1,\infty)$. 
\vspace{0.2cm}\\
\underline{1) For $\bar\kappa\in[0,1)$}\vspace{0.1cm}\\
For this case, we know the following limits: \\
A. $\lim_{g\rightarrow 0^+} \frac{\varepsilon}{g} =0$.\\
B. For $\alpha < \kappabar$, $\lim_{g\rightarrow 0^+} \frac{\varepsilon^\alpha}{g} =\infty$.\\
C. For $\alpha = \kappabar$, $\lim_{g\rightarrow 0^+} \frac{\varepsilon^\alpha}{g} =\sigma\geq 0$.\\
D. For $\alpha > \kappabar$, $\lim_{g\rightarrow 0^+} \frac{\varepsilon^\alpha}{g} =0$.\\
E. Note that $\exists$ $k_1>\bar\kappa$ such that $1-k_1>0$. Thus $\lim_{g\rightarrow 0^+}\frac{\varepsilon}{g}\ln \varepsilon=\lim_{g\rightarrow 0^+}\frac{\varepsilon^{k_1}}{g}\;\varepsilon^{1-k_1}\ln \varepsilon=0$
 \\
Therefore, by using Eq.~\eqref{eq:alpha not 1} and \eqref{eq:alpha=1} (for $\alpha=1$ separately) we have
\be\label{eq:tildeWalpha_kappa<1}
\tilde W_\alpha=
\begin{cases}
+\infty & {\rm if}~ \alpha\in [0,\bar\kappa)\vspace{0.2cm}\\
\displaystyle\frac{\alpha n B_\alpha}{\alpha-1}+\frac{1}{\bar \kappa-1}\left(\bar \kappa \frac{\varepsilon}{g}-\frac{\varepsilon^{\bar\kappa}}{g}\right)= \displaystyle\frac{\alpha n B_\alpha}{\alpha-1} +\frac{\sigma}{|\bar \kappa-1|}
 +\bo\left(\frac{\varepsilon}{g}\right)\bar\kappa&{\rm if}~\alpha=\bar\kappa\vspace{0.2cm}\\
\displaystyle\frac{\alpha n B_\alpha}{\alpha-1} + \bo\left(\frac{\varepsilon^\alpha}{g}\right)&{\rm if}~\alpha\in(\bar\kappa,1)\vspace{0.2cm}\\
\displaystyle\frac{\alpha n B_\alpha}{\alpha-1} + \bo\left(\frac{\varepsilon}{g}\right)&{\rm if}~\alpha\in(1,\infty]
\vspace{0.2cm}\\
\displaystyle\lim_{\alpha\rightarrow 1} \frac{\alpha n B_\alpha}{\alpha-1} +\bo\left(\frac{\varepsilon\ln\varepsilon}{g}\right)&{\rm if}~\alpha=1,
\end{cases}
\ee
where the expression in ~Eq.\eqref{eq:tildeWalpha_kappa<1} has been written as a leading order term, plus higher order terms that vanish in the limit $g\rightarrow 0$  \footnote{In Eq. \eqref{eq:tildeWalpha_kappa<1}, the interval $[0,0)$ is taken to be the empty set (this is relevant for the case $\bar{\kappa}=0$).}. In the second line we have used $\left(\bar \kappa {\varepsilon}-{\varepsilon^{\bar\kappa}}\right)/(\bar \kappa-1)=|\left(\bar \kappa {\varepsilon}-{\varepsilon^{\bar\kappa}}\right)/(\bar \kappa-1)|$ as $\varepsilon\rightarrow0^+$ for $\bar\kappa\in[0,1)$.

Therefore, we conclude that for $\bar \kappa\in[0,1)$ and any $\sigma\geq 0$, due to continuity in $\alpha$ of $\frac{\alpha n B_\alpha}{\alpha-1}$, 
\be\label{work bar k in 0 1}
\beta_h W_\textup{ext}= \beta_h \inf_{\alpha > 0}W_\alpha=g\cdot\left[\inf_{\alpha\geq \bar\kappa}\frac{\alpha n B_\alpha}{\alpha-1}+\bo\left(f(g)\right)\right],
\ee
where $f$ satisfies $\lim_{g\rightarrow 0^+}f(g)=0$
~in the expression of Eq.~\eqref{eq:tildeWalpha_kappa<1}.  
Note that if $\kappabar=0,$ then $\inf_{\alpha\geq \bar \kappa}$ can be exchanged for $\inf_{\alpha> \bar \kappa}$ since in Eq. \eqref{eq:new inf W alpha} the point $\alpha=0$ was already excluded.

We can now calculate 
$\lim_{g\rightarrow 0^+} \frac{\Delta S}{W_\textup{ext}}$ for $\bar \kappa\in[0,1)$ and any $\sigma\geq 0$:
\be\label{Delta S/W lim 11}
\lim_{g\rightarrow 0^+} \frac{\Delta S}{W_\textup{ext}}
=\lim_{g\rightarrow 0^+}\frac{-\varepsilon\ln\varepsilon-(1-\varepsilon)\ln(1-\varepsilon)}{\left(\inf_{\alpha\geq \bar\kappa}\frac{\alpha n B_\alpha}{\alpha-1}\right) g}
=\lim_{g\rightarrow 0^+}\frac{1}{\inf_{\alpha\geq \bar\kappa}\frac{\alpha n B_\alpha}{\alpha-1}}\left( \underbrace{-\frac{\varepsilon\ln\varepsilon}{ g}}_{\rightarrow 0\;\textup{(Item E)}}-\underbrace{\frac{\varepsilon+\bo(\varepsilon^2)}{ g}}_{\rightarrow 0\; \textup{(Item A)}}\right)=0,
\ee
where we have assumed that 
\be \label{eq:inf barkappa gamma}
\inf_{\alpha\geq \bar\kappa}\frac{\alpha n B_\alpha}{\alpha-1}>0.
\ee
As we will see later (see Eq. \eqref{eq:alpha B_alpha/(1-alpha) positive}), Eq. \eqref{eq:inf barkappa gamma} holds if $\bar\kappa>0$. However, if $\alpha=0$
\be \label{eq:if alpah=0 have zero}
\frac{\alpha n B_\alpha}{\alpha-1}=0,
\ee
and we need to use Eq. \eqref{eq:tildeWalpha_kappa<1} for the case $\alpha=\bar\kappa=0$. The relevant line in Eq. \eqref{eq:tildeWalpha_kappa<1} is the 2nd line, from which we have $\tilde W_0=\sigma = \varepsilon^{\bar\kappa}/g=1/g$. which tends to $+\infty$ as $g\rightarrow 0^+$. Alternatively, as already mentioned, this result is also clear since $W_\alpha$ for $\alpha=0$ is infinite since it expresses the rank condition for state transitions which is always satisfied regardless of how much work is being extracted. 
Hence \be \label{eq:for kppabar =0 case 2}
\lim_{g\rightarrow 0^+} \frac{\Delta S}{W_\textup{ext}} =0 
\ee
in this case also. Thus from Eqs. \eqref{eq:if alpah=0 have zero} and 
\eqref{eq:for kppabar =0 case 2}, we conclude that 
Eqs. \eqref{work bar k in 0 1} and \eqref{Delta S/W lim 11} are still valid when $\bar\kappa=0$.
To summarize, so far we have proven that whenever $\kappabar\in [0,1)$, Eq.~\eqref{eq:work extractable to first order in g} holds for some $f(g)$ which vanishes as $g$ tends to zero, and furthermore $\lim_{g\rightarrow 0^+}\frac{\Delta S}{W_\textup{ext}}=0$. \vspace{0.2cm}\\
\underline{2) For $\bar\kappa\in(1,\infty)$}\vspace{0.15cm}\\
In this regime, like the previous analysis, we can list out the following limits:\vspace{0.1cm}\\
A. $\lim_{g\rightarrow 0^+} \frac{\varepsilon}{g} =0$.\\
B. By definition of $\kappabar$, for $\alpha < 1$, $\lim_{g\rightarrow 0^+} \frac{\varepsilon^\alpha}{g} =\infty$.\\
C. $\lim_{g\rightarrow 0^+}\frac{\varepsilon\ln\varepsilon}{g}=\infty$ since both $\frac{\varepsilon}{g}$ and $\ln\varepsilon$ goes to infinity as $g\rightarrow 0$.\\
Therefore, by using Eq.~\eqref{eq:alpha not 1} and \eqref{eq:alpha=1} (for $\alpha=1$ separately) we have
\be\label{eq:tildeWalpha_kappa>1}
\tilde W_\alpha=
\begin{cases}
\frac{1}{g}\cdot \frac{1}{1-\alpha}\left[ \varepsilon^\alpha + \bo(\varepsilon)+ \bo(g) \right] & {\rm if}~ \alpha\in [0,1)\vspace{0.15cm}\\
\frac{1}{g}\cdot \left[ - \varepsilon\ln\varepsilon +\bo(\varepsilon)  +\bo(g) \right]&{\rm if}~\alpha=1\vspace{0.15cm}\\
\frac{1}{g}\cdot \frac{1}{\alpha-1} \left[ \alpha\varepsilon + \bo(\varepsilon^\alpha) + \bo(g) \right] &{\rm if}~\alpha\in (1,\infty].
\end{cases}
\ee
Note that for all of these expressions in Eq.~\eqref{eq:tildeWalpha_kappa>1}, $\tilde W_\alpha\rightarrow\infty$ as $g\rightarrow0^+$.
Next we want to calculate $W_\textup{ext}$, which is the infimum of $W_\alpha$, taken over all $\alpha\geq 0$. Note that in the limit of vanishing $g$, $\varepsilon$ also goes to zero. Therefore from Eq.~\eqref{eq:tildeWalpha_kappa>1}, we see that the equation for $g\tilde W_\alpha$ which vanishes most quickly in the limit $g\rightarrow 0$ happens when $\alpha\in (1,\infty]$. Therefore, we conclude that for $\bar\kappa\in(1,\infty)$ and any $\sigma\geq 0$,
\be
\beta_h W_\textup{ext}=\beta_h \inf_{\alpha\geq 1}W_\alpha=g\cdot \left[\inf_{\alpha\geq 1}\frac{\alpha}{\alpha-1}\frac{\varepsilon}{g}+\bo\left(f(g)\right)\right]=\varepsilon+g\cdot \bo\left(f(g)\right)
\ee
We can now calculate 
$\lim_{g\rightarrow 0^+} \frac{\Delta S}{W}$ for $\bar \kappa\in(1,\infty)$ and any $\sigma\geq 0$:
\be\label{Delta S/W lim 21}
\lim_{g\rightarrow 0^+} \frac{\Delta S}{W}=\lim_{g\rightarrow 0^+}\frac{-\varepsilon\ln\varepsilon-(1-\varepsilon)\ln(1-\varepsilon)}{\varepsilon}=\lim_{g\rightarrow 0^+} \underbrace{-\frac{\varepsilon\ln\varepsilon}{\varepsilon}}_{\rightarrow \infty}-\underbrace{\frac{\varepsilon+\bo(\varepsilon^2)}{\varepsilon}}_{\rightarrow 1}=+\infty.
\ee
From this, we note that the whole regime of $\kappabar\in (1,\infty)$ does not contain any cases corresponding to our condition of interest: $\lim_{g\rightarrow 0^+}\frac{\Delta S}{W_\textup{ext}}=0$ never holds.\vspace{0.2cm}\\
\underline{3) For $\bar\kappa=1$}\vspace{0.15cm}\\
Similar to the first two cases, we again list out the relevant limits:\\
A. $\lim_{g\rightarrow 0^+} \frac{\varepsilon}{g} =\sigma$ for some $\sigma\geq 0$.\\
B. For $\alpha <1$, $\lim_{g\rightarrow 0^+} \frac{\varepsilon^\alpha}{g} =\infty$.\\
C. For $\alpha >1$, $\lim_{g\rightarrow 0^+} \frac{\varepsilon^\alpha}{g} =0$.\\
Therefore, by using Eq.~\eqref{eq:alpha not 1} and \eqref{eq:alpha=1} (for $\alpha=1$ separately) we have
\be\label{eq:tildeWalpha_kappa>1 2}
\tilde W_\alpha=
\begin{cases}
\frac{1}{g}\cdot \frac{1}{1-\alpha}\left[ \varepsilon^\alpha + \bo(\varepsilon)+ \bo(g) \right] & {\rm if}~ \alpha\in [0,1)\vspace{0.15cm}\\
\frac{1}{g}\cdot \left[ -\varepsilon\ln\varepsilon +\bo(\varepsilon)+\bo(g) \right]&{\rm if}~\alpha=1~~\wedge~~\sigma >0 \vspace{0.15cm}\\
n\displaystyle\lim_{\alpha\rightarrow 1}\frac{\alpha B_\alpha}{\alpha-1}- \frac{\varepsilon\ln\varepsilon}{g}\geq n\displaystyle\lim_{\alpha\rightarrow 1}\frac{\alpha B_\alpha}{\alpha-1} &{\rm if}~\alpha=1~~\wedge~~\sigma =0 \vspace{0.15cm}\\
\frac{1}{\alpha-1} \left[\alpha n B_\alpha +\alpha\sigma-\bo\left(\frac{\varepsilon^\alpha}{g}\right) \right] &{\rm if}~\alpha\in (1,\infty].
\end{cases}
\ee
Note that for $\alpha\in [0,1)$ and the  case $\alpha=1 \; \wedge \;\sigma>0$, $\tilde{W}_\alpha$ tends to infinity, while for the other cases $\tilde{W}_\alpha$ is finite.

Therefore, we can conclude that for $\bar\kappa=1$,
\be\label{work bar k is 1}
\beta_h W_\textup{ext}=g\cdot \left[\left(\inf_{\alpha\geq 1} \frac{\alpha}{\alpha-1}\left( n B_\alpha+\sigma \right)\right)   +\bo\left(f(g)\right)\right],
\ee
where $f(g)=\frac{\varepsilon^\alpha}{g}$ vanishes as $g$ tends to zero.

Now, we evaluate the limit $\lim_{g\rightarrow 0^+} \frac{\Delta S}{W}$ for $\bar \kappa=1$ and any $\sigma\geq 0$:
\begin{align}\label{Delta S/W lim 31}
\lim_{g\rightarrow 0^+} \frac{\Delta S}{W}&=\lim_{g\rightarrow 0^+}\frac{-\varepsilon\ln\varepsilon-(1-\varepsilon)\ln(1-\varepsilon)}{\left(\inf_{\alpha\geq 1} \frac{1}{\alpha-1}\left( \alpha n F_\alpha+\alpha\sigma \right)\right) g}=\lim_{g\rightarrow 0^+}\frac{-\varepsilon\ln\varepsilon}{c\cdot g}-\underbrace{\frac{\varepsilon+\bo(\varepsilon^2)}{c\cdot g}}_{\rightarrow 0}.
\end{align}
This limit of interest can be zero if and only if $\lim_{g\rightarrow 0^+} \frac{\varepsilon\ln\varepsilon}{g}=0$.\vspace{0.2cm}

We have calculated the limits $\lim_{g\rightarrow 0^+} \Delta S/W_\textup{ext}$ to leading order in $g$ for all functions $\varepsilon(g)>0$ satisfying $\lim_{g\rightarrow 0^+}\varepsilon=0$. 
These are found in Eqs. \eqref{Delta S/W lim 11}, \eqref{Delta S/W lim 21}, and \eqref{Delta S/W lim 31}. We have found that $\lim_{g\rightarrow 0^+} \Delta S/W_\textup{ext}=0$ occurs only in two cases:\\ 
~i) $\bar\kappa\in [0,1)$, and  \\
ii) $\kappabar=1~{\rm and}~\lim_{g\rightarrow 0^+} \frac{\varepsilon\ln\varepsilon}{g}=0$.\\
The amount of work, $W_\textup{ext}$ is found in Eq.~\eqref{work bar k in 0 1} and \eqref{work bar k is 1} respectively. Indeed, they take the form of Eq.~\eqref{eq:work extractable to first order in g}, for different functions $f(g)$. With this, we conclude the proof of the lemma.
\end{proof}

\subsubsection{Solving the infimum for $W_\textup{ext}$}\label{subsect:solvinginfimum}
We have seen in Lemma \ref{lemma: bar gamma} that the function $\frac{\alpha B_\alpha}{\alpha-1}$ corresponds to the largest order term in $W_\textup{ext}$ w.r.t. small g (quasi-static heat engine). Our next objective is to find the infimum of $\frac{\alpha B_\alpha}{\alpha-1}$ over $\alpha\in[\kappabar,\infty]$ appearing in Eq. \eqref{eq:work extractable to first order in g}.  Such an infimum is is not easy to evaluate, but whenever the cold bath consists of multiple identical qubits, we show that the derivative $\frac{d}{d\alpha}\frac{\alpha B_\alpha}{\alpha-1}$ has some nice properties. Roughly speaking, we show that this derivative does not have many roots, which in turn means that $\frac{\alpha B_\alpha}{\alpha-1}$ does not have many turning points. We have used this to prove in Lemma \ref{lim lim lem} that the infimum is either obtained at $\alpha=\kappabar$ or $\alpha\rightarrow\infty$.
\vspace{0.1cm}

The derivative of $\frac{\alpha B_\alpha}{\alpha-1}$ w.r.t. $\alpha$ is given by 
\begin{align}
\frac{d}{d\alpha}\frac{\alpha B_\alpha}{\alpha-1} &~=\frac{B_\alpha}{\alpha-1} + \alpha\frac{B_\alpha'}{\alpha-1} - \frac{\alpha B_\alpha}{(\alpha-1)^2} = \frac{B_\alpha'}{(\alpha-1)^2} \Bigg[ \alpha (\alpha-1)-\frac{B_\alpha}{B_\alpha'} \Bigg] = \frac{B_\alpha'}{(\alpha-1)^2}G(\alpha),\label{eq:Walp_to_der}
\end{align}
where
\begin{align}
G(\alpha)&:= \alpha (\alpha-1)-\frac{B_\alpha}{B_\alpha'} .
\end{align}

Now, we shall evaluate the quantities $B_\alpha, B_\alpha',$ and $\frac{d}{d\alpha}\frac{B_\alpha}{B_\alpha'}$ for the case of qubits (see Assumption (A.5)), where the energy levels are $\lbrace 0,E \rbrace$. By using Eq. \eqref{eq:tensorprod_Hamiltonian}, we evaluate the quantity $B_\alpha$ defined by Eq. \eqref{eq:Balpha} to obtain a simple expression:
\begin{align}
B_\alpha &= E\cdot \frac{e^{-\beta_c E}}{1+e^{-\beta_c E}} -E\cdot \frac{e^{-\alpha\beta_c E} e^{-(1-\alpha)\beta_h E}}{1+e^{-\alpha\beta_c E} e^{-(1-\alpha)\beta_h E}}\\
&= E\cdot \frac{1}{1+e^{\beta_c E}} - E \cdot \frac{e^{\alpha\beta_h E}}{e^{\alpha\beta_h E}+e^{(\beta_h+\alpha\beta_c)E}}\\
&= \frac{E}{1+e^{\beta_c E}}\cdot \left[1- \frac{e^{\alpha\beta_h E}(1+e^{\beta_c E})}{e^{\alpha\beta_h E}+e^{(\beta_h+\alpha\beta_c)E}} \right]\\
&= \frac{E}{1+e^{\beta_c E}}\cdot \frac{e^{(\beta_h+\alpha\beta_c)E}-e^{(\beta_c+\alpha\beta_h)E}}{e^{\alpha\beta_h E}+e^{(\beta_h+\alpha\beta_c)E}}.\label{eq:Falp}
\end{align}
We note that Eq. \eqref{eq:Falp} is  zero only if $\alpha=1$, and thus for $\alpha\neq 1$, ${\alpha B_\alpha}/(\alpha-1)\neq 0$. From Eq. \eqref{eq:last line quant}, we know that $\lim_{\alpha\rightarrow 1}\alpha B_\alpha/{(\alpha-1)}>0$, thus due to continuity,
\be\label{eq:alpha B_alpha/(1-alpha) positive}
\frac{\alpha B_\alpha}{\alpha-1}>0\quad \forall\; \alpha>0.
\ee
We also derive the first derivative of $B_\alpha$ w.r.t. $\alpha$ for the special case of qubits:
\begin{align}\label{eq:Falp1}
B_\alpha'=\frac{d B_\alpha}{d\alpha}= \frac{E^2 (\beta_c-\beta_h)}{\left[ e^{\alpha\beta_h E}+e^{(\beta_h+\alpha\beta_c)E} \right]^2}\cdot e^{(\beta_h+\alpha\beta_c+\alpha\beta_h)E}.
\end{align}
Note that since $\beta_c > \beta_h$ by definition,
therefore whenever $E>0$, then $B_\alpha' >0$ always holds. By further algebraic manipulation, we compute the first derivative of the function
\begin{equation}
\frac{d}{d\alpha}\frac{B_\alpha}{B_\alpha'}= \frac{\cosh [w(\beta_c,\beta_h,\alpha)E]}{\cosh (\beta_cE/2)},
\end{equation}
where $w(\beta_c,\beta_h,\alpha)=(\beta_c-\beta_h)\alpha+\beta_h-\frac{\beta_c}{2}$.

We have written Eq. \eqref{eq:Walp_to_der} in this form, since for the special case of qubits, namely Eq. \eqref{eq:Falp1}, $B_\alpha' >0$ is always true. Therefore, looking at the function $G(\alpha)$ whether it is positive or negative) will tell us whether $\frac{\alpha B_\alpha}{\alpha-1}$ (and therefore $W_\alpha$) is increasing or decreasing in a particular interval. 

In Lemma \ref{zeros lemma}, we identify the conditions on the energy spacing $E$ such that several different properties of $G(\alpha)$ hold.


\begin{lemma}\label{zeros lemma}
Consider $G(\alpha)=\alpha (\alpha-1) - \frac{B_\alpha}{B_\alpha'}$, where $B_\alpha, B_\alpha'$ is defined in Eq. \eqref{eq:Falp} and \eqref{eq:Falp1}. Then the following holds:\\
1) If $E(\beta_c-\beta_h)\tanh(\beta_cE/2)> 2$, 
\begin{align}\label{T 1 u 1 inf}
\exists\, 0<\tau<1~\mathrm{s.t.}\quad G(\alpha)<0\quad \forall \alpha\in(\tau,1)\cup(1,\infty)
\end{align}
2) If $E(\beta_c-\beta_h)\tanh(\beta_cE/2)< 2$,
\begin{align}\label{eq:G''1>0}
\exists~\underline{\alpha}>1~\mathrm {s.t.}\qquad &G(\alpha)>0\quad \forall \alpha\in(0,1)\cup(1,\underline\alpha)\nonumber\\
&G(\alpha)<0\quad \forall \alpha\in(\underline\alpha,\infty).
\end{align}
3) If $E(\beta_c-\beta_h)\tanh(\beta_cE/2)= 2$,
\begin{align}\label{eq:G''=0}
&G(\alpha)>0\quad \forall \alpha\in(0,1)\nonumber\\
&G(\alpha)<0\quad \forall \alpha\in(1,\infty).
\end{align}
\end{lemma}
\begin{proof}
First we note that since $B_1=0$, therefore $G(1)=0$. Let us also compute the derivative of $G(\alpha)$ w.r.t. $\alpha$:
\begin{equation}\label{eq:G1}
G'(\alpha)=2\alpha-1-\frac{\cosh\big((-\beta_c/2+\beta_h+(\beta_c-\beta_h)\alpha)E\big)}{\cosh(\beta_cE/2)}.
\end{equation}
Before we continue, there are several properties of the function $G'(\alpha)$ which we shall make use of. Firstly, note that $G'(1)=0$, in other words, $G'$ has a root at $\alpha=1$. Also, $G'(\infty)=-\infty$ for any value of $E>0,~ \beta_h>0,~\beta_c > \beta_h$ \footnote{This is due to the fact that $2\alpha$ increases linearly w.r.t. $\alpha$, while the $\cosh$ term increases exponentially.}. Also, since $2\alpha-1$ is linear (and hence both convex and concave), while the $-\cosh$ function is strictly concave~\footnote{To be more precise; due to the concavity of $f(x)=-a\cosh(b+xc)$ for $a>0$. This follows from the strict concavity of the $\cosh$ function, the invariancy of strict concavity under an affine transformation and the invariancy of strict concavity under multiplication by a positive constant.} , therefore the function $G'(\alpha)$ is  \emph{strictly concave}. This implies that the second derivative  $G''(\alpha)=\frac{d^2 G(\alpha)}{d\alpha^2}$ is strictly decreasing w.r.t. $\alpha$. \vspace{0.25cm}

The properties of $G'(\alpha)$ indicate that we can fully analyze the function by considering 3 different cases: 
\begin{enumerate}
\item $G'$ has two roots at $\alpha=\lbrace \underline{a}, 1 \rbrace$, wherewhere $\underline{a}\in (-\infty, 1)$. This corresponds to the case $G''(1)<0.$
\item $G'$ has two roots at $\alpha=\lbrace 1, \overline{a}\rbrace$, where $\overline{a}\in (1,\infty)$. This corresponds to the case $G''(1)>0.$
\item $G'$ has a single root at $\alpha=1$. This corresponds to the case $G''(1)=0.$
\end{enumerate}

\begin{figure}[h!]
      
	\centering 
        \begin{minipage}{0.32\textwidth}
                \includegraphics[width=\textwidth]{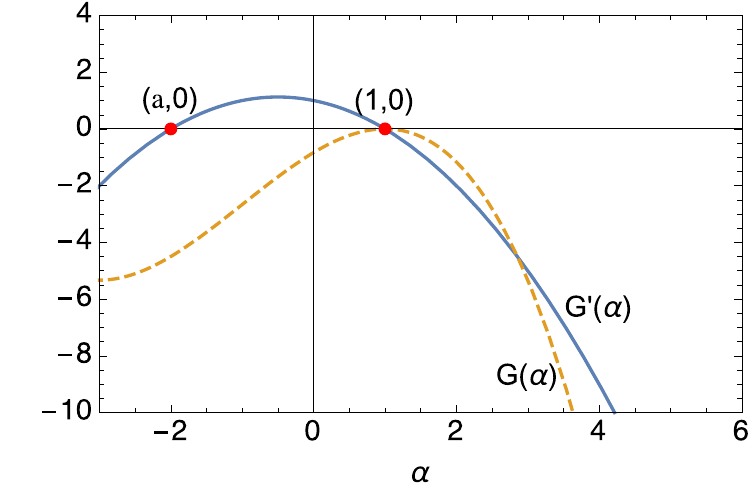}
                \caption{$ G''(1)<0$.}
                \label{fig:a}
        \end{minipage}%
        ~ 
        \begin{minipage}{0.32\textwidth}
                \includegraphics[width=\textwidth]{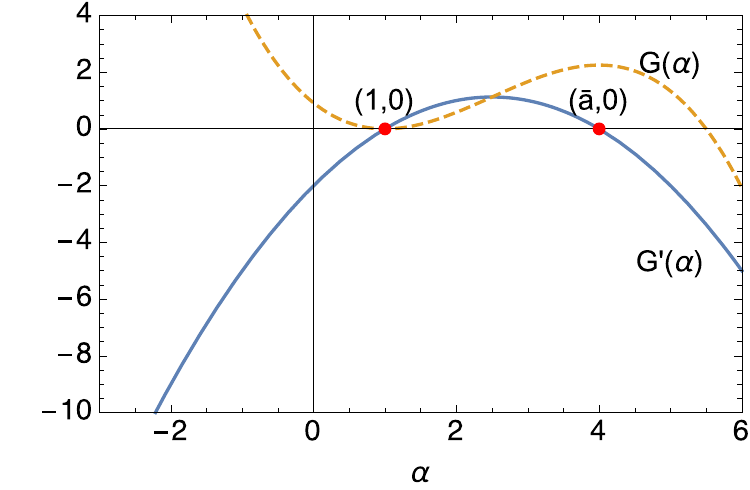}
                \caption{$ G''(1)>0$.}
                \label{fig:b}
        \end{minipage}
        ~ 
        \begin{minipage}{0.32\textwidth}
                \includegraphics[width=\textwidth]{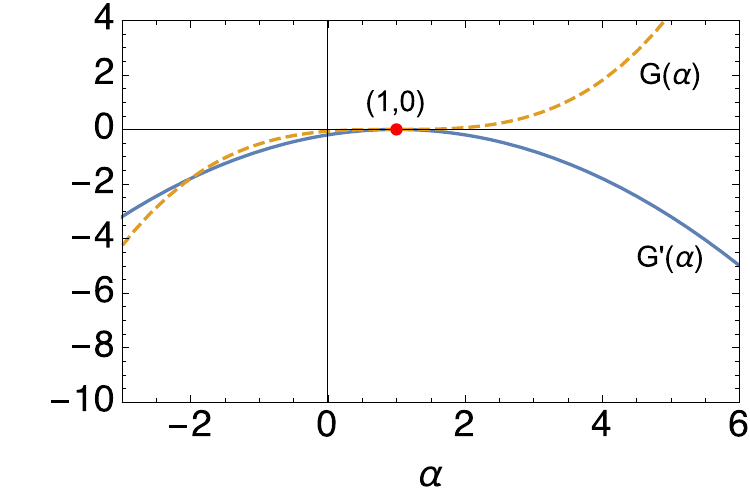}
                \caption{$ G''(1)=0$.}
                \label{fig:c}
        \end{minipage}
        \caption{A convex function $G'(\alpha)$ and its corresponding $G(\alpha)$, for different values of $G''(\alpha)$.}\label{fig:1}
\end{figure}

We shall now consider these cases one by one. Suppose that
\begin{equation}\label{eq:G''<0}
G''(1)=G''(\alpha)\Big|_{\alpha=1} = 2-  (\beta_c-\beta_h)E\tanh \left(\frac{\beta_c E}{2}\right) <0,
\end{equation}
then $G''(\alpha) <0$ for all $\alpha\in (1,\infty)$. Note that Eq.~\eqref{eq:G''<0} corresponds to the first condition in the lemma stated above.

This information about the second derivative $G''(\alpha)$ now allows us to conclude the following about $G(\alpha)$:
\begin{enumerate}
\item If for all $\alpha\in (1,\infty)$, $G''(\alpha)<0$, then we know that $G'(\alpha) <0$ holds for all $\alpha\in (1,\infty)$ too. Furthermore, this implies that $G(\alpha)$ is monotonically decreasing in the interval $(1,\infty)$ and therefore, $G(\alpha)<0$ for all $\alpha\in (1,\infty)$.
\item $G''(1)<0$ also implies that there exists an interval $(\tau,1)$ such that $G'(1)>0$ (See Fig. \ref{fig:a}). And since $G(1)=0$, this implies that within the interval $(\tau,1)$, $G(\alpha)<0$.
\end{enumerate}  
With this, we prove the first statement of the lemma.\vspace{0.2cm}

Let us now analyze the second case, where $G''(1)>0$. This implies that $G'(\alpha)>0$ at least for some interval $\alpha\in (1,\overline{a})$, then $G'(\alpha)$ changes sign exactly once at $\alpha=\overline{a}$, and goes to $-\infty$. (Refer to Fig. \ref{fig:b}). Also, recall that in the limit of $\alpha\rightarrow\infty$, $G$ also goes to $-\infty$. Therefore, we conclude that there exists some $\overline{\alpha}$ such that

\begin{equation}
G(\alpha) \left\{
\begin{array}{ll}
>0 & \alpha\in (1,\overline{\alpha})\\
<0 &  \alpha\in (\overline{\alpha},\infty)\\
\end{array}\right.
\end{equation}
With this, we prove the second statement of the lemma.

Finally, we look at the case where $G''(1)=0$, and make the following observations:
\begin{enumerate}
\item Since the function $G'(\alpha)$ is concave, and since $G''(1)=0$ implies that $\alpha =1$ is an extremum point for the function $G'(\alpha)$, we know that it must also be the global maximum. Therefore, we know that for any $\alpha\neq 1$,$G'(\alpha)<0.$
\item Since for the interval $\alpha\in(-\infty, 1), G'(\alpha)<0 $ and we know that $G(1)=0$, therefore we can deduce that for any $\alpha\in(-\infty, 1)$, $G(\alpha)>0$.
\item Since for the interval $\alpha\in (1, \infty), G'(\alpha)<0 $ and we know that $G(1)=0$, therefore we can deduce that for any $\alpha\in(1, \infty)$, $G(\alpha)<0$.
\end{enumerate}
With this, we prove the final statement of the lemma, and complete our proof.
\end{proof}

To summarize, in Lemma \ref{zeros lemma} we have identified conditions involving the energy gap of $\hat H_\onecold$, and the temperatures $\beta_h, \beta_c$. Depending on whether these conditions are satisfied, we can describe the positivity/negativity of $G(\alpha)$ for different regimes of $\alpha$. Comparing these different scenarios, we prove in Lemma \ref{lim lim lem} that for a quasi-static heat engine, the minimum of $\inf_{\alpha\geq\kappabar} \frac{\alpha B_\alpha}{\alpha-1}$ is obtained only either at $\alpha=\kappabar$ or $\alpha=\infty$.
\begin{lemma}\label{lim lim lem}
There exists some $0\leq\nu<1$ such that $\forall$ $\kappa$ satisfying $\nu<\kappa<1$, the following infimum is obtained at one of two points
\begin{align}\label{2 inf}
\inf_{\alpha\geq \kappa} \frac{\alpha B_\alpha}{\alpha-1}=\inf\left\{\lim_{\alpha \rightarrow \kappa} \frac{\alpha B_\alpha}{\alpha-1} ,\lim_{\alpha \rightarrow\infty} \frac{\alpha B_\alpha}{\alpha-1} \right\}
<\lim_{\alpha \rightarrow\beta} \frac{\alpha B_\alpha}{\alpha-1}\quad \forall\; \beta\in(\kappa,\infty),
\end{align}
where $B_\alpha$ is defined in Eq.~\eqref{eq:Falp}. Furthermore, if $E(\beta_c-\beta_h)\tanh(\beta_c E/2)\leq 2$, then we can set $\nu=0.$
\end{lemma}

\begin{proof}

\begin{enumerate}
\item If\vspace{-0.5cm}
\begin{align}\label{it 1 in inf proof}
\frac{d}{d\alpha}\frac{\alpha B_\alpha}{\alpha-1}
\begin{cases}
>0\; \forall\; \alpha\in(0,1)\cup(1,\overline\alpha) \text{ for some } \overline\alpha\geq 1\\
<0\; \forall\; \alpha\in(\overline\alpha,\infty).
\end{cases}
\end{align}
then $\forall~\kappa\in(0,1), $
\begin{align}\label{fulfill 1}
\hspace{-1.2cm}\inf_{\alpha\geq \kappa} \frac{\alpha B_\alpha}{\alpha-1}&=\inf\left\{\lim_{\alpha \rightarrow \kappa} \frac{\alpha B_\alpha}{\alpha-1}, \lim_{\alpha \rightarrow \infty} \frac{\alpha B_\alpha}{\alpha-1}\right\} <\lim_{\alpha \rightarrow\beta} \frac{\alpha B_\alpha}{\alpha-1}\quad \forall\; \beta\in(\kappa,\infty).
\end{align}
\end{enumerate}\vspace{0.2cm}

Recall from Eq. \eqref{eq:Walp_to_der} that
\begin{equation}
 \frac{d}{d\alpha}\frac{\alpha B_\alpha}{\alpha-1} = \frac{B_\alpha'}{(\alpha-1)^2} G(\alpha),
 \end{equation}
where $B_\alpha'>0$, and we have derived some properties of $G(\alpha)$ in Lemma \ref{zeros lemma}. In this proof, we apply Lemma \ref{zeros lemma} directly to consider the three scenarios detailed in Lemma \ref{zeros lemma}. 

First, consider the first statement of Lemma \ref{zeros lemma}. If $E(\beta_c-\beta_h)\tanh(\beta_cE/2)> 2$, then $\exists$ $0<t<1$ s.t.
\begin{align}\label{it 3 in inf proof}
\frac{d}{d\alpha}\frac{\alpha B_\alpha}{\alpha-1}
<0\; \forall\; \alpha\in(t,1)\cup(1,\infty)
\end{align}
then by continuity of $\frac{\alpha B_\alpha}{\alpha-1}$ in $\alpha$, we conclude that $\forall$ $\kappa$ satisfying $t<\kappa<1$
\begin{align}\label{inf eq inf 2}
\hspace{-1cm}\inf_{\alpha\geq \kappa} \frac{\alpha B_\alpha}{\alpha-1}=\lim_{\alpha \rightarrow \infty} \frac{\alpha B_\alpha}{\alpha-1}=\inf\left\{\lim_{\alpha \rightarrow \kappa} \frac{\alpha B_\alpha}{\alpha-1}, \lim_{\alpha \rightarrow \infty} \frac{\alpha B_\alpha}{\alpha-1}\right\} <\lim_{\alpha \rightarrow\beta} \frac{\alpha B_\alpha}{\alpha-1}\quad \forall\; \beta\in(\kappa,\infty).
\end{align}
Next, consider the second and third statements of Lemma \ref{zeros lemma} jointly, where $E(\beta_c-\beta_h)\tanh(\beta_cE/2)$ $\leq 2$. Note that both statements proved in Lemma \ref{zeros lemma} (namely, Eq.~\eqref{eq:G''1>0} and \eqref{eq:G''=0}) can be rewritten as the fact that there exists $\overline\alpha\geq 1$ s.t.
\begin{align}\label{con 1}
\frac{d}{d\alpha} \frac{\alpha B_\alpha}{\alpha-1}
\begin{cases}
>0\quad  \text{for } \alpha\in(0,1)\cup(1,\overline\alpha)\\
<0\quad \text{for } \alpha\in(\overline\alpha, \infty).
\end{cases}
\end{align}
In fact, the third statement is simply a special case of the second, where $\overline\alpha=1$. If Eq.~\eqref{con 1} holds, then $\forall~\kappa\in(0,1), $
\begin{align}
\hspace{-1.2cm}\inf_{\alpha\geq \kappa} \frac{\alpha B_\alpha}{\alpha-1}&=\inf\left\{\lim_{\alpha \rightarrow \kappa} \frac{\alpha B_\alpha}{\alpha-1}, \lim_{\alpha \rightarrow \infty} \frac{\alpha B_\alpha}{\alpha-1}\right\} <\lim_{\alpha \rightarrow\beta} \frac{\alpha B_\alpha}{\alpha-1}\quad \forall\; \beta\in(\kappa,\infty).
\end{align}
By setting $\tau=0$, we see that the statement of Lemma \ref{lim lim lem} is achieved.

Therefore, since we have analyzed all three cases stated in Lemma \ref{zeros lemma}, we conclude that there always exists $\nu\in [0,1)$ such that Eq.~\eqref{2 inf} will always be satisfied $\forall$ $\kappa\in (\nu,1)$.

\end{proof}

\subsubsection{Main results: evaluating the efficiency}\label{subsub:main}
In this section, we derive the efficiency of quasi-static heat engines in the nano /quantum regime. We first need to define the quantity
\be\label{def:sigma}
\Omega:=\min_{i\in \{1,\ldots,n\}} \frac{\bar E_i (\beta_c-\beta_h)}{1+e^{-\beta_c \bar E_i}},
\ee
where recall that $\bar E_i$ is the energy gap of the cold bath qubits, as described in Eq \eqref{eq:tensorprod_Hamiltonian} and the sentence right after it. Recall that $n$ denotes the number of qubits in the cold bath, where $n\in\mathbb{Z}^+$ is any positive integer.
Before stating the maximum efficiency, we will derive the efficiency as a function of $\bar\kappa$ defined in Lemma \ref{lemma:existence kappabar} (recall that this parameter is determined by the choice of $\varepsilon$). For simplicity, we will still consider the special case where $\bar E_i=E$ for all $i$ in Lemma \ref{lemma: efficiency as a function of kappa}, (i.e. all qubits of the cold bath are identical). Lemma \ref{lemma: efficiency as a function of kappa} shows us that under the condition of extracting near perfect work, one can choose $\varepsilon$ (and therefore $\kappabar$) such that a certain maximum efficiency value is achieved. The closer $\kappabar$ is to unity, the slower $\lim_{g\rightarrow 0^+}\Delta S/W$ converges to zero, and also the closer the efficiency is to the Carnot efficiency. 

Using this lemma, we prove the achievability of the Carnot efficiency which depends on $\Omega$. \textbf{This is the main result of our work, which is stated in Theorem \ref{th:EfficiencyRandomEnergyGaps}.} 

\begin{lemma}[Quasi-static efficiencies as a function of $\bar\kappa$]\label{lemma: efficiency as a function of kappa}
For any $n\in\mathbb{Z}^+$ number of qubits, consider quasi-static heat engines (Def. \ref{def:quasi static}) as a function of $\bar\kappa$ (defined in Lemma \ref{lemma:existence kappabar}) which extract near perfect work (Def. \ref{def:near perfect work}). For any $\kappa\in (0,\infty)\backslash \lbrace 1 \rbrace$, define
\begin{equation}\label{eq:def_gamma}
\gamma(\kappa):=\frac{\kappa B_{\kappa}}{\kappa-1}
\end{equation}
where $B_{\kappa}$ is defined in Eq.~\eqref{eq:Falp}, while $\gamma(1)$ and $\gamma(\infty)$ are defined by taking the limits $\kappa\rightarrow 1, \infty$ respectively.\\ 
If $\Omega\leq 1$ (see Eq. \eqref{def:sigma}):
\begin{itemize}
\item[1)] There exists $\nu\in [0,1)$ such that for any $\bar\kappa\in (\nu,1]$  (and $\lim_{g\rightarrow 0^+}(\varepsilon \ln \varepsilon)/g=0$ if $\kappabar=1$), the maximum efficiency is
\begin{align}\label{echi eff}
\eta^{-1}(\bar\kappa)=1+\frac{\beta_h}{\beta_c-\beta_h}\frac{\gamma(1)}{\gamma(\bar\kappa)}+\bo(f(g))+\bo(g)+\bo(\varepsilon),
\end{align}
where $\gamma(1)\geq \gamma({\bar\kappa})$ with equality iff ${\bar\kappa}=1$ and $\lim_{g\rightarrow 0+}f(g)=0$.
\item[2)]
The corresponding amount of work extracted is
\begin{equation}\label{eq:W ext kappa bar}
W_\textup{ext}(\bar\kappa)=g\frac{n}{\beta_h}\left[\gamma({\bar\kappa}) +\bo\left(f(g)\right)\right].
\end{equation}

\end{itemize}
If $\Omega>1$:
\begin{itemize}
\item [1)] There exists $\nu'\in [0,1)$ such that for any $\bar\kappa\in (\nu',1]$ (and $\lim_{g\rightarrow 0^+}(\varepsilon \ln \varepsilon)/g=0$ if $\kappabar=1$), the maximum efficiency is
\begin{align}\label{echi eff 2}
\eta^{-1}(\bar\kappa)=1+\frac{\beta_h}{\beta_c-\beta_h}\frac{\gamma(1)}{\gamma(\infty)}+\bo(f(g))+\bo(g)+\bo(\varepsilon),
\end{align}
where $\gamma(1)< \gamma(\infty)$.
\item[2)]
The corresponding amount of work extracted is
\begin{equation}\label{eq:W ext for Omega>1}
W_\textup{ext}(\bar\kappa)=g\frac{n}{\beta_h}\left[\gamma(\infty) +\bo\left(f(g)\right)\right]
\end{equation}
\end{itemize}
\end{lemma}
\begin{proof}
Firstly, let us begin by deriving the explicit form for $\gamma (1)$ and $\gamma (\infty)$:
\begin{align}\label{eq:gamma1}
\gamma (1) &= \lim_{\alpha\rightarrow 1} \frac{\alpha}{\alpha-1} B_\alpha = \lim_{\alpha\rightarrow 1} B_\alpha + \alpha B_\alpha' =  \frac{E^2 (\beta_c-\beta_h)}{(1+e^{\beta_c E})^2} e^{\beta_c E},
\end{align}
where we have made use of the L'H\^ospital rule. For $\alpha\rightarrow\infty$, since
\begin{align*}
\lim_{\alpha\rightarrow\infty} B_\alpha &= \lim_{\alpha\rightarrow\infty} \frac{E}{1+e^{\beta_c E}} \frac{e^{\beta_h E}-e^{\beta_c E}e^{-\alpha (\beta_c-\beta_h)E}}{e^{\beta_h E}+e^{-\alpha (\beta_c-\beta_h) E}}= \frac{E}{1+e^{\beta_c E}},
\end{align*}
therefore we have
\begin{align}\label{eq:gammainf}
\gamma (\infty) &= \lim_{\alpha\rightarrow 1} \left(1+\frac{1}{\alpha-1}\right)\cdot B_\alpha = \frac{E}{1+e^{\beta_c E}}.
\end{align}
By Lemma \ref{lim lim lem}, we know that the infimum of $\gamma (\alpha)$ for $\alpha\in[\bar\kappa,\infty)$ and $\bar\kappa\in(\nu,1]$ is either at $\alpha=\bar\kappa$ or $\alpha\rightarrow\infty$. Therefore, if we take the ratio of Eqs. \eqref{eq:gamma1} and \eqref{eq:gammainf} to be
\begin{equation}\label{gam 1 to gam inf <1}
\frac{\gamma(1)}{\gamma(\infty)} = \frac{E (\beta_c-\beta_h)}{1+e^{-\beta_c E}}=\Omega \leq 1,
\end{equation}
then $\gamma (\infty) \geq \gamma (1)> \gamma(\bar\kappa)$,  therefore the infimum of $\gamma (\alpha)$ for $\alpha\in[\bar\kappa,\infty)$ and $\bar\kappa\in(\nu,1]$ has to be obtained at $\alpha=\bar\kappa$. Taking this into account and using the condition which is equivalent to that of near perfect work in Eq. \eqref{eq:alternative equiv to n.p.w. def}, we can use Lemma \ref{lim lim lem}, to calculate the amount of work extracted:
\begin{equation}
W_\textup{ext}=\inf_{\alpha\geq 0} W_\alpha=g\cdot \left[\inf_{\alpha> {\bar\kappa}}\frac{n }{\beta_h} \gamma ({\bar\kappa}) + \bo\left(f(g)\right)\right]=g\frac{n}{\beta_h}\left[\gamma({\bar\kappa}) +\bo\left(f(g)\right)\right],
\end{equation}
where $\lim_{g\rightarrow 0+}f(g)=0$. On the other hand, we can calculate $\Delta C$, which is the change of average energy in the cold bath system, (recall this is done by Taylor expansion around $g=0$)
\begin{equation}\label{delta C small g}
\Delta C=n\left(\la E^2\ra_{\beta_c}-\la E\ra^2_{\beta_c}\right) g+\bo\left(g^2\right)=\frac{n\gamma(1)}{\beta_c-\beta_h} g+\bo\left(g^2\right).
\end{equation}
Using Eq. \eqref{eq:Delta W def as average}, we have $\Delta W=(1-\varepsilon)W_\textup{ext}$. The (inverse) efficiency, according to the definition \eqref{eq:efficiency}, is thus
\begin{align}
\eta^{-1}(\bar\kappa)&=1+\frac{\Delta C}{W_\textup{ext}}-\varepsilon=1+\frac{n\gamma(1)/(\beta_c-\beta_h)g+\bo\left(g^2\right)}{n\gamma(\bar\kappa)g/\beta_h +\bo\left(g f(g)\right)}-\varepsilon\\
&=1+\frac{\beta_h}{(\beta_c-\beta_h)}\frac{\gamma(1)}{\gamma(\bar\kappa)}+\bo(f(g))+\bo(g)+\bo(\varepsilon),
\end{align}
where we have used $\lim_{g\rightarrow 0^+} f(g)=0$ which is proven in Lemma \ref{lemma: bar gamma}. We will now investigate the efficiency when $\Omega>1$ is satisfied. Using $\Omega>1$ and Eq. \eqref{gam 1 to gam inf <1}, we have that $\gamma(\infty)< \gamma(1).$ Thus from Lemma \ref{lim lim lem}, due to continuity in $\bar\kappa$ of $\gamma(\bar\kappa)$ 
we conclude that there exists a $\nu'\in[0,1)$ such that for any $\bar\kappa\in(\nu',1]$, 
\be
\inf_{\alpha\geq \bar\kappa} \gamma(\alpha)=\gamma(\infty).
\ee
Therefore, since we are considering near perfect work, Eq. \eqref{eq:alternative equiv to n.p.w. def} holds and we can use Lemma \ref{lemma: bar gamma} to calculate the amount of work extracted
\begin{equation}
W_\textup{ext}=\inf_{\alpha\geq 0} W_\alpha=g\cdot \left[\inf_{\alpha> {\bar\kappa}}\frac{n }{\beta_h} \gamma ({\bar\kappa}) + f(g)\right]=g\frac{n}{\beta_h}\left[\gamma(\infty) +\frac{\beta_h}{n}f(g)\right],
\end{equation}
where $\lim_{g\rightarrow 0+}f(g)=0$. Thus using the definition of inverse efficiency (Eq. \eqref{eq:efficiency}), together with Eq. \eqref{delta C small g}, we have
\begin{align}\label{eq:reducedCarnotLem14}
\eta^{-1}(\bar\kappa)&=1+\frac{\Delta C}{W_\textup{ext}}-\varepsilon=1+\frac{n\gamma(1)/(\beta_c-\beta_h)g+\bo\left(g^2\right)}{n\gamma(\infty)g/\beta_h +\bo\left(g f(g)\right)}-\varepsilon\\
&=1+\frac{\beta_h}{(\beta_c-\beta_h)}\frac{\gamma(1)}{\gamma(\infty)}+\bo(f(g))+\bo(g)+\bo(\varepsilon),
\end{align}
where we have used $\lim_{g\rightarrow 0^+} f(g)=0$ which is proven in Lemma \ref{lemma: bar gamma}.
\end{proof}
We will now use Lemma \ref{lemma: efficiency as a function of kappa} to conclude our main result of this letter.
\begin{lemma}\label{Quantum/Nano heat engine efficiency} Consider the case of near perfect work (Def. \eqref{def:near perfect work}) and all cold bath qubits are identical (i.e. $\bar E_i=E$ for $i=1,\ldots,n$), then:
\begin{itemize}
\item [1)]
If $\Omega\leq 1$ (see Eq. \eqref{def:sigma}) the optimal achievable efficiency $\eta_\textup{max}$ (see Eq. \eqref{eq:max eff general}) is the Carnot efficiency:
\be
\eta_\textup{max}=\left(1+\frac{\beta_h}{\beta_c-\beta_h}\right)^{-1}
\ee
What is more, this efficiency is achieved for quasi-static heat engines, i.e. $\eta_\textup{max}=\eta_\textup{max}^\textup{stat}$ (see Eq. \eqref{def:quasi static eff nano}).
\item [2)]
If $\Omega>1$ and the heat engine is quasi-static, then the optimal achievable efficiency is (see Eq. \eqref{def:quasi static eff nano})
\be\label{eq:reduced eff}
\eta_\textup{max}^\textup{stat}=\left(1+\frac{\beta_h}{\beta_c-\beta_h}\Omega\right)^{-1}.
\ee
\item [3)] If $\Omega>1$ the maximum achievable efficiency $\eta_\textup{max}$ (see Eq. \eqref{eq:max eff general}), is strictly less that the Carnot efficiency for quasi-static heat engines.
\end{itemize}
\end{lemma}
\begin{proof}
In Lemma \ref{lem:cannot do better than carnot with porb failuer}, we found that the Carnot Efficiency is an upper bound for the efficiency when we are extracting near perfect work. 
We also found that Eq. \eqref{eq:alternative equiv to n.p.w. def} is satisfied iff we are extracting near perfect work. In Lemma \ref{lemma: efficiency as a function of kappa}, we derived the optimal achievable efficiency for quasi-static heat engines as a function of $\bar\kappa$ when  Eq. \eqref{eq:alternative equiv to n.p.w. def} is satisfied. By choosing $\kappabar<1$ arbitrarily close to one, if $\Omega\leq 1$ is satisfied, we will thus achieve an efficiency arbitrarily close to the Carnot efficiency. Thus since the upper bound is equal to the lower bound, we prove part 1) of the Theorem. Part 2) of the Theorem follows from setting $\bar\kappa=1$ in Lemma \ref{lem:cannot do better than carnot with porb failuer} when $\Omega>1$ is satisfied. 
\end{proof}

By making use of Lemma \ref{lemma: efficiency as a function of kappa}, one can generalize Lemma \ref{Quantum/Nano heat engine efficiency} to consider the more general case stated in A5 (at the beginning of Section \ref{subsect:evaluatingWextquasi}) where the cold bath still consists of qubits, however the energy gaps of the qubits can be arbitrary. For convenience, we re-write the general cold bath Hamiltonian here: for a set of variables $ \bar E_1>0,\cdots, \bar E_n>0 $,
\begin{equation}\label{eq:general cold bath again}
\hat H_\cold=\sum_{k=1}^n \id^{\otimes (k-1)}\otimes \hat H_\onecold^k\otimes\id^{\otimes (n-k)},\quad \textup{where}\quad\hat H_\onecold^k=\bar E_k|E\ra\la E|,
\end{equation}
Under the more general form of the cold bath Eq. \eqref{eq:general cold bath again}, we have the following theorem.
\begin{theorem}\label{th:EfficiencyRandomEnergyGaps}[Quantum/Nano heat engine efficiency]
 Consider a quasi-static heat engine (Def. \ref{def:quasi static}) which is extracting near perfect work (Def. \eqref{def:near perfect work}), when the cold bath consists of multiple qubits with energy gaps $\lbrace \bar E_i\rbrace_{i=1}^n$. 
\begin{itemize}
\item [1)]
If $\Omega\leq 1$ (see Eq. \eqref{def:sigma}) the optimal achievable efficiency $\eta_\textup{max}^\textup{stat}$ (see Eq. \eqref{def:quasi static eff nano}) is the Carnot efficiency:
\be
\eta_\textup{max}^\textup{stat}=\eta_C=\left(1+\frac{\beta_h}{\beta_c-\beta_h}\right)^{-1}
\ee
\item [2)]
If $\Omega>1$, then the maximum achievable efficiency is
\be
\eta_\textup{max}^\textup{stat}=\left(1+\frac{\beta_h}{\beta_c-\beta_h}\Omega\right)^{-1},
\ee

which is strictly less than the Carnot efficiency $\eta_C$.
\item[3)] Allowing for correlations between the final state of the battery and cold bath cannot improve the efficiencies achieved in 1) and 2) above.
\end{itemize}

\end{theorem}
\begin{proof}
1) is relatively simple to prove: as long as there exists a qubit with energy $\bar E_i$ such that $\frac{\bar E_i(\beta_c-\beta_h)}{1+e^{-\beta_h \bar E_i}}\leq 1$, one way to achieve Carnot efficiency is to simply disregard the rest of the cold bath, and act only on such qubits. The result is a simple application of 1) in Lemma \ref{Quantum/Nano heat engine efficiency}. This strategy might not be optimal in terms of work extracted, but it is sufficient for our proof.
\begin{figure}[h!]
	\centering 
        \begin{minipage}{0.32\textwidth}
                \includegraphics[width=\textwidth]{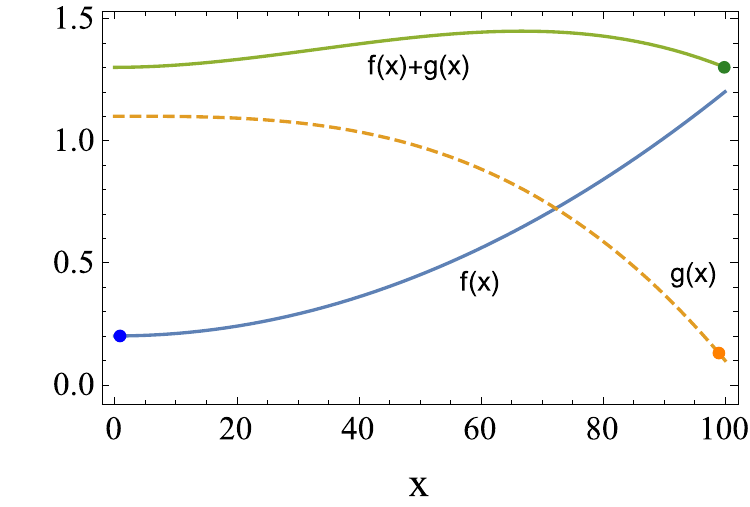}
        \end{minipage}%
        ~ 
        \begin{minipage}{0.32\textwidth}
                \includegraphics[width=\textwidth]{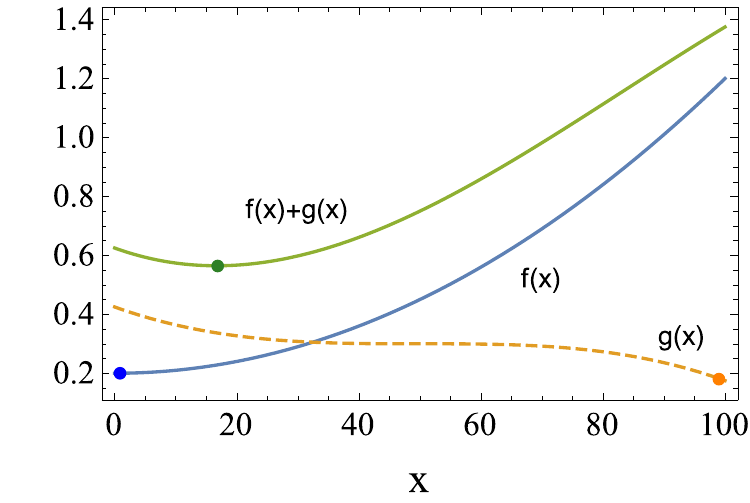}
        \end{minipage}
        ~ 
        \begin{minipage}{0.32\textwidth}
                \includegraphics[width=\textwidth]{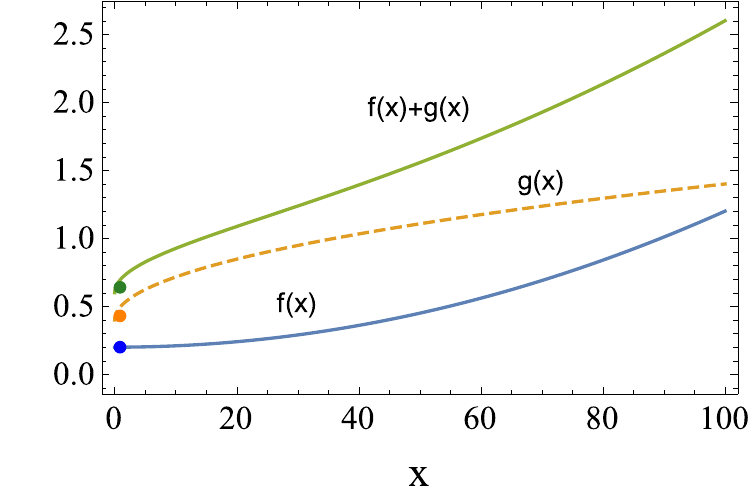}
        \end{minipage}
        \caption{Illustration of the minima of two individual functions $f(x),g(x)$ and minima of $f(x)+g(x)$.}\label{fig:5}
\end{figure}

For 2) suppose that $\Omega>1$. From the Eq. \eqref{def:sigma}, we conclude that for all $\bar E_i$ where $1\leq i\leq n$, $\Omega_i:=\frac{\bar E_i(\beta_c-\beta_h)}{1+e^{-\beta_h \bar E_i}}> 1$. By Lemma \ref{lemma: efficiency as a function of kappa}, we see that this implies that the work extractable for all the individual qubits (which is an optimization problem over all $\alpha\geq 0$) is obtained at $\alpha\rightarrow\infty$. In general, considering the qubits collectively does not mean that the collective $W_\textup{ext}$ is additive. This is because the minima of two functions is not necessarily the minima of these individual functions added together, as illustrated in the l.h.s. and middle diagrams of Figure.~\ref{fig:5}. However, (as illustrated on r.h.s. diagram of Figure.~\ref{fig:5}), when all the functions have their minima at the same value, then the collective minima is also obtained at that value. 

Next, we show that no matter which subset of qubits $\mathcal{S}$ one picks, Carnot efficiency cannot be achieved. 
We begin by introducing the notation $\gamma_i(\alpha)$,
where $\gamma_i(\alpha)$ is defined similarly with $\gamma(\alpha)$ in Eq.~\eqref{eq:def_gamma} and \eqref{eq:Falp}, and the index $i$ indicates that $E$ is substituted by $\bar E_i$ in Eq.~\eqref{eq:Falp}. Furthermore, recall that from Eq.~\eqref{gam 1 to gam inf <1}, $\Omega_i>1$ is equivalent to $\gamma_i(1)>\gamma_i(\infty)$. Now, consider any subset of qubit indices $\mathcal{S}$, the amount of extractable work (as a function of $g$) is
\begin{equation}
W_\textup{ext}^\mathcal{S}= \frac{g}{\beta_h} \left[\sum_{i\in\mathcal{S}}\gamma_i(\infty) +f(g)\right],
\end{equation}
where $\lim_{g\rightarrow 0+}f(g)=0$. 

On the other hand, we have that $\Delta C$ depends on the individual reduced qubit states, since there are no interaction terms in $\hat H_\cold$. Therefore, similar to Eq.~\eqref{delta C small g},
\begin{equation}
\Delta C^\mathcal{S}=\frac{g}{\beta_c-\beta_h} \sum_{i\in\mathcal{S}}\gamma_i(1)+\bo\left(g^2\right).
\end{equation}
Following the same proof in Eq.~\eqref{eq:reducedCarnotLem14} Lemma \ref{lemma: efficiency as a function of kappa}, 
\begin{align}
\eta^{-1}(\bar\kappa)=1+\frac{\Delta C}{W_\textup{ext}}-\varepsilon=1+\frac{\beta_h}{\beta_c-\beta_h}\frac{\sum_{i\in\mathcal{S}}\gamma_i(1)}{\sum_{i\in\mathcal{S}}\gamma_i(\infty)}+\bo(g)+\bo(\varepsilon).
\end{align}
As we have observed before, the inverse of the Carnot efficiency $\eta_C^{-1} = 1+\frac{\beta_h}{\beta_c-\beta_h}$. Furthermore, notice that by Eq.~\eqref{gam 1 to gam inf <1}, the condition $\Omega_i >1$ implies that $\gamma_i(1)>\gamma_i(\infty)$. Since $\Omega_i>1$ is true  for all $1\leq i\leq n$, therefore $\frac{\sum_{i\in\mathcal{S}}\gamma_i(1)}{\sum_{i\in\mathcal{S}}\gamma_i(\infty)}>1$. 

Lastly, part 3) is proven in Section \ref{Extensions to the setup}.
\end{proof}

Suppose $n$ is large. Then since we have a spectrum which looks like a quasi-continuum: the full range of the spectrum is very large, compared to the individual energy gaps. If one expects that in such a case, baths are of high temperature (small values of $\beta$), then the effects of quantization should give us the classical observations of being able to achieve the Carnot efficiency always. It can be seen, that for $E_\textup{min} = \displaystyle\min_{i\in\lbrace 1,\cdots, n\rbrace} \bar E_i$, if the quantities $\beta_c E_\textup{min}, \beta_h E_\textup{min} \ll 1$, then
\begin{equation}
\Omega = \frac{E_\textup{min} (\beta_c-\beta_h)}{1+e^{-\beta_c E_\textup{min}}} \leq E_\textup{min} (\beta_c-\beta_h) \ll 1.
\end{equation}
Whenever $\Omega \leq 1$, we know that Carnot efficiency is achievable.
\subsection{Running the heat engine for many cycles quasi-statically}\label{Running the heat engine for many cycles quasi-statically}
We have so far proven that a heat engine can achieve the Carnot efficiency when $\Omega\leq 1$. However, as like with macroscopic heat engines, this can only be achieved when the heat engine runs quasi-statically. Macroscopic heat engines can then extract a finite amount of work by running the heat engine over many cycles (in fact, over any infinite number of cycles if they want to obtain the Carnot efficiency in order to run quasi-statically). The following lemma, shows that when $\Omega\leq 1$, a nano-scale heat engine with a machine that runs over infinitely many cycles can also achieve the Carnot efficiency, while extracting any finite amount of work $W$ with vanishing entropy increase in the battery.

\begin{figure}[h!]
	\centering 
        \begin{minipage}{0.5\textwidth}
                \includegraphics[width=\textwidth]{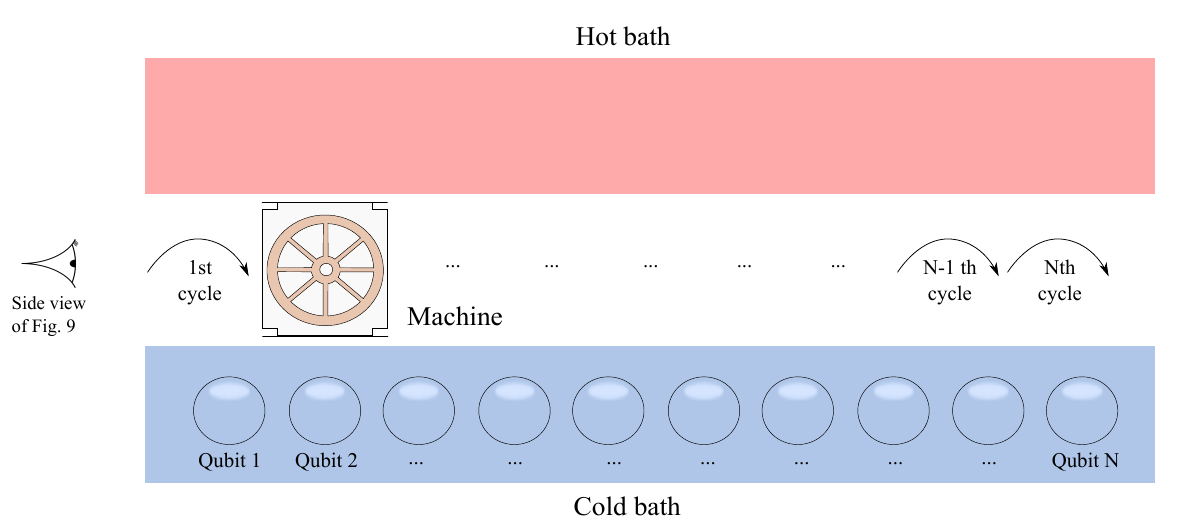}\caption{Depiction (top view) of a heat engine comprising of a hot bath, a cold bath consisting of $n$ identical qubits, a machine and a battery. In each cycle, the machine interacts specifically with one qubit from the cold bath, together with the hot bath and battery. After the end of one cycle, the machine is returned to its original state, and acts on a different qubit in the cold bath.}\label{fig:many cycles}
        \end{minipage}
        \hspace{0.2cm}
        \begin{minipage}{0.4\textwidth}
                \includegraphics[width=\textwidth]{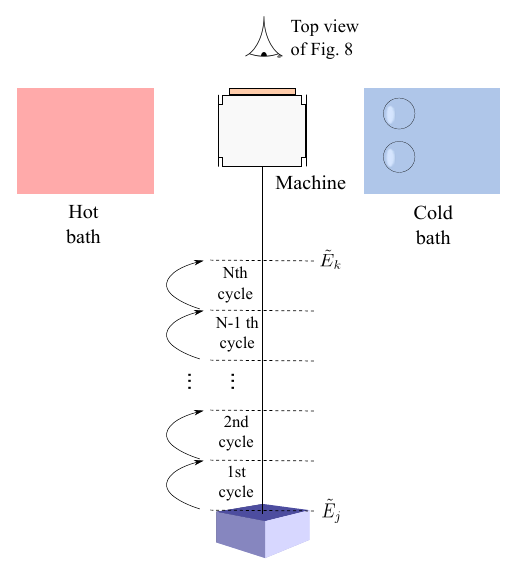}\caption{Side view of the heat engine. After each cycle of the machine, the battery, depicted here as a weight moves upward by a small amount. After $N$ machine cycles, it has been lifted from its original position $\ket{\tilde E_j}$ to a final state that has most of its weight on $\ket{\tilde E_k}$.}\label{fig:weight}
        \end{minipage}
       \end{figure}

For simplicity, we will work with the case in which the quasi-continuum battery has a part of its spectrum equal to that of at least $N$ qubits, each with an energy gap $W_\textup{ext}$. We work within this subspace. We will run a heat engine between a hot bath, cold bath using a machine which performs $N$ cyclic cycles. Let $\tilde E_j$ and $\tilde E_k$ be the smallest and largest energy eigenvalues within this subspace respectively. We let the initial state of the battery be 
\be 
\rho_\battery^0 =|\tilde E_j\rangle\!\langle \tilde E_j|,
\ee
$\hat H_\battery|\tilde E_j\rangle=\tilde E_j|\tilde E_j\rangle$ while we wish the final state of the battery to be of the form
\be 
\rho_\battery^1 =  r |\tilde E_k\rangle\!\langle\tilde E_k|+(1-r)\,\rho_\psi,
\ee
where $\hat H_\battery|\tilde E_k\rangle=\tilde E_k|\tilde E_k\rangle$, $\rho_\psi$ is some orthogonal state to $|\tilde E_k\rangle$ and the value of the probability $r$ is to be specified in the following lemma.
 We will define the amount of work extracted from the  machine for $N$ cycles
\be\label{def:work multi cyc}
W_\textup{cyc}:=\tilde E_k-\tilde E_j.
\ee
For simplicity, we will consider the case that the cold bath consists of $n$ identical qubits with $\Omega \leq 1$, and during each cycle the machine interacts with one qubit from the cold bath. The running of the heat engine is depicted in Fig. \ref{fig:many cycles} and \ref{fig:weight}.

\begin{corollary}\label{multi cycle lemma}[Many quasi-static heat engine cycles]
Let $W$ be the finite amount of work we wish to extract. Then for all $W>0$ and $\delta>0$ there exists an $n$ identical qubit cold bath (with $\Omega \leq 1$) and an $N\in \nn^+$ number of machine cycles with $n\geq N$ such that:
\begin{itemize}
\item[1)] $\eta_c\geq \eta\geq \eta_c-\delta,\quad$ where the efficiency $\eta$ is the efficiency per cycle and is defined by Eq. \eqref{eq:efficiency}, and $\eta_c=1-\beta_h/ \beta_c$ is the Carnot efficiency,\\
\item[2)] $W_\textup{cyc}\geq W-\delta$,\\
\item[3)] $S(\rho^0_\battery )=0$, $S(\rho^1_\battery )\leq \delta$, and\\
\item[4)] $r\geq 1-\delta$.
\end{itemize}
whats more, $\delta\rightarrow 0$ as $N\rightarrow +\infty$.
\end{corollary}
\begin{proof}
Since in the qubit subspace, the spectrum is that of at least $N$ qubits, we can write the initial state in the form
\be 
\rho_\battery^0 =  \ketbra{E_j}{E_j}^{\otimes N},
\ee 
with $\hat H_\battery \ket{E_j}^{\otimes N}=\tilde E_j \ket{E_j}^{\otimes N}$. We can now apply the heat engine results of Lemma \ref{lemma: efficiency as a function of kappa} to the setup. Namely, we can apply the results of one cycle to each of the qubit subspaces of the battery in parallel. From Lemma \ref{lemma: efficiency as a function of kappa} we conclude that this can be achieved with an efficiency given by Eq. \eqref{echi eff} and extract an amount of work per qubit/cycle given by Eq. \eqref{eq:W ext kappa bar}. For simplicity, we will run the heat engine using one qubit of the cold bath at a time.
The final state of the battery is thus 
\be\label{eq:rho1 initial multi} 
\rho^1_\battery=\left[ (1-\varepsilon)\ketbra{E_k}{E_k}+\varepsilon \ketbra{E_j}{E_j}\right]^{\otimes N}. 
\ee
Noting that $|\tilde E_k\rangle\!\langle\tilde E_k|=\ketbra{E_k}{E_k}^{\otimes N}$ by definition, Eq. \eqref{eq:rho1 initial multi} can be written as
\be\label{eq:rho in top plus rest form multi}
\rho^1_\battery=(1-\varepsilon)^N|\tilde E_k\rangle\!\langle\tilde E_k|+\left[1-(1-\varepsilon)^N\right]\,\rho_\psi. 
\ee
with $\rho_\psi$ orthogonal to $|\tilde E_k\rangle$.
From Eq. \eqref{def:work multi cyc} it follows,
\be 
W_\textup{cyc}=N W_\textup{ext}=\frac{N g}{\beta_h}\left[\gamma({\bar\kappa}) +\bo\left(f(g)\right)\right],
\ee 
where in the last line we have used Eq. \eqref{eq:W ext kappa bar}. We now set
\be 
N=N(g)=\frac{\beta_h}{\gamma(\bar\kappa)}\frac{W}{g}
\ee
for all $g>0$ satisfying the constraint $N(g)\in \nn^+$. 
For any positive constant $\frac{\beta_h W}{\gamma (\bar\kappa)} > 0$, one can always consider the values of $\frac{\beta_h W}{\gamma (\bar\kappa)}>g> 0$ so that $N(g)$ is large.
This constraint imposes $g=\beta_h W/(\gamma(\bar\kappa)\, N)$, where $N$ has to be an integer. Therefore, $g$ now belongs to a subset of the positive real line, rather than the positive real line itself as previously. However, since $g$ monotonically decreases to zero as $N$ increases to infinity, 
we can still take the limit $g\rightarrow 0^+$ as before.
Thus achieving
\be 
W_\textup{cyc}=W+\bo\left(f(g)\right).
\ee
Since $\lim_{g\rightarrow 0^+} f(g)=0$, we conclude part 2) of Corollary \ref{multi cycle lemma}. For the entropy of the final state of the battery we have
\begin{align}\label{eq:S or rho final bat multi}
S(\rho_\battery^1) &= N S\left( (1-\varepsilon)\ketbra{E_k}{E_k}+\varepsilon \ketbra{E_j}{E_j}\right)=\frac{\beta_h W}{\gamma(\bar\kappa)}\frac{(1-\varepsilon)\ln (1-\varepsilon)+\varepsilon\ln \varepsilon}{g}=\bo\left(\frac{\varepsilon\ln \varepsilon}{g}\right).
\end{align}
As stated above the efficiency is given by Eq. \eqref{echi eff}, and thus we can always choose $\bar \kappa \in (0,1)$, and $g$ (recall $\epsilon\rightarrow 0^+$ as $g\rightarrow 0^+$) such that 1) in Corollary \ref{multi cycle lemma} is satisfied. Furthermore, recall from the proof of Lemma \ref{lemma: bar gamma} that
\be\label{eq:eplnep/g}
\lim_{g\rightarrow 0^+}\frac{\varepsilon\ln\varepsilon}{g}=0,
\ee
for all $\bar\kappa \in (0,1)$. Thus, from Eq. \eqref{eq:S or rho final bat multi} we conclude that 3) in Corollary \ref{multi cycle lemma}. We will now prove part 4) of the Corollary. From Eq. \eqref{eq:rho in top plus rest form multi} and part 4) of the Corollary, we can identify $r=(1-\varepsilon)^N$. We thus study the limit 
\be 
\lim_{g\rightarrow 0^+} (1-\varepsilon)^N=\left(\lim_{g\rightarrow 0^+} (1-\varepsilon)^{1/g}\right)^{\beta_h W/\gamma(\bar\kappa)}=\left(\lim_{g\rightarrow 0^+} \left(\underbrace{(1-\varepsilon)^{1/\varepsilon}}_{\rightarrow\, \me}\right)^{\varepsilon/g}\right)^{\beta_h W/\gamma(\bar\kappa)}=1,
\ee 
where going to the last line, we have used that fact that Eq. \eqref{eq:eplnep/g} implies that $\varepsilon/g\rightarrow 0$ as $g\rightarrow 0^+$. We thus conclude part 4) of the corollary.
\end{proof}
Thus by choosing $\delta>0$ sufficiently small in Corollary \ref{multi cycle lemma}, we can extract any finite amount of work with an arbitrarily small entropy contribution with an efficiency arbitrarily close to the Carnot efficiency as long as $\Omega\leq 1.$

\section{Extensions to the setup}\label{Extensions to the setup}
Arguably, one may think that the inability to always achieve the Carnot efficiency in the nano regime is due to some subtlety of our setup (even though we have shown that according to the standard free energy one can always achieve the Carnot efficiency with our setup). 
For such reasons, in the next few sections we show that even under more general conditions than those laid out in Section \ref{section:The setting}, one still cannot achieve the Carnot efficiency when $\Omega>1$. 

In Section \ref{Correlated final state of battery}, we show that allowing for correlations between the final state of the battery and cold bath (and/or the finite dimensional machine) does not allow us to achieve the Carnot efficiency. The main result is Theorem \ref{corrs dont help}.

In Section \ref{A more general final battery state}, we show that allowing for the battery to be \textit{any} state with trace distance $\varepsilon$ from $\ketbra{E_k}{E_k}_\battery$ cannot allow us to achieve the Carnot efficiency when $\Omega>1$. This shows that whenever we are unable to achieve the Carnot efficiency, it is not a artificial defect from an overly specified battery model. The main result is Theorem \ref{thm:general_batt}.

\subsection{Final correlations between battery, cold bath, and machine}\label{Correlated final state of battery}
In Section \ref{subsub:highcertainty}, we stated that the final state of the heat engine after tracing out the hot bath was of the form
\be\label{eq: old not correlate battery c bath}
\tr_\hot(\rho^1_\total)=\rho^1_\cold\otimes\rho^1_\textup{M}\otimes\rho^1_\textup{W}
\ee
where $\rho^1_\textup{W}=\varepsilon |E_j\ra\la E_j|_\textup{W}+(1-\varepsilon)|E_k\ra\la E_k |_\textup{W},$ i.e. the final state of the charged battery was a tensor product with the cold bath. We also demanded that the heat engine is cyclic i.e. that $\rho^1_\textup{M}=\rho^0_\textup{M}$. In this section, we show that if one allows for the final state of the battery, cold bath and machine to become correlated\footnote{Recall that the final state of the cold bath, machine and battery are already allowed to become correlated with the hot bath}, one still cannot achieve the Carnot efficiency when $\Omega>1$. That is to say, in this section we allow the final state to be
\be\label{eq:new correlate battery c bath}
\tr_\hot(\rho^1_\total)=\rho^1_{\cold\textup{M}\textup{W}}
\ee
with only two natural constrains, namely that our heat engine actually extracts work, i.e. that
\be\label{eq:only constrint}
\rho^1_\textup{W}=\varepsilon |E_j\ra\la E_j|_\textup{W}+(1-\varepsilon)|E_k\ra\la E_k |_\textup{W},
\ee
as before, and also that the heat engine is still cyclic, i.e.
\be\label{eq:only constrint 2}
\rho^1_\textup{M}=\rho^0_\textup{M}.
\ee
Throughout this section, (unless stated otherwise) we will write $\rho^1_{\cold\textup{M}\textup{W}}$ to refer to any generic tripartite quantum state on the cold bath, machine and battery satisfying Eqs. \eqref{eq:only constrint} and \eqref{eq:only constrint 2}. 
\begin{itemize}
\item In Section \ref{subsub:def_gen_eff}, we first define the generalized efficiency where one is allowed to consider correlated final states. We see that although this may potentially affect the amount of extractable work $W_\textup{ext}$, the amount of heat change in the bath remains the same, by making use of energy conservation and the fact that the global Hamiltonian $H_\total$ does not contain interaction terms between subsystems.
\item In Section \ref{subsub:finalcorr_cannotsurpassCE}, we make use of the generalized second law when $\alpha = 1$ (which is also the macroscopic second law), in order to show that final correlations still do not allow the surpassing of Carnot efficiency. This can be proven by noting that the von Neumann entropy is subadditive, and the result is summarized in Lemma \ref{lem:coherences cannot help fixed state}. A proof sketch can be found in the beginning of Section \ref{subsub:finalcorr_cannotsurpassCE}.
\item In Section \ref{subsub:achievability_fundamentally_reduced}, we turn to the case where $\Omega >1$, where without final correlations it is shown in Theorem \ref{th:EfficiencyRandomEnergyGaps} that Carnot efficiency cannot be achieved. 
\end{itemize}

\subsubsection{Defining the generalized efficiency}\label{subsub:def_gen_eff}
Recall that before (see Section \ref{sub:finding_simplified_expression}), we have shown in Eq.~\eqref{eq:eff explicit function of cold bath} that if the following assumptions hold: 
\begin{enumerate}[label=(\roman*)]
\item the final reduced state of the battery $\rho_\batt^1$ is fixed by Eq.~\eqref{eq:battery final state},
\item the state of the machine is preserved,
\item the final state is of tensor product form, i.e. $\rho_\CMW^1 = \rho_\cold^1\otimes\rho_\mach^1\otimes\rho_\batt^1$,
\end{enumerate}
then the efficiency for a particular transformation $\rho^0_{\cold\hot \textup{MW}} \rightarrow \rho^1_{\cold\hot \textup{MW}}$ simplifies to being only an explicit function of $\rho^1_\cold$ instead of the global final state. This simplified expression of the efficiency in Eq.~\eqref{eq:eff explicit function of cold bath} is then used to evaluate, for example, $\eta^\textup{mac} (\rho_\cold^1)$ in Eq.~\eqref{eq:max eff macro function of rho_C^1}.
Since we now drop Assumption (iii) for the final state being uncorrelated, the efficiency and the work extracted $W_\textup{ext}$ will 
now depend on the tripartite final state $\rho^1_{\cold\textup{M}\textup{W}}$ instead.

Therefore, let us first write a generalized expression for the maximum efficiency corresponding to a transition $\rho^0_{\cold\hot \textup{MW}} \rightarrow \rho^1_{\cold\hot \textup{MW}}$ via the unitary operator $U(t)$ in this generalised setting:
\begin{align}\label{eq:quant eff def}
\eta^\textup{qm}(\rho^1_{\cold\textup{M}\textup{W}}):=\sup_{W_\textup{ext}} \eta (\rho^1_{\cold}, W_\textup{ext})
\quad\textup{s.t.}~~&\tr_\hot [U(t)\rho_\total U(t)^\dagger] = \rho_\CMW^1,\\
& [U(t), \hat H]=0,\\
& \rho^1_\textup{W}=\varepsilon |E_j\ra\la E_j|_\textup{W}+(1-\varepsilon)|E_k\ra\la E_k |_\textup{W},\\
& \rho^1_\textup{M}=\rho^0_\textup{M}.
\end{align}
See Fig (1) in main text 
for a definition of the other quantities appearing in Eq. \eqref{eq:quant eff def}.
Recall that the definition of $\eta$ is given by $\eta=W_{\rm ext}/\Delta H$ as in Eq. \eqref{eq:efficiency}. In Section \ref{sub:finding_simplified_expression} we showed that this can be simplified to
\be\label{eq:eta gain for corr}
\eta=(1-\varepsilon+\Delta C/W_\textup{ext})^{-1},
\ee
where $\Delta C=\Delta C(\rho^1_\cold)$. This equation holds under Assumption (i) and (ii), together with the fact that the global Hamiltonian does not contain interaction terms between both baths, battery, and machine. Since the derivation of Eq.~\eqref{eq:eta gain for corr} does not require Assumption (iii), it still holds for a general tripartite final state $\rho_\CMW^1$. However, dropping Assumption (iii) may potentially allow for larger values of $W_\textup{ext}$, and therefore subsequently might affect $\eta^\textup{qm}$. For this reason we write $\eta^\textup{qm}=\eta^\textup{qm}(\rho^1_{\cold\textup{M}\textup{W}})$ to remind ourselves that it is a function of the entire final state $\rho_\CMW^1$. 

We have written $\eta=\eta(\rho^1_\cold,W_\textup{ext})$ to explicitly show the $W_\textup{ext}$ dependency of $\eta$. Although not written explicitly in Eq. \eqref{eq:quant eff def}, we should remember that $U(t),\rho^0_\textup{M},\hat H_\hot$ and $\hat H_\textup{M}$ are arbitrary, other than satisfying condition \textbf{(A.4)} in Section \ref{section:The setting}. As such, by maximizing $\eta$ over $W_\textup{ext}$, these quantities will accommodate their optimal values to maximize $\eta^\textup{qm}(\rho^1_{\cold\textup{M}\textup{W}})$\footnote{This is an advantage, since it rules out cases such as when the Hamiltonian does not support a thermal state (e.g. when the corresponding thermal state's partition function diverges). In this section we consider any cold bath Hamiltonian $\hat H_\cold$ that satisfies \textbf{(A.6)} in Section \ref{section:The setting} (i.e. finite dimensional). As such it will always have a well defined thermal state. }.
Throughout this section, we analyze Eq. \eqref{eq:quant eff def} only in the case of \emph{near perfect work} (Def. \eqref{def:near perfect work}) since the proof that perfect work is not possible (see Lemma \ref{lem: lem 1}) also applies to Eq. \eqref{eq:quant eff def}\footnote{For the sake of full generality, some of the lemmas in this section are proven irrespective to whether we are considering perfect or near perfect work}.


For the purpose of our proofs, we need to define a new family of intermediate efficiencies. They provide the maximum possible efficiency, when considering only a particular instance $\alpha\geq 0$ of the generalized second laws. For any $\alpha\in [0,\infty)$, let us denote
\begin{align}\label{eq:intermediate efficiencies}
\eta^\textup{qm}_\alpha(\rho_{\cold\textup{M} \textup{W}}^1)=\sup_{W_\textup{ext}}\eta
(\rho_{\cold}^1, W_\textup{ext}) \quad \text{s.t.}~~
&F_\alpha(\tau_\cold^0\otimes \rho^0_\textup{M}\otimes\rho^0_\textup{W},\tau^h_{\cold\textup{M}\textup{W}})\geq F_\alpha(\rho^1_{\cold\textup{M}\textup{W}},\tau^h_{\cold\textup{M}\textup{W}}),& 
\\& \tr(\hat H_t \rho_\sttotal^0)=\tr(\hat H_t \rho_\sttotal^1),\label{eq:mean enrg contrain intermediate}&\\
 &\rho^1_\textup{W}=\varepsilon |E_j\ra\la E_j|_\textup{W}+(1-\varepsilon)|E_k\ra\la E_k |_\textup{W},&\\
 & \rho^1_\textup{M}=\rho^0_\textup{M}.
\end{align}
See Eq. \eqref{eq:generalfreeenergy} for definition of $F_\alpha$. We denote $\eta_\infty^\textup{qm}=\displaystyle\lim_{\alpha\rightarrow\infty}\eta_\alpha^\textup{qm}$. The condition Eq. \eqref{eq:mean enrg contrain intermediate}, is always satisfied when all the second laws are satisfied. We add the condition as a constraint here, since we will need it in order to write the efficiency $\eta$ in the form of Eq. \eqref{eq:eta gain for corr}.  

\subsubsection{Final correlations do not allow the surpassing of Carnot efficiency}\label{subsub:finalcorr_cannotsurpassCE}
Here, we show that Carnot efficiency cannot be surpassed in a quasi-static heat engine even when we allow arbitrary final correlations in the final state $\rho_\CMW^1$. This can be done in the following steps:
\begin{enumerate}
\item Using the definitions of generalized efficiency (allowing correlations) in Eq.~\eqref{eq:quant eff def} and generalized intermediate efficiencies in Eq.~\eqref{eq:intermediate efficiencies}, we prove an inequality between $\eta^\textup{qm}(\rho^1_{\cold\textup{M}\textup{W}})$ and $\eta^\textup{qm}_\alpha(\rho^1_{\cold\textup{M}\textup{W}})$, for all $\alpha\geq 0$. This is done in Lemma \ref{lem:alpha up q eff bound}. From this, we also conclude that $\eta^\textup{qm}(\rho^1_{\cold\textup{M}\textup{W}})\leq \eta^\textup{qm}_1(\rho^1_{\cold\textup{M}\textup{W}})$. 
\item On the other hand, we show that for any final state of the cold bath, machine and battery $\rho_\CMW^1$, the generalized intermediate efficiency for $\alpha =1$ only increases, if we consider the tensor product of the marginals $\rho_\CMW^1$. In other words, $\eta_1^\qm (\rho_\CMW^1) \leq \eta_1^\qm (\rho_\cold^1\otimes\rho_\batt^1\otimes\rho_\mach^1)$. One can intuitively see why this is true: it comes from the fact that the von Neumann entropy is subadditive, therefore the final state $\rho_\cold^1\otimes\rho_\batt^1\otimes\rho_\mach^1$ contains more entropy than $\rho_\CMW^1$. Therefore according to the $\alpha=1$ second law, one can potentially draw more work by going to the state $\rho_\cold^1\otimes\rho_\batt^1\otimes\rho_\mach^1$ instead of a correlated state $\rho_\CMW^1$. 
\item Since the argument for $\eta_1^\qm (\rho_\cold^1\otimes\rho_\batt^1\otimes\rho_\mach^1)$ is of tensor product form, Assumption (iii) holds as before, and therefore the efficiency only depends on the final state of the cold bath $\rho_\cold^1$. This means that Eq.~\eqref{eq:intermediate efficiencies} for $\alpha =1 $ reduces to Eq.~\eqref{eq:max eff macro function of rho_C^1}. Lastly, by using Lemma \ref{lem:general_batt}, this allows us to further show in Lemma \ref{lem:coherences cannot help fixed state} that even by allowing correlations in $\rho_\CMW^1$, the efficiency $\eta^\textup{qm}(\rho^1_{\cold\textup{M}\textup{W}})$ can never surpass the Carnot value.
\end{enumerate}

Firstly, let us fix the following notation:
for an $R$-partite state $\rho_{\textup{A}_1 \textup{A}_2 \ldots \textup{A}_R}$, define the uncorrelated counterpart
\be 
\underline{\rho_{\textup{A}_1 \textup{A}_2 \ldots \textup{A}_R}}:=\bigotimes_{i=1}^R \rho_{\textup{A}_i}.
\ee
Comparing $\rho_{\textup{A}_1 \textup{A}_2 \ldots \textup{A}_R}$ and $\underline{\rho_{\textup{A}_1 \textup{A}_2 \ldots \textup{A}_R}}$, one will see that each subsystem has the same reduced state, but the global state is different. Another useful thing is to note that if one is given a Hamiltonian which does not contain any interaction terms between each subsystem, i.e. 
\begin{equation}
\hat H_{\textup{A}_1 \textup{A}_2 \ldots \textup{A}_R} = \displaystyle\sum_{i=1}^R \id_{\textup{A}_1}\otimes\cdots\hat H_{\textup{A}_i}\cdots\id_{\textup{A}_R},
\end{equation}
then we may conclude that 
\begin{align}\label{eq:nonintH_TPstate}
\tr (\hat H_{\textup{A}_1 \textup{A}_2 \ldots \textup{A}_R} \rho_{\textup{A}_1 \textup{A}_2 \ldots \textup{A}_R}) &= \sum_{i=1}^R \tr (\hat H_{\textup{A}_i} \rho_{\textup{A}_i}) =  \sum_{i=1}^R \tr (\hat H_{\textup{A}_i} \underline{\rho_{\textup{A}_i}}) = 
\tr (\hat H_{\textup{A}_1 \textup{A}_2 \ldots \textup{A}_R} \underline{\rho_{\textup{A}_1 \textup{A}_2 \ldots \textup{A}_R}}).
\end{align}

\begin{lemma}\label{lem:alpha up q eff bound}
For all $\alpha\geq 0$ and all states $\rho^1_{\cold\hot \textup{MW}}$, 
\be\label{eq:quant upper bound}
\eta^\textup{qm}(\rho^1_{\cold\textup{M}\textup{W}})\leq \eta_\alpha^\textup{qm}(\rho^1_{\cold\textup{M}\textup{W}}),
\ee
where $\eta^\textup{qm}$ and $\eta_\alpha^\textup{qm}$ are defined in Eqs. \eqref{eq:quant eff def} and \eqref
{eq:intermediate efficiencies} respectively.
\end{lemma}
\begin{proof}
For every $\alpha\geq 0$, Eq. $F_\alpha(\tau_\cold^0\otimes \rho^0_\textup{M} \otimes\rho^0_\textup{W},\tau^h_{\cold\textup{M}\textup{W}})\geq  F_\alpha(\rho^1_{\cold\textup{M}\textup{W}},\tau^h_{\cold\textup{M}\textup{W}})$ in Eq. \eqref{eq:intermediate efficiencies} is a necessary condition for the transformation $\rho^0_{\cold\textup{M}\textup{W}} \rightarrow \rho^1_{\cold\textup{M}\textup{W}}$ to occur under an energy preserving unitary with the aid of a catalyst \cite{2ndlaw}. Energy preserving unitaries also preserve the average energy and thus the Eq. $\tr(\hat H_t \rho_\sttotal^0)=\tr(\hat H_t \rho_\sttotal^1)$  
in Eq. \eqref{eq:intermediate efficiencies} is also a necessary condition. 
If a unitary $U(t)$ satisfies the conditions in Eq.~\eqref{eq:quant eff def}, then by the second laws it satisfies Eq.~\eqref{eq:intermediate efficiencies} for any particular $\alpha\geq 0$.
As a consequence of these observations, the set of allowed unitaries $U(t)$ in Eq. \eqref{eq:quant eff def} is a subset of allowed unitaries facilitating the catalytic thermal operation which transforms $\rho^0_{\cold\textup{M}\textup{W}}$ to $\rho^1_{\cold\textup{M}\textup{W}}$ in Eq. \eqref{eq:intermediate efficiencies}.
\end{proof}

\begin{lemma}\label{lem:eta1_achieved_with_equality}
For any final state $\rho_\CMW^1$, consider the quantity $\eta_1^\textup{qm} (\rho_\CMW^1)$ defined in Eq.~\eqref{eq:intermediate efficiencies}. Consider the optimization problem
\begin{align}
a(\rho_{\cold\textup{M} \textup{W}}^1):=\sup_{W_\textup{ext}}\eta
(\rho_{\cold}^1, W_\textup{ext}) \quad \text{s.t.}
~~ &F_1(\tau_\cold^0\otimes\rho^0_\textup{M}\otimes\rho^0_\textup{W},\tau^h_{\cold\textup{M}\textup{W}})= F_1(\rho^1_{\cold\textup{MW}},\tau^h_{\cold\textup{MW}}),\\
&\tr(\hat H_t \rho_\sttotal^0)=\tr(\hat H_t \rho_\sttotal^1),\\
 &\rho^1_\textup{W}=\varepsilon |E_j\ra\la E_j|_\textup{W}+(1-\varepsilon)|E_k\ra\la E_k |_\textup{W},&\\
 & \rho^1_\textup{M}=\rho^0_\textup{M}.
\end{align}
Then, $\eta_1^\textup{qm} (\rho_\CMW^1) = a (\rho_\CMW^1)$. 
\end{lemma}
\begin{proof}
We begin by noting that the free energy $F_1$ can be written as
\be\label{eq:generalised free energy}
F_1(\rho,\tau^h)
=\tr( \hat H \rho)
 -\frac{1}{\beta_h} S(\rho),
\ee
where $\langle \hat H\rangle_\rho:=\tr[\hat H\rho]$, and $S(\rho)=-\tr(\rho\ln\rho)$ is the von Neumann entropy, while $\tau^h$ is the thermal state at inverse temperature $\beta_h$ for the Hamiltonian $\hat H$. Also, let us recall that $W_\textup{ext}=E_k^\textup{W}-E_j^\textup{W}>0$ where $E_j^\textup{W}$ is a constant.

Next, we consider the free energies $F_1(\tau_\cold^0\otimes\rho^0_\textup{M}\otimes\rho^0_\textup{W},\tau^h_{\cold\textup{M}\textup{W}})$ and $F_1(\rho^1_{\cold\textup{MW}},\tau^h_{\cold\textup{M}\textup{W}})$ respectively, and how they relate to $W_\textup{ext}$.
First of all, note that the quantity $F_1(\tau_\cold^0\otimes\rho^0_\textup{M}\otimes\rho^0_\textup{W},\tau^h_{\cold\textup{M}\textup{W}})$ is simply a constant that does not depend on $W_\textup{ext}$. This is because
\begin{align}
F_1(\tau_\cold^0\otimes\rho^0_\textup{M}\otimes\rho^0_\textup{W},\tau^h_{\cold\textup{M}\textup{W}}) 
&= F_1(\tau_\cold^0,\tau^h_{\cold}) + F_1(\tau_\mach^0,\tau^h_{\mach}) + F_1(\tau_\batt^0,\tau^h_{\batt}) \\
&= F_1(\tau_\cold^0,\tau^h_{\cold}) + F_1(\tau_\mach^0,\tau^h_{\mach}) + \tr (\hat H_\batt \rho_\batt^0) - \beta_h^{-1} S(\rho_\batt^0)\\
&= F_1(\tau_\cold^0,\tau^h_{\cold}) + F_1(\tau_\mach^0,\tau^h_{\mach}) + E_j^\batt,
\end{align}
where the first two terms do not depend on the battery Hamiltonian at all, while in the last equality we have made use of the fact that $\rho_\batt^0 = \ketbra{E_j}{E_j}_\batt$.
On the other hand, 
\begin{align}
F_1(\rho^1_\CMW,\tau^h_{\cold\textup{M}\textup{W}}) &= \tr(\hat H_\CMW \rho_\CMW^1) - \beta_h^{-1} S(\rho_\CMW^1)\\
&=\tr(\hat H_\cold \rho_\cold^1)+\tr(\hat H_\mach \rho_\mach^1)+\tr(\hat H_\batt \rho_\batt^1)- \beta_h^{-1} S(\rho_\CMW^1)\\
&= \tr(\hat H_\cold \rho_\cold^1)+\tr(\hat H_\mach \rho_\mach^1)- \beta_h^{-1} S(\rho_\CMW^1)+\varepsilon E_j^\batt + (1-\varepsilon) E_k^\batt.
\end{align}
Note that again, $\tr(\hat H_\cold \rho_\cold^1)$ and $\tr(\hat H_\mach \rho_\mach^1)$ do not depend on the battery Hamiltonian and therefore do not depend on $E_k^\batt$. Similarly, $S(\rho_\CMW^1)$ depends only on the eigenvalues of the state, and is independent of $E_k^\batt$. Since $\varepsilon\in [0,1)$, we may conclude the following: \textit{$F(\rho^1_\CMW,\tau^h_{\cold\textup{M}\textup{W}})$ is a continuous function that strictly increases w.r.t. $E_k^\batt$, and therefore it also strictly increases w.r.t. $W_\textup{ext}$.}

To prove this lemma, it suffices to show that the supremum over $W_\textup{ext}$ in Eq.~\eqref{eq:intermediate efficiencies} for $\alpha=1$ has to be achieved when $F_1(\tau_\cold^0\otimes\rho^0_\textup{M}\otimes\rho^0_\textup{W},\tau^h_{\cold\textup{M}\textup{W}})= F_1(\rho^1_{\cold\textup{MW}},\tau^h_{\cold\textup{MW}})$. We prove this by contradiction. Suppose that $\hat W_\textup{ext}$ achieves the supremum for $\eta_1^\qm$, and for this value of $\hat W_\textup{ext}$, $F_1(\tau_\cold^0\otimes\rho^0_\textup{M}\otimes\rho^0_\textup{W},\tau^h_{\cold\textup{M}\textup{W}}) > F_1(\rho^1_{\cold\textup{MW}},\tau^h_{\cold\textup{MW}})$. Since we know that $F(\rho^1_\CMW,\tau^h_{\cold\textup{M}\textup{W}})$ strictly increases w.r.t. $W_\textup{ext}$, there must exist an $W_\textup{ext}' > \hat W_\textup{ext}$ such that $F_1(\tau_\cold^0\otimes\rho^0_\textup{M}\otimes\rho^0_\textup{W},\tau^h_{\cold\textup{M}\textup{W}})\geq F_1(\rho^1_{\cold\textup{MW}},\tau^h_{\cold\textup{MW}})$. Furthermore, since by Eq.~\eqref{eq:eta gain for corr} we know that the efficiency is monotonically increasing w.r.t. $W_\textup{ext}$ as well, it follows that $W_\textup{ext}'$ achieves a higher value of efficiency compared to $\hat W_\textup{ext}$ while satisfying the required constraints at the same time. This is a contradiction, and therefore we conclude that the optimization for $\eta_1^\qm$ can be simplified to $a(\rho_\CMW^1)$, where the constraint on $F_1$ holds with equality.
\end{proof}
\begin{lemma}\label{lem:coherences cannot help fixed state}
For any final state $\rho^1_{\cold\hot \textup{MW}}$, and any Hamiltonian of the form in Eq.~\eqref{eq:totalH}, 
then for perfect or near perfect work extraction (see Defs. \ref{def:perfect work} and \ref{def:near perfect work}), we have
\be \label{eq:equation main lemma in extension}
\eta^\textup{qm}\left(\rho^1_{\cold\textup{M} \textup{W}}\right)
	~~\overset{\textup{(1)}}{\leq}~~ \eta_1^\textup{qm}\left(\rho^1_{\cold\textup{M} \textup{W}}\right)
	~~\overset{\textup{(2)}}{\leq}~~ \eta_1^\textup{qm}\left(\underline{\rho^1_{\cold\textup{M} \textup{W}}}\right)		
	~~\overset{\textup{(3)}}{=}~~  \eta^\textup{mac}\left(\underline{\rho_\cold^1}\right)
	~~\overset{\textup{(4)}}{\leq}~~ 1-\frac{\beta_h}{\beta_c},
\ee
with equality in \textup{(2)} iff $\rho^1_{\cold\textup{M} \textup{W}}=\underline{\rho^1_{\cold\textup{M} \textup{W}}}$.
The quantities $\eta_1^\textup{qm}$ and $\eta^\textup{mac}$ are defined in Eq. \eqref{eq:intermediate efficiencies} and Eq. \eqref{eq:max eff macro function of rho_C^1} respectively.
\end{lemma}
\begin{proof}
Note that inequality (1) is a direct consequence of Lemma \ref{lem:alpha up q eff bound}, while inequality (4) holds because of Lemma \ref{lem:CEmaximum_generalbatt}. It remains to prove inequalities (2) and (3).\\
\textit{Proof of inequality (2):}
Using the definition in Eq.~\eqref{eq:intermediate efficiencies} together with Lemma \ref{lem:eta1_achieved_with_equality}, let us compare the quantities
\begin{align}\label{eq:eta1qmrho}
\eta^\textup{qm}_1(\rho_{\cold\textup{M} \textup{W}}^1)=\sup_{W_\textup{ext}}\eta
(\rho_{\cold}^1, W_\textup{ext}) \quad \text{s.t.}~~
&F_1(\tau_\cold^0\otimes \rho^0_\textup{M}\otimes\rho^0_\textup{W},\tau^h_{\cold\textup{M}\textup{W}})= F_1(\rho^1_{\cold\textup{M}\textup{W}},\tau^h_{\cold\textup{M}\textup{W}}),& 
\\& \tr(\hat H_t \rho_\sttotal^0)=\tr(\hat H_t \rho_\sttotal^1),&\\
 &\rho^1_\textup{W}=\varepsilon |E_j\ra\la E_j|_\textup{W}+(1-\varepsilon)|E_k\ra\la E_k |_\textup{W},\label{eq:batt_constraint}\\
 & \rho^1_\textup{M}=\rho^0_\textup{M},\label{eq:mach_constraint}
\end{align}
and 
\begin{align}\label{eq:eta1qmrhoTP}
\eta^\textup{qm}_1(\underline{\rho_{\cold\textup{M} \textup{W}}^1})=\sup_{W_\textup{ext}}\eta
(\underline{\rho_{\cold}^1}, W_\textup{ext}) \quad \text{s.t.}~~
&F_1(\tau_\cold^0\otimes \rho^0_\textup{M}\otimes\rho^0_\textup{W},\tau^h_{\cold\textup{M}\textup{W}})= F_1(\underline{\rho^1_{\cold\textup{M}\textup{W}}},\tau^h_{\cold\textup{M}\textup{W}}),& 
\\& \tr(\hat H_t \rho_\sttotal^0)=\tr(\hat H_t \rho_\sttotal^1),&\\
 &\rho^1_\textup{W}=\varepsilon |E_j\ra\la E_j|_\textup{W}+(1-\varepsilon)|E_k\ra\la E_k |_\textup{W},&\\
 & \rho^1_\textup{M}=\rho^0_\textup{M}.
\end{align}
We first make the following observations:
\begin{itemize}
\item By our definition of $\underline{\rho_{\cold\textup{M} \textup{W}}^1}$, we have that $\rho_\cold^1 = \underline{\rho_\cold^1}$. Therefore, the term $\Delta C$ in Eq.~\eqref{eq:eta gain for corr} which is only a function of the reduced state on the cold bath is the same for both efficiencies in Eq.~\eqref{eq:eta1qmrho} and Eq.~\eqref{eq:eta1qmrhoTP}. Therefore, to compare the efficiencies, we need only to compare the value of $W_\textup{ext}$ that satisfies the free energy constraint in both optimization problems.
\item In \cite{shannon} (pg. 395) it has been proven that the von Neumann entropy is subadditive
\be 
S(\rho_{AB})\leq S\left(\underline{\rho_{AB}}\right),
\ee
with equality iff $\rho_{AB}=\underline{\rho_{AB}}$. Furthermore, since $\hat H_\CMW$ does not contain interaction terms, as we have demonstrated earlier in Eq.~\eqref{eq:nonintH_TPstate},
\begin{align}
\tr (\hat H_\CMW \rho_\CMW^1) 
&= \tr (\hat H_\CMW \underline{\rho_\CMW^1}).
 \end{align} 
 Thus, by Eq. \eqref{eq:generalised free energy} we conclude that
\be \label{eq:subadd free energy}
F_1(\underline{\rho^1_{\cold\textup{M}\textup{W}}})\leq  F_1(\rho^1_{\cold\textup{M}\textup{W}}),
\ee
with equality iff $\rho^1_{\cold\textup{M}\textup{W}}=\underline{\rho^1_{\cold\textup{M}\textup{W}}}$.
\item For any final state $\rho_\CMW^1$ where $\rho_\batt^1 = \varepsilon |E_j\ra\la E_j|_\textup{W}+(1-\varepsilon)|E_k\ra\la E_k |_\textup{W}$, we have seen in the proof of Lemma \ref{lem:eta1_achieved_with_equality} that $F_1 (\rho_\CMW^1,\tau^h_{\cold\textup{M}\textup{W}})$ is a continuous function that strictly increases with $W_\textup{ext}$.
\end{itemize}

With these three observations we can now prove inequality (2). Note that when $\rho_\CMW^1=\underline{\rho_\CMW^1}$, equality holds trivially. Therefore, let us consider the case where $\rho_\CMW^1\neq\underline{\rho_\CMW^1}$.
Suppose $\hat W_\textup{ext}$ achieves the supremum in $\eta_1^\qm (\rho_\CMW^1)$, and for such a value of $\hat W_\textup{ext}$, \begin{equation}
F_1(\tau_\cold^0\otimes \rho^0_\textup{M}\otimes\rho^0_\textup{W},\tau^h_{\cold\textup{M}\textup{W}})= F_1(\rho^1_{\cold\textup{M}\textup{W}},\tau^h_{\cold\textup{M}\textup{W}}) > F_1(\underline{\rho^1_{\cold\textup{M}\textup{W}}},\tau^h_{\cold\textup{M}\textup{W}}).
\end{equation}
We note also that since $F_1 (\underline{\rho_\CMW^1},\tau^h_{\cold\textup{M}\textup{W}})$ strictly increases with $W_\textup{ext}$, and therefore there exists some $W_\textup{ext}' > \hat W_\textup{ext}$ such that $F_1(\tau_\cold^0\otimes \rho^0_\textup{M}\otimes\rho^0_\textup{W},\tau^h_{\cold\textup{M}\textup{W}}) = F_1(\underline{\rho^1_{\cold\textup{M}\textup{W}}},\tau^h_{\cold\textup{M}\textup{W}})$. Therefore, $W_\textup{ext}'$ is a feasible solution for Eq.~\eqref{eq:eta1qmrhoTP}, i.e. it satisfies the constraints in the optimization problem. In conclusion, we have

\begin{align}
\eta^\qm_1(\rho_\CMW^1) &= \left[1-\varepsilon + \frac{\Delta C}{\hat W_\textup{ext}}\right]^{-1}
\leq \left[1-\varepsilon + \frac{\Delta C}{W_\textup{ext}'}\right]^{-1}
\leq \eta^\qm_1(\underline{\rho_\CMW^1}).
\end{align}
\textit{Proof of equality (3):} 
Consider the quantity $\eta^\qm_1(\underline{\rho^1_{\cold\textup{M}\textup{W}}})$. Since the state $\underline{\rho^1_{\cold\textup{M}\textup{W}}}$ takes on a product structure form between all the subsystems now, Assumption (iii) in the beginning of Section \ref{subsub:def_gen_eff} holds again. By Eqns.~\eqref{eq:batt_constraint} and \eqref{eq:mach_constraint}, we know that Assumptions (i) and (ii) also hold. Therefore, we know that under these assumptions the efficiency does not depend anymore on the global state $\underline{\rho^1_{\cold\textup{M}\textup{W}}}$, but only $\rho_\cold^1$. Again comparing the conditions of $\eta^\textup{mac} (\underline{\rho_\cold^1})$ and $\eta_1^\qm (\underline{\rho^1_{\cold\textup{M}\textup{W}}})$, we see that they are exactly the same quantity.
\end{proof}

Therefore, Lemma \ref{lem:coherences cannot help fixed state} tells us that correlations between the final states of the cold bath, machine and battery cannot allow you to achieve an efficiency greater than the Carnot efficiency.

\subsubsection{Achievability of Carnot efficiency still depends on more than temperature}\label{subsub:achievability_fundamentally_reduced}
Earlier in Section \ref{subsub:finalcorr_cannotsurpassCE}, we proved in Lemma \ref{lem:coherences cannot help fixed state} that Carnot efficiency gives an upper bound to the efficiency of any arbitrary final state $\rho_\CMW^1$. In this section, we want to prove that when $\Omega>1$ holds, quasi-static heat engines cannot achieve the Carnot efficiency even when allowing correlations between the final states of the battery and the cold bath. This can be done in the following steps:
\begin{itemize}
\item According to Lemma \ref{lem:coherences cannot help fixed state}, Carnot efficiency can be attained only when all the inequalities in Eq.~\eqref{eq:equation main lemma in extension} are satisfied with equalities. We use this to prove in Lemma \ref{Lemm:vanishing coors necesary condition} that in order to achieve the Carnot efficiency, we may only consider the limit where correlations in the final state vanish. Not only so, the magnitude of these correlations also have to vanish quickly enough in order for Carnot efficiency to be achieved. In particular, we define a parameter $k$ which quantifies the amount of correlations, and show that $k$ has to vanish faster than the quasi-static parameter $g$, in order to achieve the Carnot efficiency $\eta_C$.
\item Next, in Lemma \ref{lem:corres last lemma}, we show that if the parameter $k$ vanishes faster than the quasi-static parameter $g$, then whenever $\Omega >1$, one can derive an upper bound for the intermediate efficiency $\eta_\infty^\qm (\rho_\CMW^1)$ which considers the amount of work extractable by invoking only the generalized second law of $\alpha\rightarrow\infty$. Combining Lemma \ref{Lemm:vanishing coors necesary condition} and Lemma \ref{lem:corres last lemma}, we conclude in Corollary \ref{corrs dont help} that when $\Omega >1$, $\eta^\qm\leq \eta_\infty^\qm \leq \eta_C$ is strictly upper bounded away from the Carnot efficiency. 
\end{itemize}

Before we begin, let us note that by definition, the initial state $\rho^0_{\cold\textup{W}}$ is diagonal in its energy eigenbasis. Furthermore, the state $\rho^0_{\cold\textup{M}\textup{W}}$ is of the form $\rho^0_\cold\otimes\rho^0_\textup{M}\otimes\rho^0_\textup{W}$. Since w.l.o.g. we can assume that $\hat H_\textup{M}$ is proportional to the identity (or called the \textit{trivial Hamiltonian}, see \cite{2ndlaw}), $\rho^0_\textup{M}$ can always be written as a diagonal state in an energy eigenbasis of its Hamiltonian. Therefore the state $\rho^0_{\cold\textup{M}\textup{W}}$ is always diagonal in the energy eigenbasis of the Hamiltonian $\hat H_{\cold\textup{M}\textup{W}}:=\hat H_\cold+\hat H_\textup{M}+\hat H_\textup{W}$. Since catalytic thermal operations cannot create coherences \cite{2ndlaw}, $\rho^1_{\cold\textup{M}\textup{W}}$ has to be also diagonal in the energy eigenbasis of $\hat H_{\cold\textup{M}\textup{W}}$. 

We observe that any $\rho_\CMW^1$ can always be written as 
\begin{equation}
\rho_\CMW^1 = (1-k^*) \underline{\rho_\CMW^1} + k^*\rho_\CMW^\corr,
\end{equation}
where $k^* = \min \lbrace k\in[0,1] | \rho_\CMW^1 = (1-k) \underline{\rho_\CMW^1} + kQ, Q\geq \mathbf{0}\rbrace$. This means that $\rho_\CMW^1$ can be written as a convex combination of two states: one being $\underline{\rho_\CMW^1}$, and the other $\rho_\CMW^\corr$ containing all other correlations. Note that such a $k^*$ always exists, in particular, $k=1$ is always a feasible solution.

We now define a particular parametrization of the final states, 
\be\label{eq:rho 1 corr def}
\rho^1_{\cold\textup{M}\textup{W}}(k,\rho^\textup{\ncorr}_{\cold\textup{M}\textup{W}},\rho^\textup{\corr}_{\cold\textup{M}\textup{W}}):= (1-k)\rho^\textup{\ncorr}_{\cold\textup{M}\textup{W}}+k\rho^\textup{corr}_{\cold\textup{M}\textup{W}},\quad k\in[0,k^*]
\ee
where the following holds:
\begin{align}
\textup{(i)}&~~\rho^\ncorr_\CMW= \underline{\rho_\CMW^1},
\label{eq:def of b_i 2}\\
\textup{(ii)}&~~\rho^\corr_\CMW \neq \rho^\ncorr_{\cold\textup{M}\textup{W}} ,\label{eq: no corr corr neq}\\
\textup{(iii)}&~~\rho_\mach^1 = (1-k)\rho_\mach^\ncorr+ k \rho_\mach^\corr = \rho_\mach^0.\label{eq:preservemachine}
\end{align}
Since in our heat engine, the initial state has no coherences, it suffices to consider $\rho_\CMW^1$ which is diagonal in the energy eigenbasis. This implies that $\rho_\CMW^\ncorr = \underline{\rho_\CMW^1}$ is also diagonal in the energy eigenbasis, and therefore the same holds for $\rho_\CMW^\corr$ due to Eq.~\eqref{eq:rho 1 corr def}.
All correlations between the individual systems of cold bath, machine and battery are contained only in $\rho^\corr_{\cold\textup{M}\textup{W}}$.
Therefore, $\rho^1_{\cold\textup{M}\textup{W}}(\cdot,\cdot,\cdot)$ parametrizes every possible quantum state on $\mathcal{H}_{\cold\textup{M}\textup{W}}$ which is diagonal in the global energy eigenbasis and that returns the machine locally to its initial state after one cycle of the heat engine. 
In Eq. \eqref{eq:preservemachine}, $\rho_\mach^1$ is the final state of the machine, since the heat engine is cyclic, recall from Section \ref{section:The setting} that we require $\rho_\textup{M}^1=\rho_\textup{M}^0$. 
\begin{lemma}\label{Lemm:vanishing coors necesary condition}
For every family of states $\rho^1_{\cold\textup{M}\textup{W}}(k,\rho^\textup{\ncorr}_{\cold\textup{M}\textup{W}},\rho^\textup{\corr}_{\cold\textup{M}\textup{W}})$ parametrized by $k$, (see Eqs. \eqref{eq:rho 1 corr def}-\eqref{eq:preservemachine}), if the quantum efficiency $\eta^\qm_1$ defined in Eq. \eqref{eq:intermediate efficiencies} of a quasi-static heat engine achieves the Carnot efficiency
\be
\eta_1^\qm(\rho^1_\CMW)=1-\frac{\beta_h}{\beta_c},
\ee
then the following conditions are satisfied:
\begin{itemize}
\item [1)]  The state $\underline{\rho^1_\CMW}$ is the final state of a quasi-static heat engine (see Def. \ref{def:quasi static})
\be \label{eq:quasistatic no corrs part}
\underline{\rho^1_\CMW}= \tau(g)\otimes\rho^0_\textup{M}(g)\otimes \rho^1_\textup{W}\quad  \text{with } g\rightarrow 0^+.
\ee 
\item [2)]  The correlations must vanish sufficiently quickly. That is to say, the parameter $k$ in Eq.~\eqref{eq:rho 1 corr def}
vanishes more quickly compared to $g$, i.e.  
\be\label{eq:necessary k_i/g lim}
\lim_{g\rightarrow 0^+}\frac{k}{g}=0.
\ee
\end{itemize}
\end{lemma} 
\begin{proof}
Firstly, suppose that Carnot efficiency is achieved, i.e. $\eta^\qm (\rho_\CMW^1) = 1 - \frac{\beta_h}{\beta_c}$. Then according to Lemma \ref{lem:coherences cannot help fixed state}, all inequalities in Eq.~\eqref{eq:equation main lemma in extension} should be satisfied with equality, in particular inequality (4). We have established in Lemma \ref{lem:CEmaximum_generalbatt} that this equality is achieved in the quasi-static limit, i.e. $\rho_\cold^1 = \tau_\cold (g)$ where $g\rightarrow 0^+$. This implies Condition 1) in the statement of the lemma.

The proof for Condition 2) consists of calculating $W_\textup{ext}$ for $\alpha=1$ in Eq. \eqref{eq:intermediate efficiencies} to leading order in $g$ and $k$. This $W_\textup{ext}$ quantity can be later used to evaluate $\eta_1^\qm$. We will show that we can write the expression for $\eta_1^\qm$ into two terms: one term describes the efficiency when there are no final correlations, and the other term is a strictly negative contribution which must vanish in order to achieve the Carnot efficiency. This latter constraint will give us Eq. \eqref{eq:necessary k_i/g lim}.

Let us denote $W_\textup{ext}'$ as the value of battery energy gap $W_\textup{ext} = E_k^\batt - E_j^\batt$ that solves the equation
\be\label{eq:F_1 for k_i/g necessity lemma 1}
F_1(\tau_\cold^0\otimes \rho^0_\textup{M}(g)\otimes\rho^0_\textup{W},\tau^h_{\cold\textup{M}\textup{W}})=  F_1(\rho^1_\CMW(k,\rho^\textup{\ncorr}_{\cold\textup{M}\textup{W}},\rho^\corr_\CMW),\tau^h_{\cold\textup{M}\textup{W}}) \footnote{We denote $\rho_\mach^0 (g)$ because for different values of $g$, we are allowed to choose different initial machine states, as long as $\rho_\mach^1 (g) = \rho_\mach^0 (g) $.},
\ee
while $\hat W_\textup{ext}$ as the value that solves the case where $k=0$, i.e.
\be\label{eq:F_1 for k_i/g necessity lemma 2}
F_1(\tau_\cold^0\otimes \rho^0_\textup{M}(g)\otimes\rho^0_\textup{W},\tau^h_{\cold\textup{M}\textup{W}})=  F_1(\rho^\ncorr_\CMW,\tau^h_{\cold\textup{M}\textup{W}}).
\ee
Since $\rho^\ncorr_\CMW = \underline{\rho_\CMW^1} = \rho_\cold^1\otimes\rho_\mach^0 (g)\otimes\rho_\batt^1$ contains no correlations, $\hat W_\textup{ext}$ was given by Eq. \eqref{eq:W ext std}.
According to Lemma \ref{lem:eta1_achieved_with_equality}, we know that $W_\textup{ext}'$ and $\hat W_\textup{ext}$ are the values of $W_\textup{ext}$ which solve  
$\displaystyle\sup_{W_\textup{ext}}\eta^\qm_1(\rho_\CMW^1 ,W_\textup{ext})$ and $\displaystyle\sup_{W_\textup{ext}}\eta^\qm_1(\underline{\rho_\CMW^1} ,W_\textup{ext})$ respectively.
 Solving Eq. \eqref{eq:F_1 for k_i/g necessity lemma 1} for $W_\textup{ext}'$ with the aid of Eq. \eqref{eq:generalised free energy}, we find
\be
W_\textup{ext}'=\hat W_\textup{ext}-\chi,
\ee
where $W_\textup{ext}$ is the solution to Eq.~\eqref{eq:F_1 for k_i/g necessity lemma 2} when $k=0$, given by Eq. \eqref{eq:W ext std}, while
\be\label{eq:chi def}
\chi
:=\frac{1}{\beta_h}\frac{1}{1-\varepsilon} \left[ S(\rho^\textup{\ncorr}_{\cold\textup{M}\textup{W}})-S\left(\rho^1_{\cold\textup{M}\textup{W}}(k,\rho^\textup{\ncorr}_{\cold\textup{M}\textup{W}},\rho^\textup{\corr}_{\cold\textup{M}\textup{W}})\right)\right].
\ee
Let us first note some properties of $\chi$, which we will later use:
\begin{itemize}
\item Since $S(\cdot)$ is subadditive, due to the parametrization of $\rho^1_{\cold\textup{M}\textup{W}}(\cdot,\cdot,\cdot)$ in Eq. \eqref{eq:rho 1 corr def}, we have 
\be\label{eq:first chi property}
\chi\geq 0
\ee
with equality iff $\underline{\rho^1_{\cold\textup{M}\textup{W}}}=\rho^1_{\cold\textup{M}\textup{W}}$ i.e. iff $k=0$. Therefore, we may conclude that $\frac{\hat W_\textup{ext}}{W_\textup{ext}'}\geq 1$.
\item We have that
\be\label{eq:second chi property}
\frac{d}{dk} \chi(k,\rho^\textup{\ncorr}_{\cold\textup{M}\textup{W}},\rho^\textup{\corr}_{\cold\textup{M}\textup{W}})\bigg{|}_{k=k_0}=0
\ee
if and only if
\be\label{eq:second chi property iff}
\rho^1_{\cold\textup{M}\textup{W}}(k_0,\rho^\textup{\ncorr}_{\cold\textup{M}\textup{W}},\rho^\textup{\corr}_{\cold\textup{M}\textup{W}})=\id_{\cold\textup{M}\textup{W}}/N.
\ee
Eqs. \eqref{eq:second chi property} and \eqref{eq:second chi property iff} are direct consequences of the observations: \\
1) Entropy is strictly concave, i.e. $S\left(\rho^1_{\cold\textup{M}\textup{W}}(k,\rho^\textup{\ncorr}_{\cold\textup{M}\textup{W}},\rho^\textup{\corr}_{\cold\textup{M}\textup{W}})\right)$ is strictly concave in $k\in[0,1]$. Therefore, by Eq.~\eqref{eq:second chi property iff} $\chi$ is strictly convex in $k\in[0,1]$. When the first derivative of the convex function $\frac{d\chi}{d k}=0$, this must be the global minimum \cite{convex_opt}. \\
2) However, we know that the entropy is uniquely maximized (and therefore $\chi$ is uniquely minimized) for the maximally mixed state.
\end{itemize}
Returning to evaluate the efficiency, we may use Eq. \eqref{eq:eta gain for corr} to calculate the inverse efficiency,
\begin{align}
[\eta^\qm_1(\rho^1_\CMW)]^{-1}&=1-\varepsilon+ \frac{\Delta C(\rho^1_{\cold})}{W'_\textup{ext}}\\
&=1-\varepsilon+\frac{\Delta C(\rho^1_{\cold})}{\hat W_\textup{ext}}\frac{\hat W_\textup{ext}}{W_\textup{ext}'}\label{one over Q eff 1 tesimal}\\
&\geq 1-\varepsilon+\frac{\Delta C(\rho^1_{\cold})}{\hat W_\textup{ext}}.
\end{align}
The last term in Eq.~\eqref{one over Q eff 1 tesimal} is non-negative because we know the terms $\Delta C (\rho_\cold^1), \hat W_\textup{ext}$ and $W_\textup{ext}'$ are all non-negative. 

With Condition 1), we now know that
\be\label{eq: c eff contraint}
1-\varepsilon+ \frac{\Delta C(\rho^1_{\cold})}{W_\textup{ext}}= 1-\frac{\beta_h}{\beta_c},
\ee
in the quasi-static limit, and therefore 
a necessary condition to achieve the Carnot efficiency is that $\lim_{g\rightarrow 0} \frac{\hat W_\textup{ext}}{W_\textup{ext}'} =1$ also in the quasi-static limit. Using the relation $ W_\textup{ext}'= \hat W_\textup{ext} + \chi$, we have the requirement that
\be\label{eq:chi frak 0 lim}
\lim_{g\rightarrow 0^+}\frac{\chi(k,\rho^\textup{\ncorr}_{\cold\textup{M}\textup{W}}(g),\rho^\textup{\corr}_{\cold\textup{M}\textup{W}})}{\hat W_\textup{ext}(\rho^\ncorr_{\cold\textup{M}\textup{W}}(g))} = 0.
\ee

First, let us observe that $\hat W_\textup{ext}(\rho^\ncorr_{\cold\textup{M}\textup{W}}(g)) = W_\textup{ext}(\beta_c-g)$ given by Eq. \eqref{eq:simplifyDeltaW}. The leading order term of $W_\textup{ext}(\beta_c-g)=\bo(g)$ as $g\rightarrow 0^+$. Therefore, in order to satisfy Eq. \eqref{eq:chi frak 0 lim}, we must firstly have $\lim_{g\rightarrow 0}\chi = 0$. From Eqs. \eqref{eq:rho 1 corr def}, \eqref{eq:chi def}, this implies that we need $k\rightarrow 0$ for all $\rho^\ncorr_{\cold\textup{M}\textup{W}}$. 

Since the numerator and denominator of Eq. \eqref{eq:chi frak 0 lim} both go to zero, by L'Hospital rule, to evaluate the limit we need to take the derivative of both terms w.r.t. $g$. Therefore, we expand $\chi$ to first order in $k$ and $g$. From Eq. \eqref{eq:chi def} it follows
\begin{align}
\chi(k,\rho^\textup{\ncorr}_{\cold\textup{M}\textup{W}}(g),\rho^\textup{\corr}_{\cold\textup{M}\textup{W}})=&\frac{d}{dk}\chi(k,\rho^\textup{\ncorr}_{\cold\textup{M}\textup{W}}(0),\rho^\textup{\corr}_{\cold\textup{M}\textup{W}})\bigg{|}_{k=0} k \nonumber\\
&+\frac{d}{dg} \chi(0,\rho^\textup{\ncorr}_{\cold\textup{M}\textup{W}}(g),\rho^\textup{\corr}_{\cold\textup{M}\textup{W}})\bigg{|}_{g=0} g+\soo (gk)+\soo(k^2)+\soo (g^2)\\
=&\frac{d}{dk} \chi(k,\rho^\textup{\ncorr}_{\cold\textup{M}\textup{W}}(0),\rho^\textup{\corr}_{\cold\textup{M}\textup{W}})\bigg{|}_{k=0} k +\soo (gk)+\soo (k^2)+\soo (g^2).
\end{align} 
The term $\frac{d}{dg} \chi(0,\rho^\textup{\ncorr}_{\cold\textup{M}\textup{W}}(g),\rho^\textup{\corr}_{\cold\textup{M}\textup{W}})\bigg{|}_{g=0} =0$ since when $k=0$, $\chi$ will be constant for all $g$.
Next, we note that since Eqs. \eqref{eq:first chi property} holds, it must be that $\frac{d}{dk} \chi(k,\rho^\textup{\ncorr}_{\cold\textup{M}\textup{W}}(0),\rho^\textup{\corr}_{\cold\textup{M}\textup{W}})\bigg{|}_{k=0}\geq 0$. Furthermore, from Eq.~\eqref{eq:second chi property}, we have that
\be 
\frac{d}{dk}\chi(k,\rho^\textup{\ncorr}_{\cold\textup{M}\textup{W}}(0),\rho^\textup{\corr}_{\cold\textup{M}\textup{W}})\bigg{|}_{k=0}\neq 0,
\ee
for all $\rho^\corr_\CMW$ since by definition $\rho^1_\CMW(0,\rho^\textup{\ncorr}_{\cold\textup{M}\textup{W}}(0),\rho^\textup{\corr}_{\cold\textup{M}\textup{W}})\neq \id_{\cold\textup{M}\textup{W}}/N$. We can infer that $\rho^1_\CMW$ is not maximally mixed from a few observations, for example: this is true because we have required that the reduced state on the battery is not maximally mixed since we consider near perfect work extraction. 

Thus, taking into account $W_\textup{ext}(\beta_c-g)=\bo(g)$, Eq. \eqref{eq:chi frak 0 lim} implies Eq. \eqref{eq:necessary k_i/g lim}.
\end{proof}

By now, we have established a constraint on how quickly the correlations have to vanish w.r.t. $g$, for the possibility of achieving Carnot efficiency. In the next Lemma \ref{lem:corres last lemma}, we will show that the constraints given by Eq.~\eqref{eq:necessary k_i/g lim} can be used to derive an upper bound for $\eta_\infty^\qm$.

\begin{lemma}\label{lem:corres last lemma}
If Eqs. \eqref{eq:quasistatic no corrs part} and \eqref{eq:necessary k_i/g lim} are satisfied, then the quantity $\eta_\infty^\qm$ can be upper bounded by
\begin{align}
\eta_\infty^\qm (\rho^1_{\cold\textup{M}\textup{W}}(k,&\rho^\textup{\ncorr}_{\cold\textup{M}\textup{W}}(g),\rho^\textup{\corr}_{\cold\textup{M}\textup{W}}))\\
&\leq \left[1+\frac{\beta_h}{\beta_c-\beta_h}\frac{\gamma(1)}{\gamma(\infty)}\right]^{-1}+\bo(f(g))+\bo(k/g)+\bo(g)+\bo(\varepsilon),\label{eq:F infty lemma corr}
\end{align}
with $\lim_{g\rightarrow 0^+}f(g)= 0$.
\end{lemma}
\begin{proof}
The main idea of our proof is as follows: we show that if Eqns. \eqref{eq:quasistatic no corrs part} and \eqref{eq:necessary k_i/g lim} hold, then we can upper bound $W_\textup{ext}$ while considering only the $F_\infty $ condition. This bound differs from the value given when no correlations are present by only a small amount.
Substituting this into the expression for $\eta_\infty^\qm$, we obtain Eq.~\eqref{eq:F infty lemma corr}.


Let us begin by analyzing the difference in eigenvalues of the states $\rho^1_\CMW$ and $\underline{\rho^1_\CMW}$. 
Recall that 
\begin{align}
\rho^1_\CMW (k,\rho^\ncorr_\CMW,\rho^\corr_\CMW) & = (1-k)\rho^\ncorr_\CMW + k \rho^\corr_\CMW
\end{align}
where $\rho^\ncorr_\CMW,\rho^\corr_\CMW$ are both diagonal in the energy eigenbasis. 
Since $\rho^1_\CMW$ is a mixture of two energy-diagonal states, it is also diagonal. Let us denote its eigenvalues as $[\rho^1_{\cold\textup{M}\textup{W}}]_i$. 

As for $\underline{\rho^1_\CMW }$, Eqn.~\eqref{eq:quasistatic no corrs part} gives the explicit form of the state, 
\begin{equation}
\underline{\rho_\CMW^1} =\rho_\cold^1\otimes\rho_\mach^1\otimes\rho_\batt^1 = \tau(g)\otimes\rho^0_\textup{M}(g)\otimes \rho^1_\textup{W}.
\end{equation}
Let us denote its eigenvalues as $[\underline{\rho^1_\CMW}]_i$. 

We first observe two properties involving trace distance $d(\cdot,\cdot)$:
\begin{enumerate}[label=(P.\roman*)]
\item Consider two states $\sigma_1, \sigma_2$ diagonal in the same eigenbasis. Then if $\rho = (1-k)\sigma_1+k\sigma_2$ for some $k\in [0,1]$, then one can conclude that the distance 
\begin{align}
d(\rho,\sigma_1) \leq k.
\end{align}
\item For any two states $\rho,\sigma$ diagonal in the same basis, with eigenvalues $p_i,q_i$, if their trace distance
\begin{equation}
d(\rho,\sigma)=\frac{1}{2} \| \rho-\sigma \|_1 \leq \varepsilon,
\end{equation}
then this implies that their eigenvalues cannot differ by more than $\varepsilon$, i.e. $\forall i, |p_i-q_i|\leq \varepsilon$. 
By using this fact, we may first calculate the trace distance between $\rho_\CMW^1$ and $\underline{\rho_\CMW^1}$, then bound the difference of their eigenvalues.
\end{enumerate}
We find that
\begin{align}
d(\rho_\CMW^1,\underline{\rho_\CMW^1}) &\leq d(\rho_\CMW^1,\rho_\CMW^\ncorr) +  d(\rho_\CMW^\ncorr,\underline{\rho_\CMW^1})\\
& \leq k + d(\rho_\cold^\ncorr,\rho_\cold^1)+ d(\rho_\mach^\ncorr,\rho_\mach^1)+ d(\rho_\batt^\ncorr,\rho_\batt^1)\\
& \leq 4k.\label{eq:distance_between_correlated_underlined}
\end{align}
The first inequality is a triangle inequality that holds for all states. The second inequality holds because of (P.i), and because trace distance is subadditive under tensor product (note that both $\rho_\CMW^\ncorr$ and $\underline{\rho_\CMW^1}$ are tensor product states). The third inequality holds because we know $d(\rho_\CMW^1,\rho_\CMW^\ncorr) \leq k$ and that trace distance decreases under partial trace. By (P.ii), Eq.~\eqref{eq:distance_between_correlated_underlined} tells us that $\forall i$, 
\begin{equation}\label{eq:eigenvaluediff}
[\rho^1_\CMW]_i=[\underline{\rho^1_\CMW}]_i+\soo(k).
\end{equation}

With Eq.\eqref{eq:eigenvaluediff}, we may relate the $F_\infty$ quantities for the states $\rho_\CMW^1$ and $\underline{\rho_\CMW^1}$. From Eq. \eqref{eq:generalfreeenergy}, we have
\begin{align} 
F_\infty \big(\rho^1_\CMW(k,\rho^\textup{\ncorr}_{\cold\textup{M}\textup{W}}(g)&,\rho^\textup{\corr}_{\cold\textup{M}\textup{W}}),\tau^h_{\cold\textup{M}\textup{W}}\big) \\
&=\ln\max_i\left\{\frac{[\rho^1_{\cold\textup{M}\textup{W}}]_i}{\tau_i}\right\},\\
&=\ln\max_i\left\{\frac{[\underline{\rho^1_\CMW}]_i}{\tau_i}\right\}+\soo (k),\\
&=F_\infty\left(\tau(g)\otimes\rho^0_\textup{M}(g)\otimes \rho^1_\textup{W},\tau^h_{\cold\textup{M}\textup{W}}\right)+\soo (k),
\end{align}

where we used Eq. \eqref{eq:generalfreeenergy} in the last line.

The next step is to evaluate the restriction on $W_\textup{ext}$ that satisfies
\begin{align}\label{eq:F leq F in corr}
F_\infty(\tau_\cold^0\otimes \rho^0_\textup{M}\otimes\rho^0_\textup{W},\tau^h_{\cold\textup{M}\textup{W}})
&\geq  F_\infty\left(\rho^1_{\cold\textup{M}\textup{W}}(k,\rho^\textup{\ncorr}_{\cold\textup{M}\textup{W}}(g),\rho^\textup{\corr}_{\cold\textup{M}\textup{W}}),\tau^h_{\cold\textup{M}\textup{W}}\right)\\
&=F_\infty\left(\tau(g)\otimes\rho^0_\textup{M}(g)\otimes \rho^1_\textup{W},\tau^h_{\cold\textup{M}\textup{W}}\right)+\soo (k),\label{eq:337}
\end{align}
for $W_\textup{ext}$ up to order $\soo  (k)$. Taking into account the additivity of $F_\infty$ under tensor product, we can rearrange Eq.~\eqref{eq:337} to provide an upper bound on $W_\textup{ext}$,
\be\label{eq:W ext with corr}
W_\textup{ext}\leq \frac{ng}{\beta_h}\left[\gamma(\infty)+\bo(f(g))+\soo (k/g)\right], 
\ee
where $\displaystyle\lim_{g\rightarrow 0^+}f(g)=0$, $\gamma(\infty)$ is given by Eq. \eqref{eq:gammainf}. The bound in Eq. \eqref{eq:W ext with corr} is achievable since the $F_\infty$ conditions imposed by Eq. \eqref{eq:F leq F in corr} are achievable with equality. 

Lastly, by using the expression for efficiency in Eqs. \eqref{eq:eta gain for corr}, and substituting $W_\textup{ext}$ from Eq.~\eqref{eq:W ext with corr} (with equality for the maximum possible $W_\textup{ext}$) followed by $\Delta C$ from Eq.~\eqref{delta C small g}, we have
\begin{align} 
\sup_{W_\textup{ext}>0}\eta (\rho^1_\cold,W_\textup{ext})&=\sup_{W_\textup{ext}>0} \left(1-\varepsilon+\frac{\Delta C}{W_\textup{ext}}\right)^{-1}\\
&=\left[1+\frac{\beta_h}{(\beta_c-\beta_h)}\frac{\gamma(1)}{\gamma(\infty)}\right]^{-1}+\bo(f(g))+\soo (k/g)+\bo(g)+\bo(\varepsilon).\label{eq:eff inf corr}
\end{align}
Hence using Eqs. \eqref{eq:intermediate efficiencies}, \eqref{eq:eff inf corr} we find Eq. \eqref{eq:F infty lemma corr}. Note that in Eq \eqref{eq:F infty lemma corr} we have an inequality, this is due to the fact that in the optimisation problem Eq. \eqref{eq:intermediate efficiencies}, there is an additional constraint (namely mean energy conservation) which is not taken into account in the derivation of Eq. \eqref{eq:eff inf corr}.
\end{proof}
Finally, the above lemmas allow us to conclude that allowing further correlations in the final state cannot allow quasi-static heat engines to achieve the Carnot efficiency when $\Omega >1$.
\begin{theorem}\label{corrs dont help}[Correlations do not improve efficiency] Suppose that $\Omega>1$. Parametrizing the final state of the heat engine by Eq.~\eqref{eq:rho 1 corr def}-\eqref{eq:preservemachine}, the quantum efficiency $\eta^\qm$ defined in Eq.~\eqref{eq:quant eff def} in a quasi-static heat engine is strictly upper bounded by the Carnot efficiency,
\be
\sup_{k\in[0,1],~\rho^\ncorr_\CMW}\eta^\qm\left(\rho^1_\CMW(k,\rho^\ncorr_\CMW,\,\rho^\corr_\CMW)\right) < 1-\frac{\beta_h}{\beta_c}.
\ee
\end{theorem}
\begin{proof}
From Lemma \ref{lem:alpha up q eff bound}, we have that both $\eta^\qm\leq \eta^\qm_1$ and $\eta^\qm\leq \eta^\qm_\infty$ hold. Thus a necessary condition to achieve the Carnot efficiency for a particular $\rho^1_\CMW$, is that both $\eta^\qm_1$ and $\eta^\qm_\infty$ are equal to or greater than the Carnot efficiency. 

Lemma \ref{Lemm:vanishing coors necesary condition} proves that Eqs. \eqref{eq:quasistatic no corrs part} and \eqref{eq:necessary k_i/g lim} are necessary conditions for $\eta^\qm_1$ to achieve the Carnot efficiency. However, when  Eqs. \eqref{eq:quasistatic no corrs part}, \eqref{eq:necessary k_i/g lim} are satisfied, then Lemma \ref{lem:corres last lemma} provides an upper bound on the efficiency $\eta^\qm_\infty$ in Eq. \eqref{eq:F infty lemma corr}. 

Now, suppose $\Omega>1$. Since it is shown in Eq. \eqref{gam 1 to gam inf <1} that $\gamma(1)/\gamma(\infty)=\Omega$, plugging this into the leading term appearing in Eq. \eqref{eq:F infty lemma corr}
\be 
\left[1+\frac{\beta_h}{(\beta_c-\beta_h)}\frac{\gamma(1)}{\gamma(\infty)}\right]^{-1},
\ee
we have that the quantity $\eta_\infty^\qm$ (and therefore also $\eta^\qm$) is strictly less than the Carnot efficiency $1-\beta_h/\beta_c$.
\end{proof}
\subsection{A more general final battery state}\label{A more general final battery state}
For the simplicity of our analysis, we have assumed that the battery is left in the specific final state described in Eq.~\eqref{eq:battery initial state}, i.e. an amount of work $W_\textup{ext} = E_k - E_j$ is extracted, except with failure probability $\varepsilon$ that the battery remains in the initial state $\ketbra{E_j}{E_j}_\batt$. In this section, we show that this is a simplification which can be removed in general, i.e. the final battery state is allowed to be  any state within the $\varepsilon$-ball of $\ketbra{E_k}{E_k}_\batt$. In particular, our result that the Carnot efficiency cannot be achieved when $\Omega>1$ still holds. 

In Lemma \ref{lem:general_batt}, we show that for any final state of the cold bath $\rho_\cold^1$, allowing a more general final battery state does not affect the amount of work bounded by the $F_\infty$ condition. We then use this to prove in Theorem \ref{thm:general_batt} that when $\Omega >1$, the Carnot efficiency cannot be achieved even if we allow a more general battery final state.

\begin{lemma}\label{lem:general_batt} 
For any given $\rho_\cold^0,\rho_\cold^1$, with $\rho_\batt^0 = \ketbra{E_j}{E_j}_\batt$, consider the maximum $W_{\infty}^1 := E_{k_1}-E_j$ such that $\rho_\cold^0\otimes\rho_\batt^0\rightarrow\rho_\cold^1\otimes\rho_\batt^1$ is allowed by the non-increasing $F_\infty$ condition (Eq. \eqref{2nd law eq}) i.e.
\begin{equation}\label{eq:dinf_batt1}
D_\infty (\rho_\cold^0\|\tau_\cold^{\beta_h}) + D_\infty(\rho_\batt^0\|\tau_\batt^{\beta_h}) \geq D_\infty (\rho_\cold^1\|\tau_\cold^{\beta_h}) + D_\infty(\rho_\batt^1\|\tau_\batt^{\beta_h}),
\end{equation}
with 
\begin{equation}\label{eq:batt1}
\rho_\batt^1 =(1-\varepsilon)\ketbra{E_{k_1}}{E_{k_1}}_\battery+\varepsilon \ketbra{E_j}{E_j}_\battery.
\end{equation}

On the other hand, consider any battery final state
\begin{equation}\label{eq:batt2}
\rho_\batt^2 =(1-\varepsilon)\ketbra{E_{k_2}}{E_{k_2}}_\battery+\varepsilon\rho_\battery^\textup{junk},
\end{equation}
where $\rho_\battery^\textup{junk}$ is an energy-diagonal state orthogonal to $\ketbra{E_{k_2}}{E_{k_2}}_\batt$ which may depend on $\varepsilon $, i.e. $\rho_\battery^\textup{junk} = \sum_{i} p_i \ketbra{E_i}{E_i}_\batt$ with $p_{k_2}=0$ and $\sum_i p_i = 1$. 
Define $W_{\infty}^2 := E_{k_2}-E_j$ such that $\rho_\cold^0\otimes\rho_\batt^0\rightarrow\rho_\cold^1\otimes\rho_\batt^2$ is allowed by the non-increasing $F_\infty$ condition, i.e.
\begin{equation}\label{eq:dinf_batt2}
D_\infty (\rho_\cold^0\|\tau_\cold^{\beta_h}) + D_\infty(\rho_\batt^0\|\tau_\batt^{\beta_h}) \geq D_\infty (\rho_\cold^1\|\tau_\cold^{\beta_h}) + D_\infty(\rho_\batt^2\|\tau_\batt^{\beta_h}).
\end{equation}
Then for all $0<\varepsilon\leq\hat\varepsilon =\left[ 1+e^{\beta_h (E_{\rm max}-E_{j})} \right]^{-1}$, we have  $W_{\infty}^1=W_{\infty}^2$. 

\end{lemma}
\begin{proof}
Firstly, note that any energy-diagonal state $\rho_\batt^2$ with trace distance $d(\rho_\batt^2,\ketbra{E_{k_2}}{E_{k_2}}_\batt) = \varepsilon$ can be written in the form of Eq.~\eqref{eq:batt2}. Rearranging the terms in Eq.~\eqref{eq:dinf_batt1},
\begin{align}\label{eq:XXX}
D_\infty(\rho_\batt^1\|\tau_\batt^{\beta_h}) &\leq D_\infty(\rho_\batt^0\|\tau_\batt^{\beta_h}) + D_\infty (\rho_\cold^0\|\tau_\cold^{\beta_h}) -D_\infty (\rho_\cold^1\|\tau_\cold^{\beta_h}) =: A.
\end{align}
One can use the definition of $D_\infty$ in Eq.~\eqref{eq:D_inft def} to expand the L.H.S. of Eq.~\eqref{eq:XXX}, obtaining
\begin{align}
&\log \max \lbrace(1-\varepsilon)e^{\beta_h E_{k_1}},\varepsilon e^{\beta_h E_j}\rbrace \leq A- \log Z_\batt^{\beta_h}.\label{eq:dinf_condition_rho1}
\end{align}
We know that since near perfect work is extracted, $\varepsilon$ is arbitrarily small. This implies that for $\varepsilon$ small enough, $\displaystyle\max \lbrace(1-\varepsilon)e^{\beta_h E_{k_1}},\varepsilon e^{\beta_h E_j}\rbrace = (1-\varepsilon)e^{\beta_h E_{k_1}}$.

Similarly, one can evaluate Eq.~\eqref{eq:dinf_batt1} to obtain
\begin{align}\label{eq:dinf_rho2batt}
\log \max \lbrace(1-\varepsilon)e^{\beta_h E_{k_2}},
\lbrace\varepsilon p_i e^{\beta_h E_i}\rbrace_{i\neq k_2}\rbrace &\leq A- \log Z_\batt^{\beta_h}.
\end{align}
Note that the maximization in Eq.~\eqref{eq:dinf_rho2batt} only picks out the maximum value. In particular, denoting $E_{\rm max}$ to be the largest energy eigenvalue of the battery, then whenever 
\begin{equation}\label{eq:1-ep general battery inequality}
(1-\varepsilon) e^{\beta_h E_{k_2}} \geq \varepsilon e^{\beta_h E_{\rm max}},
\end{equation}
or equivalently
\begin{equation}\label{eq:340}
\varepsilon\leq \left[ 1+e^{\beta_h (E_{\rm max}-E_{k_2})} \right]^{-1},
\end{equation}
then $\displaystyle\max \lbrace(1-\varepsilon)e^{\beta_h E_{k_2}},\lbrace\varepsilon p_i e^{\beta_h E_i}\rbrace_{i\neq k_2}\rbrace = (1-\varepsilon)e^{\beta_h E_{k_2}}$. In other words, as long as $\varepsilon$ is upper bounded by Eq.~\eqref{eq:340}, we know which terms attains the maximization in Eq.~\eqref{eq:dinf_condition_rho1}.
However, we also want an upper bound that is independent of any limit involving the final state $\rho^1_{\cold\textup{M}\battery}$ we wish to take, or any amount of work extracted (and therefore, we want the bound to be independent of $E_{k_2}$). As such, let us construct the following upper bound $\varepsilon \leq \hat\varepsilon$ where,
\be 
\hat \varepsilon := \inf_{\substack{E_{k_2}\\ W^2_\infty>0}} \left[ 1+e^{\beta_h (E_{\rm max}-E_{k_2})} \right]^{-1} = \left[ 1+e^{\beta_h (E_{\rm max}-E_{j})} \right]^{-1}
\ee 
Now, we see that $E_{k_1}$ and $E_{k_2}$ correspond to the solutions for Eq.~\eqref{eq:dinf_condition_rho1} and Eq.~\eqref{eq:dinf_rho2batt}, which for $\varepsilon\leq \hat\varepsilon$ reduce to exactly the same equation. Therefore, $E_{k_1}=E_{k_2}$ and hence $W_{\infty}^1 = W_{\infty}^2$.
\end{proof}

We will use Lemma \ref{lem:general_batt} to prove Theorem \ref{thm:general_batt}. But before we proceed, let us fix some notation: we define the efficiency as a function of $\alpha\geq 0$  :
\begin{align}
\eta_\alpha^J(\rho_\cold^1)
=\sup_{E_{k_J}-E_j>0} \eta(\rho_\cold^1) \quad \textup{ subject to }&\quad F_\alpha(\rho^0_\battery\otimes\tau_\cold^0,\tau_\CW^h)\geq F_\alpha(\rho^J_\battery\otimes\rho_\stcold^1,\tau_\CW^h),\\
\textup{ and }&\quad
 \tr(\hat H_t \rho_\sttotal^0)=\tr(\hat H_t \rho_\sttotal^{1,J}).
\end{align}
with $J=1,2$ denoting the final battery state $\rho^J_\textup{W}$. We also define an $\alpha$ independent efficiency:
\begin{align}
\eta^J(\rho_\cold^1)
=\sup_{E_{k_J}-E_j>0} \eta(\rho_\cold^1) \quad \textup{ subject to }&\\
 F_\alpha(\rho^0_\battery\otimes&\tau_\cold^0,\tau_\CW^h)\geq F_\alpha(\rho^J_\battery\otimes\rho_\stcold^1,\tau_\CW^h)\,\forall \alpha\geq 0.
\end{align}
For any $\alpha\geq 0$, and any state $\rho_\cold^1$, $\eta_\alpha^J (\rho_\cold^1) \geq \eta^J (\rho_\cold^1)$ holds.

We already know that when $\Omega >1$, for any final cold bath state $\rho_\cold^1$, the efficiency $\eta^1(\rho_\cold^1)$ is strictly less than the Carnot efficiency. Theorem \ref{thm:general_batt} shows that this is also true for $\eta^2(\rho_\cold^1)$, i.e. when allowing a more general battery final state.
\begin{theorem}\label{thm:general_batt}
[General battery states do not improve efficiency] Consider a quasi-static heat engine with a cold bath consisting of $n$ qubits, extracting near perfect work. Let $\Omega >1$ (definition in Eq.~\eqref{def:sigma}). Then 
the efficiency 
$\lim_{g\rightarrow 0^+}\eta^2 (\tau_{\beta_f})$, where $\tau_{\beta_f}$ is the final state of a quasi-static heat engine (Def. \ref{def:quasi static}), is strictly less than the Carnot efficiency.
\end{theorem}

\begin{proof}
Firstly, suppose that $\Omega >1$. By Lemma \ref{lemma: efficiency as a function of kappa} we know that the infimum is obtained at $\alpha=\infty$, and by Lemma \ref{Quantum/Nano heat engine efficiency} we know that the efficiency for quasi-static heat engine is strictly less than the  Carnot value:
\begin{equation}\label{eq:eta1=etamax1<etac}
\lim_{g\rightarrow 0^+}\eta^1 (\tau_{\beta_f}) =\lim_{g\rightarrow 0^+}\eta^1_\infty (\tau_{\beta_f}) < \eta_C.
\end{equation}
In other words, for all other final states $\rho_\cold^1$ we know that Carnot efficiency cannot be achieved. Therefore, it suffices to see that in the quasi-static limit,
\begin{equation}\label{eq:finally!}
\lim_{g\rightarrow 0^+}\eta^2 (\tau_{\beta_f}) \leq \lim_{g\rightarrow 0^+}\eta^2_\infty (\tau_{\beta_f}) = \lim_{g\rightarrow 0^+}\eta^1_\infty (\tau_{\beta_f}) = \lim_{g\rightarrow 0^+}\eta^1 (\tau_{\beta_f})< \eta_C.
\end{equation}
The second equality is obtained by noting that for any state $\tilde\rho_\cold^1$ (and therefore for $\tau_{\beta_f}$): 
\begin{enumerate}
\item $\Delta C$ is the same for both expressions of efficiency $\eta^1_\infty(\rho_\cold^1)$ and $\eta^2_\infty(\rho_\cold^1)$.
\item By Lemma \ref{lem:general_batt}, for all $0<\varepsilon<\left[ 1+e^{\beta_h (E_{\rm max}-E_{j})} \right]^{-1}$,  $W_\infty^1 (\tilde\rho_\cold^1)= W_\infty^2 (\tilde\rho_\cold^1)$.
\end{enumerate} 
Hence, from Items 1 and 2, one concludes that $\eta^1_{\infty} (\tilde\rho_\cold^1) = \eta^2_{\infty} (\tilde\rho_\cold^1)$. 
The third equality in Eq.~\eqref{eq:finally!} comes directly from Eq.~\eqref{eq:eta1=etamax1<etac}.
\end{proof}

\end{document}